\documentclass[11pt, reqno]{amsart}

\usepackage{amsmath,amssymb,amsthm,mathrsfs,mathtools}
\usepackage{microtype}
\usepackage{hyperref}
\hypersetup{
  colorlinks=true,
  linkcolor=blue,
  citecolor=blue,
  urlcolor=blue,
  pdftitle={Quantitative local recovery of Kerr--de~Sitter parameters from high-frequency equatorial quasinormal modes},
  pdfauthor={Ruiliang Li}
}

\numberwithin{equation}{section}

\newtheorem{theorem}{Theorem}
\newtheorem{proposition}[theorem]{Proposition}
\newtheorem{lemma}[theorem]{Lemma}
\newtheorem{corollary}[theorem]{Corollary}
\theoremstyle{definition}
\newtheorem{definition}[theorem]{Definition}
\theoremstyle{remark}
\newtheorem{remark}[theorem]{Remark}

\providecommand{\doi}[1]{\href{https://doi.org/#1}{doi:#1}}
\providecommand{\arxiv}[1]{\href{https://arxiv.org/abs/#1}{arXiv:#1}}

\title[Inverse Kerr--de~Sitter spectroscopy]{Quantitative local recovery of Kerr--de~Sitter parameters from high-frequency equatorial quasinormal modes}

\author{Ruiliang Li}
\address{Tsinghua University, Beijing 100084, China}
\email{lirl23@mails.tsinghua.edu.cn}

\date{\today}


\begin{document}

\begin{abstract}
We study an inverse resonance problem for the scalar wave equation on the Kerr--de~Sitter family.  
In a compact subextremal slow-rotation regime and at a fixed overtone index, high-frequency quasinormal modes admit semiclassical quantization and a real-analytic labeling by angular momentum indices.  
Using this structure, we first prove that a finite equatorial high-frequency package of quasinormal-mode frequencies determines the mass and rotation parameter $(M,a)$ (for fixed cosmological constant $\Lambda>0$), with a quantitative stability estimate.
As a key geometric input we compute explicit second-order (in $a$) corrections to the equatorial photon-orbit invariants which control the leading real and imaginary parts of the quasinormal modes.
Finally, allowing $\Lambda$ to vary in a compact interval, we show that adding one damping observable (the scaled imaginary part of a single equatorial mode) yields a three-parameter inverse theorem: a finite package of three independent real observables determines $(M,a,\Lambda)$ locally in the slow-rotation regime away from $a=0$.
\end{abstract}

\maketitle

\section{Introduction}\label{sec:intro}

The observation of gravitational waves has turned the ``ringdown'' phase of a perturbed black hole into a quantitative probe of strong-field general relativity. In the simplest idealization, the late-time signal is well-approximated by a finite superposition of exponentially damped sinusoids whose complex frequencies are \emph{quasinormal modes} (QNMs); see, e.g., \cite{DreyerEtAl2004BlackHoleSpectroscopy,BertiCardosoStarinets2009QNMReview}. From the mathematical viewpoint relevant here, QNMs are poles of a meromorphic continuation of a cutoff resolvent (or equivalently of a stationary family of operators obtained by Fourier transformation in time). The modern microlocal framework that gives a robust Fredholm setup for Kerr--de~Sitter and related geometries originates in Vasy's work \cite{Vasy2013MicrolocalAHKdS} and was further developed in a series of papers establishing existence, distribution, and resonance expansions in the slowly rotating Kerr--de~Sitter setting \cite{Dyatlov2011QNMKerrDeSitter,Dyatlov2012AsymptoticQNMKdS}; see also the surveys \cite{Zworski2012SemiclassicalAnalysis,Zworski2017ScatteringResonances} for background on resonances in scattering theory. Recent advances establish real-analyticity of Kerr/Kerr--de~Sitter resonant states (quasinormal modes as functions on the spacetime) under real-analytic coefficients \cite{PetersenVasy2023AnalyticityQNM} and remove slow-rotation restrictions for resonance expansions of linear waves in the full subextremal Kerr--de~Sitter range \cite{PetersenVasy2025WaveKdS}; in parallel, resolvent technology near normally hyperbolic trapping has been sharpened and systematized in \cite{WunschZworski2011ResolventNHT,HintzVasy2015NonTrappingNHT,Dyatlov2016SpectralGapsNHT}.

\medskip

\noindent\textbf{Inverse resonance problems and black-hole parameters.}
Motivated by black-hole spectroscopy, one may ask to what extent geometric parameters can be recovered from (partial) spectral data. In the non-rotating de~Sitter--Schwarzschild case, Uhlmann--Wang \cite{UhlmannWang2023RecoverMassSingleQNM} proved that the black hole mass can be recovered locally from a \emph{single} QNM, with H\"older stability, thereby providing a rigorous counterpart to the heuristic ``one complex frequency determines one real parameter'' principle. For rotating black holes, the parameter space becomes two-dimensional (mass and angular momentum), and the question becomes both subtler and richer: one expects that a judiciously chosen pair of observables extracted from the QNM spectrum should determine the parameters locally, but quantitative stability depends on the microlocal structure of trapping and on the non-selfadjoint nature of the underlying spectral problem.

In the same de~Sitter--Schwarzschild setting, the high-frequency distribution of resonances and the associated pseudopole lattice were described earlier by S\'a Barreto--Zworski \cite{SaBarretoZworski1997SphericalBH}, providing a rigorous semiclassical underpinning for the ``lattice'' patterns observed in numerical studies.  Dyatlov's subsequent work \cite{Dyatlov2012AsymptoticQNMKdS} extends this semiclassical quantization paradigm to slowly rotating Kerr--de~Sitter spacetimes and supplies the forward microlocal input for our inverse analysis.

\medskip

\noindent\textbf{Setting and the equatorial high-frequency regime.}
Fix a cosmological constant $\Lambda>0$. For each subextremal Kerr--de~Sitter metric $g_{M,a}$ with mass $M>0$ and (small) rotation parameter $a$, consider the scalar wave equation
\begin{equation}\label{eq:intro-wave}
\Box_{g_{M,a}}u=0.
\end{equation}
Upon separation of variables, QNMs are labeled by an overtone number $n\in\mathbb{N}_0$, an angular momentum $\ell\in\mathbb{N}$, and an azimuthal index $k\in\mathbb{Z}$ with $|k|\le\ell$. In this paper we focus on fixed $n$ and the semiclassical (high-frequency) limit $\ell\to\infty$ in the \emph{equatorial} families $k=\pm\ell$; throughout, we use ``co-/counter-rotating'' to mean $k>0$/$k<0$ (equivalently, the sign of the corresponding geodesic angular velocity in the chosen Boyer--Lindquist $\varphi$-coordinate). Denote the corresponding QNMs by
\[
\omega_\pm(M,a):=\omega_{n,\ell,\pm\ell}(M,a)\in\mathbb{C}.
\]
The semiclassical mechanism underlying these families is barrier-top quantization at the photon sphere: to leading order, $\Re\omega_\pm$ captures an orbital frequency, while $-\Im\omega_\pm$ is governed by a Lyapunov exponent of the unstable null geodesic. The precise quantization we use is recalled and proved in Sections~\ref{sec:setup}--\ref{sec:HF-quantization}.

\medskip

\noindent\textbf{A two-mode package using only oscillation frequencies.}
A key point for inverse problems is that $\Re\omega_\pm$ can (in principle) be estimated more robustly from time-domain data than $\Im\omega_\pm$, and it is therefore natural to ask whether one can recover $(M,a)$ from \emph{real parts alone}. We show that, in the slowly rotating regime and for $\ell$ large, the pair of observables
\begin{equation}\label{eq:intro-U-V}
U:=\Re\frac{\omega_+(M,a)+\omega_-(M,a)}{2\ell},\qquad
V:=\Re\frac{\omega_+(M,a)-\omega_-(M,a)}{2\ell}
\end{equation}
determines $(M,a)$ locally and stably on precompact subextremal parameter sets. The normalization is chosen so that $U$ is of order $1$ and primarily sensitive to $M$, while $V$ is of order $a$ and measures the Zeeman-type splitting of the two equatorial families.

Our principal inverse result for the \emph{true} QNM spectrum is stated as Theorem~\ref{thm:inverse-true-QNM}: for any precompact set
\[
\mathcal K \Subset \{(M,a):\ \text{$g_{M,a}$ is subextremal}\}\cap\{|a|\le a_0\},
\]
with $a_0>0$ sufficiently small, there exists $\ell_0=\ell_0(n,\mathcal K)$ such that for all $\ell\ge\ell_0$ the data map
\[
\mathcal G^{\mathrm{QNM}}_\ell:\mathcal K\to\mathbb{R}^2,\qquad
\mathcal G^{\mathrm{QNM}}_\ell(M,a):=(U,V),
\]
is a real-analytic diffeomorphism onto its image with a uniform two-sided Lipschitz bound (cf.\ \eqref{eq:stability-true-omega}). In particular, $(M,a)$ can be reconstructed from $(U,V)$ with a stability constant depending only on $\mathcal K$ and $\Lambda$.

\medskip

\noindent\textbf{On labeling, orientation, and recovery from unlabeled data.}
The map $\mathcal G^{\mathrm{QNM}}_\ell$ uses \emph{labeled} equatorial modes: one needs to know which frequency corresponds to $k=+\ell$ (co-rotating) and which corresponds to $k=-\ell$ (counter-rotating), relative to a fixed choice of the axial Killing field $\partial_\varphi$.
This is natural from the PDE/separation viewpoint, but it matters conceptually because the Kerr--de~Sitter metrics with parameters $a$ and $-a$ are isometric after reversing the axial coordinate $\varphi\mapsto-\varphi$ (equivalently, after swapping $k\mapsto-k$).
Accordingly, if one is given only the \emph{unordered} pair of equatorial oscillation frequencies $\{\Re\omega_{n,\ell,\ell},\Re\omega_{n,\ell,-\ell}\}$ (or, more generally, the unlabeled QNM set), then one cannot expect to determine the \emph{sign} of $a$ without an additional orientation convention; at best one can recover $|a|$.
We make this precise by proving in Section~\ref{sec:transfer} an unlabeled variant of the two-parameter inverse theorem: for $\ell$ large, the data $(U,|V|)$ determined by the unordered pair $\{\Re\omega_+,\Re\omega_-\}$ recovers $(M,|a|)$ locally and stably (Corollary~\ref{cor:unlabeled-recovery}).

\medskip

\noindent\textbf{Single-mode recovery in the rotating case (extension of Uhlmann--Wang).}
In addition to the two-mode/real-part reconstruction above, we record a complementary \emph{single-mode} recovery statement, in the spirit of \cite{UhlmannWang2023RecoverMassSingleQNM}, which becomes available once rotation is allowed: the pair $(\Re\omega_+,\Im\omega_+)$ provides two real observables and thus has the correct dimension to determine $(M,a)$. Using the same semiclassical normal form that produces the equatorial quantization, one obtains a one-mode inverse map in the slowly rotating regime, with stability uniform for $\ell\gg 1$. Concretely, in the equatorial family $k=\ell$ we consider the scaled observables
\[
\widetilde U:=\Re\frac{\omega_+(M,a)}{\ell},\qquad
\widetilde W:=-\frac{\Im\,\omega_+(M,a)}{n+\frac12},
\]
and prove that $(M,a)\mapsto(\widetilde U,\widetilde W)$ is locally invertible on $\mathcal K$ for $\ell$ large (a precise statement is given later as a corollary of the pseudopole inversion and the transfer to true QNMs). This can be viewed as a rotating analogue of \cite{UhlmannWang2023RecoverMassSingleQNM}: when $a=0$, $\widetilde U$ and $\widetilde W$ reduce to the orbital frequency and Lyapunov exponent at the de~Sitter--Schwarzschild photon sphere, whereas for $a\neq0$ the first-order Zeeman splitting supplies the missing second parameter.

\medskip

\noindent\textbf{Method and intermediate pseudopole problem.}
As in \cite{Dyatlov2011QNMKerrDeSitter,Dyatlov2012AsymptoticQNMKdS}, the forward QNM problem is non-selfadjoint and is naturally approached via semiclassical microlocal analysis of the separated radial operator. We first construct \emph{pseudopoles}---explicit semiclassical approximations to QNMs obtained by quantization of a barrier-top normal form---and show that they lie within $O(\ell^{-N})$ of true QNMs for every $N$ (see \eqref{eq:superpoly-close}). The inverse theorem is proved first at the pseudopole level (Theorem~\ref{thm:inverse-pseudopoles}), where one can explicitly track the dependence of the leading terms on $(M,a)$ and establish a uniform Jacobian lower bound. An explicit leading reconstruction is given in Proposition~\ref{prop:explicit-inverse}, and the corresponding statement for true QNMs is Proposition~\ref{prop:explicit-true}.

\medskip

\noindent\textbf{Structure of the paper.}
Section~\ref{sec:setup} recalls the Kerr--de~Sitter geometry, separation of variables, and the Fredholm/meromorphic setup for QNMs following \cite{Vasy2013MicrolocalAHKdS,Dyatlov2011QNMKerrDeSitter}. In Section~\ref{sec:HF-quantization} we develop the high-frequency quantization of the equatorial modes and introduce the pseudopole set. Section~\ref{sec:nondegeneracy} isolates the geometric nondegeneracy mechanism of the equatorial principal frequency map on compact slow-rotation parameter sets. Section~\ref{sec:inverse-pseudopoles} proves the inverse theorem for pseudopoles and derives closed-form leading reconstructions. Section~\ref{sec:transfer} transfers the inverse result to true QNMs using sharp pseudopole--QNM proximity estimates. Finally, Section~\ref{sec:discussion} discusses extensions and open problems, including the possibility of relaxing the slow-rotation restriction in light of \cite{PetersenVasy2025WaveKdS} and of incorporating resolvent technology near normally hyperbolic trapping \cite{WunschZworski2011ResolventNHT,HintzVasy2015NonTrappingNHT,Dyatlov2016SpectralGapsNHT}.

\section{Geometric and analytic setup}\label{sec:setup}

Throughout we fix a cosmological constant $\Lambda>0$ and consider the $3+1$ dimensional
Kerr--de~Sitter family $g_{M,a}$ with mass $M>0$ and rotation parameter $a\in\mathbb{R}$,
in the subextremal regime described below.
Our inverse statements will be uniform for $(M,a)$ ranging in a fixed compact subset of this
subextremal parameter region; in particular, we shall keep track of how the geometric objects
and the stationary wave operator depend on $(M,a)$.
The forward resonance theory that underlies our analysis is by now well-established:
in the slowly rotating regime by Dyatlov \cite{Dyatlov2011QNMKerrDeSitter,Dyatlov2012AsymptoticQNMKdS},
in a general microlocal Fredholm framework by Vasy \cite{Vasy2013MicrolocalAHKdS},
and in the full subextremal Kerr--de~Sitter range by Petersen--Vasy \cite{PetersenVasy2025WaveKdS}.
We use these results as input and focus on extracting inverse information from the high-frequency spectrum.

\subsection{The Kerr--de~Sitter exterior and the subextremal parameter region}\label{subsec:kds}

\paragraph{Boyer--\allowbreak Lindquist form.}
Let $(t,r,\theta,\varphi)$ be Boyer--\allowbreak Lindquist coordinates on
$\mathbb{R}_t\times (0,\infty)_r\times (0,\pi)_\theta\times \mathbb{S}^1_\varphi$ and set
\begin{align}
\rho^2(r,\theta) &:= r^2+a^2\cos^2\theta, \label{eq:rho2}\\
\Delta_\theta(\theta) &:= 1+\frac{\Lambda a^2}{3}\cos^2\theta, \label{eq:Delta_theta}\\
\Xi &:= 1+\frac{\Lambda a^2}{3}, \label{eq:Xi}\\
\Delta_r(r) &:= (r^2+a^2)\Bigl(1-\frac{\Lambda r^2}{3}\Bigr)-2Mr. \label{eq:Delta_r}
\end{align}
With the sign convention $(-,+,+,+)$, the Kerr--de~Sitter metric is
\begin{equation}\label{eq:kds_metric_BL}
\begin{aligned}
g=g_{M,a} :={}&
-\frac{\Delta_r}{\rho^2}\Bigl(dt-\frac{a\sin^2\theta}{\Xi}\,d\varphi\Bigr)^2
+\frac{\rho^2}{\Delta_r}\,dr^2+\frac{\rho^2}{\Delta_\theta}\,d\theta^2\\
&\quad
+\frac{\Delta_\theta\sin^2\theta}{\rho^2}\Bigl(a\,dt-\frac{r^2+a^2}{\Xi}\,d\varphi\Bigr)^2.
\end{aligned}
\end{equation}
The vector fields $\partial_t$ and $\partial_\varphi$ are Killing.

\paragraph{Subextremality and horizons.}
We adopt the standard Kerr--de~Sitter notion of subextremality used in the resonance theory
in the full parameter range \cite{PetersenVasy2025WaveKdS}.
We say that $(M,a)$ is \emph{subextremal} (for the fixed $\Lambda$) if the quartic polynomial
$\Delta_r(\cdot;M,a)$ has four \emph{distinct real} simple roots
\begin{equation}\label{eq:horizon_roots}
\begin{aligned}
&r_0(M,a)<0<r_-(M,a)<r_e(M,a)<r_c(M,a),\\
&\Delta_r(r_j)=0,\qquad \Delta_r'(r_j)\neq 0\ \text{for } r_j\in\{r_0,r_-,r_e,r_c\}.
\end{aligned}
\end{equation}
Here $r=r_-$ is the Cauchy horizon, $r=r_e$ the (outer) event horizon, and $r=r_c$ the cosmological horizon.
Our analysis takes place exclusively in the \emph{domain of outer communications} between $r_e$ and $r_c$.
The domain of outer communications is
\begin{equation}\label{eq:DOC}
\mathcal M^\circ_{M,a}:=\mathbb{R}_t\times (r_e,r_c)_r\times \mathbb{S}^2_{\theta,\varphi}.
\end{equation}
By the real-analytic implicit function theorem, the roots in \eqref{eq:horizon_roots} depend
real-analytically on $(M,a)$ in the subextremal set; in particular, $r_-(M,a),\allowbreak r_e(M,a),\allowbreak r_c(M,a)$ are
real-analytic.

\begin{remark}[Why we keep the full root structure]\label{rem:root-structure}
For the inverse problems studied below, only the outer horizons $r_e$ and $r_c$ enter.
Nevertheless we use the full definition \eqref{eq:horizon_roots} to align with the modern resonance literature
and to avoid any ambiguity about which horizon is meant by $r_e$.
All constructions below (star coordinates, Fredholm spaces, separation of variables) are performed on
$\mathcal M^\circ_{M,a}$ and its future extension across the \emph{outer} horizons.
\end{remark}

For later use, we record the (unsigned) surface gravities
\begin{equation}\label{eq:surface_gravity}
\kappa_h(M,a) := \frac{|\Delta_r'(r_h)|}{2\,(r_h^2+a^2)},\qquad r_h\in\{r_e,r_c\},
\end{equation}
which likewise vary real-analytically in the subextremal regime.

\paragraph{A fixed reference manifold for varying parameters.}
Let $\mathcal P_\Lambda\subset(0,\infty)\times\mathbb R$ denote the set of \emph{subextremal} Kerr--de~Sitter parameters
for the fixed $\Lambda$, i.e.\ those $(M,a)$ for which $\Delta_r(\cdot;M,a)$ has four distinct real simple roots as in
\eqref{eq:horizon_roots}.
Throughout we fix a compact set
\begin{equation}\label{eq:Kcompact}
\mathcal K\Subset \mathcal P_\Lambda\cap\{|a|\le a_0\},
\end{equation}
where $a_0>0$ is chosen sufficiently small so that
\begin{enumerate}
\item[(i)] the semiclassical quantization and pseudopole construction of Dyatlov
\cite{Dyatlov2011QNMKerrDeSitter,Dyatlov2012AsymptoticQNMKdS} (together with the parameter-differentiable remainder bounds
in Appendix~\ref{app:parameter-quantization}) apply uniformly on $\mathcal K$; and
\item[(ii)] the slow-rotation Taylor remainders in the geometric nondegeneracy mechanism (Proposition~\ref{prop:geo-nondegeneracy}) can be absorbed uniformly on $\mathcal K$.
\end{enumerate}
Condition (ii) is purely geometric: by Proposition~\ref{prop:geo-nondegeneracy} we have, uniformly on $\mathcal K$,
\begin{equation*}
\det D\mathcal G_{\mathrm{geo}}(M,a)=\Omega_{\mathrm{ph}}'(M)c_Z(M)+\mathcal O(a^2).
\end{equation*}
Hence shrinking $a_0$ if necessary yields a uniform lower bound.

\begin{remark}[Varying the cosmological constant]\label{rem:varying-Lambda}
For notational clarity we fix $\Lambda>0$ throughout most of the paper and view $(M,a)$ as the parameter vector.  
All constructions below (including the geometric quantities attached to photon orbits, the semiclassical quantization, and the analytic labeling of high-frequency modes) depend real-analytically on $\Lambda$ as well.  
In Section~\ref{sec:three-parameter} we therefore enlarge the parameter vector to $(M,a,\Lambda)$ and prove a three-parameter inverse theorem by adding a single damping observable.
\end{remark}

By compactness and the strict ordering in \eqref{eq:horizon_roots}, the outer-horizon separation
\begin{equation}\label{eq:Ldef}
L(M,a):=r_c(M,a)-r_e(M,a)
\end{equation}
is bounded above and below by positive constants on $\mathcal K$.
Consequently there exists $\delta_0\in(0,1)$ such that for all $(M,a)\in\mathcal K$,
\begin{equation}\label{eq:uniform_delta}
r_e-\delta_0\,L>r_-+\delta_0\,L>0,\qquad
\Delta_r(r)\neq 0\ \text{for}\ r\in [r_e-\delta_0 L,\,r_e)\cup (r_c,\,r_c+\delta_0 L].
\end{equation}
Set
\begin{equation}\label{eq:IMa-def}
I_{M,a}:=(r_e(M,a)-\delta_0\,L(M,a),\ r_c(M,a)+\delta_0\,L(M,a))
\end{equation}
and fix the reference interval $I:=(-\delta_0,\,1+\delta_0)$.
For each $(M,a)\in\mathcal K$, define the affine diffeomorphism
\begin{equation}\label{eq:affine_identification}
\Phi_{M,a}: I\to I_{M,a},\qquad
\Phi_{M,a}(x):=r_e(M,a)+x\,(r_c(M,a)-r_e(M,a)).
\end{equation}
We thus identify $I_{M,a}\times\mathbb{S}^2$ with the fixed manifold
\begin{equation}\label{eq:fixed_X}
X:= I\times\mathbb{S}^2_{\theta,\varphi}
\end{equation}
via $(x,\theta,\varphi)\mapsto (\Phi_{M,a}(x),\theta,\varphi)$.
Pulling back tensors by $\Phi_{M,a}$, we may regard $g_{M,a}$, $\Box_{g_{M,a}}$, and the stationary
family $P_{M,a}(\omega)$ defined below as a real-analytic family of operators on a fixed underlying
manifold. This fixed-manifold viewpoint is convenient for the inverse analysis, where we differentiate
spectral data with respect to parameters.

\subsection{Regular coordinates and the future extension across horizons}\label{subsec:starcoords}

The Boyer--\allowbreak Lindquist chart \eqref{eq:kds_metric_BL} is singular at the horizons where $\Delta_r=0$.
For the microlocal/Fredholm definition of quasinormal modes, it is standard to pass to
Eddington--Finkelstein type coordinates which extend smoothly across the \emph{future} event and
cosmological horizons; see, for instance, \cite[\S2]{PetersenVasy2025WaveKdS} and
\cite[\S1.1]{PetersenVasy2023AnalyticityQNM} for a careful discussion in the Kerr--de~Sitter setting.

\paragraph{Star coordinates.}
Fix a smooth function $f\in C^\infty(I)$ with $f\equiv -1$ near $x=0$ and $f\equiv +1$ near $x=1$.
On $I_{M,a}$ we write $f_{M,a}(r):=f(\Phi_{M,a}^{-1}(r))$.
Define the tortoise-type derivatives
\begin{equation}\label{eq:tortoise_primitives}
F_{M,a}'(r)=\frac{r^2+a^2}{\Delta_r(r)}\,f_{M,a}(r),\qquad
G_{M,a}'(r)=\frac{a\,\Xi}{\Delta_r(r)}\,f_{M,a}(r),
\end{equation}
and impose the differential relations
\begin{equation}\label{eq:star_differentials}
dt_* := dt - F_{M,a}'(r)\,dr,\qquad
d\varphi_* := d\varphi - G_{M,a}'(r)\,dr.
\end{equation}
Since $f_{M,a}$ is constant near each horizon, \eqref{eq:star_differentials} agrees near $r=r_e$
with the standard ingoing Eddington--Finkelstein transformation and near $r=r_c$ with the corresponding
choice that is regular at the future cosmological horizon.
In particular, expressed in $(t_*,r,\theta,\varphi_*)$ the metric extends smoothly across both
$r=r_e$ and $r=r_c$.

\paragraph{An extended spacetime.}
Using the identification \eqref{eq:affine_identification}, we work on the fixed extended spacetime
\begin{equation}\label{eq:fixed_spacetime}
\mathcal M := \mathbb{R}_{t_*}\times X,
\quad\text{with coordinates}\quad (t_*,x,\theta,\varphi_*),
\end{equation}
where the parameter-dependent metric $g_{M,a}$ is pulled back from
$\mathbb{R}_{t_*}\times I_{M,a}\times\mathbb{S}^2$ via $\Phi_{M,a}$.
We write $\mathcal M^\circ\subset\mathcal M$ for the physical domain of outer communications,
corresponding to $x\in(0,1)$ (i.e.\ $r\in(r_e,r_c)$).

\paragraph{Killing fields.}
In $(t_*,\varphi_*)$ coordinates, the stationary and axial Killing fields are
\begin{equation}\label{eq:Killing}
T:=\partial_{t_*},\qquad \Phi:=\partial_{\varphi_*}.
\end{equation}
We define quasinormal modes via Fourier analysis with respect to $T$.

\begin{remark}[Alternative Killing fields and the full subextremal range]\label{rem:altKilling}
In the slow-rotation regime considered here, the Kerr--de~Sitter Fredholm setup of Vasy \cite{Vasy2013MicrolocalAHKdS}
and Dyatlov \cite{Dyatlov2011QNMKerrDeSitter,Dyatlov2012AsymptoticQNMKdS} is naturally formulated using the stationary
generator $T=\partial_{t_*}$, and this is the convention we adopt for quasinormal frequencies.
In the full subextremal Kerr--de~Sitter range, Petersen--Vasy \cite{PetersenVasy2025WaveKdS} work with a modified stationary
problem associated with a Killing field of the form $K=T+\Omega\,\Phi$ in order to obtain global mode expansions.
Since our inverse statements are restricted to a compact slow-rotation set $\mathcal K$, we keep the canonical choice $T$.
\end{remark}

\begin{remark}[Frequencies in $(t_*,\varphi_*)$ and Boyer--Lindquist $(t,\varphi)$]\label{rem:tstar-t}
The differential relations \eqref{eq:star_differentials} integrate (after fixing constants of integration) to
$t_*=t-F_{M,a}(r)$ and $\varphi_*=\varphi-G_{M,a}(r)$.
Hence $\partial_{t_*}=\partial_t$ and $\partial_{\varphi_*}=\partial_\varphi$.
In particular, the temporal frequency $\omega$ and azimuthal number $k$ in the separated ansatz
$e^{-i\omega t}e^{ik\varphi}$ used in Section~\ref{sec:HF-quantization} agree with those appearing in the
Fourier decomposition with respect to $(t_*,\varphi_*)$ in Definition~\ref{def:QNM}.
\end{remark}

\subsection{The wave operator and stationary reduction}\label{subsec:stationary_reduction}

\paragraph{Wave operator.}
Let $\Box_g$ denote the d'Alembertian associated with $g=g_{M,a}$:
\begin{equation}\label{eq:dalembertian}
\Box_g u = |g|^{-1/2}\,\partial_\mu\bigl(|g|^{1/2}g^{\mu\nu}\partial_\nu u\bigr),
\end{equation}
where $|g|:=|\det(g_{\mu\nu})|$.
We work with the scalar wave operator
\begin{equation}\label{eq:P_def}
P:=\Box_g.
\end{equation}
(The discussion applies with minor changes to $\Box_g-\mu^2$ and to tensorial wave operators, but we
restrict to the scalar case for clarity.)
The principal symbol of $P$ is
\begin{equation}\label{eq:principal_symbol}
p(x,\xi)=g^{\mu\nu}(x)\,\xi_\mu\xi_\nu,\qquad (x,\xi)\in T^*\mathcal M.
\end{equation}

\paragraph{Stationary family.}
Since $[P,T]=0$, we may Fourier transform in $t_*$.
For $u\in C_c^\infty(\mathbb{R}_{t_*};C^\infty(X))$, define
\[
\widehat u(\omega,\cdot):=\int_{\mathbb{R}} e^{i\omega t_*}u(t_*,\cdot)\,dt_*.
\]
Formally, $u(t_*,x)=e^{-i\omega t_*}v(x)$ solves $Pu=0$ if and only if
\begin{equation}\label{eq:stationary_family}
P(\omega)v=0,\qquad
P(\omega):=e^{i\omega t_*}\,P\,e^{-i\omega t_*}\big|_{t_*=\mathrm{const}},
\end{equation}
i.e.\ $P(\omega)$ is obtained from $P$ by replacing $D_{t_*}:=\frac1{i}\partial_{t_*}$ by $\omega$.
In local coordinates on $X$, $P(\omega)$ is a second order differential operator whose coefficients
depend real-analytically on $(M,a)\in\mathcal K$ and polynomially on $\omega$:
\begin{equation}\label{eq:polynomial_omega}
P(\omega)=P_0+\omega P_1+\omega^2 P_2,
\end{equation}
where $P_j$ are (parameter-dependent) differential operators on $X$ of order $\le 2-j$.
Its principal symbol on $T^*X$ is
\begin{equation}\label{eq:principal_symbol_stationary}
p_\omega(x,\xi):=p\bigl((t_*,x),\omega\,dt_*+\xi\bigr)\Big|_{t_*=\mathrm{const}},
\qquad (x,\xi)\in T^*X,
\end{equation}
which is independent of $t_*$.

\paragraph{Azimuthal decomposition.}
Since $[P,\Phi]=0$, we can decompose in Fourier modes in $\varphi_*$:
\begin{equation}\label{eq:azimuthal_modes}
L^2(X)=\bigoplus_{k\in\mathbb Z} L^2_k(X),\qquad
L^2_k(X):=\{v\in L^2(X):\ \Phi v = ik\,v\}.
\end{equation}
The operator $P(\omega)$ preserves $L^2_k(X)$; we denote the restriction by $P_k(\omega)$.

\subsection{Sobolev spaces, Fredholm setup, and the definition of quasinormal modes}\label{subsec:fredholm_qnm}

Quasinormal modes are most robustly defined as poles of a meromorphically continued inverse of
$P(\omega)$, formulated as a Fredholm problem with outgoing conditions at the horizons.
In the Kerr--de~Sitter setting, this is achieved by combining radial-point microlocal analysis with
analytic Fredholm theory; see Vasy's general framework \cite{Vasy2013MicrolocalAHKdS} and, for Kerr--de~Sitter
specifically, Dyatlov \cite{Dyatlov2011QNMKerrDeSitter} (slow rotation) and Petersen--Vasy
\cite{PetersenVasy2025WaveKdS} (full subextremal range).

\paragraph{Variable order Sobolev spaces.}
Fix a smooth positive density $d\mu$ on $X$.
For $s\in\mathbb{R}$, let $H^s(X)$ be the standard Sobolev space.
We also use variable order Sobolev spaces $H^{\mathbf s}(X)$, where the order function
$\mathbf s\in C^\infty(S^*X)$ varies on the cosphere bundle; we refer to \cite[\S2]{Vasy2013MicrolocalAHKdS}
for a systematic construction and mapping properties.

\paragraph{Outgoing order functions and Fredholm spaces.}
For fixed $\omega\in\mathbb{C}$, the characteristic set $\Sigma_\omega=\{p_\omega=0\}\subset T^*X$
has radial sets at the horizons. One chooses an order function $\mathbf s=\mathbf s_\omega$ which is
above the radial point threshold at the \emph{future} radial sets (thus imposing outgoing regularity),
below threshold at the \emph{past} radial sets, and strictly decreasing along the Hamilton flow between them;
this is the standard choice in the Kerr--de~Sitter Fredholm setup \cite{Vasy2013MicrolocalAHKdS,PetersenVasy2025WaveKdS}.
Given such $\mathbf s$, define
\begin{equation}\label{eq:XsYs}
\mathcal Y^{\mathbf s}:=H^{\mathbf s}(X),\qquad
\mathcal X^{\mathbf s}:=\{u\in H^{\mathbf s}(X):\ P(\omega)u\in H^{\mathbf s-1}(X)\},
\end{equation}
endowed with the graph norm
$\|u\|_{\mathcal X^{\mathbf s}}:=\|u\|_{H^{\mathbf s}}+\|P(\omega)u\|_{H^{\mathbf s-1}}$.
Then $P(\omega):\mathcal X^{\mathbf s}\to\mathcal Y^{\mathbf s-1}$ is bounded.

\begin{theorem}[Meromorphic resolvent and quasinormal modes]\label{thm:meromorphy}
Fix $\Lambda>0$ and $(M,a)\in\mathcal K\Subset \mathcal P_\Lambda$.
For each $\omega\in\mathbb{C}$ one can choose an outgoing variable order $\mathbf s_\omega$ such that
\begin{equation}\label{eq:fredholm}
P(\omega):\mathcal X^{\mathbf s_\omega}\longrightarrow \mathcal Y^{\mathbf s_\omega-1}
\end{equation}
is Fredholm of index $0$.
Moreover, there exists $\omega_0\in\mathbb{R}$ such that for $\Im \omega>\omega_0$ the operator $P(\omega)$
is invertible, and the inverse family
\begin{equation}\label{eq:resolvent_family}
R(\omega):=P(\omega)^{-1}:\mathcal Y^{\mathbf s_\omega-1}\longrightarrow \mathcal X^{\mathbf s_\omega}
\end{equation}
extends meromorphically to $\omega\in\mathbb{C}$ with finite rank residues.
The set of poles $\mathrm{Res}(P)\subset\mathbb{C}$ is discrete and is called the set of
\emph{quasinormal mode frequencies} (or \emph{resonances}).
For $\omega\in \mathrm{Res}(P)$, the nontrivial elements of $\ker P(\omega)\subset\mathcal X^{\mathbf s_\omega}$
are called \emph{resonant states} (quasinormal modes).
\end{theorem}

\begin{definition}[Quasinormal modes]\label{def:QNM}
Fix parameters $(M,a)$ in the subextremal Kerr--de~Sitter regime and let $P(\omega)$ be the stationary family
\eqref{eq:stationary_family}.
A complex number $\omega\in\mathbb C$ is called a \emph{quasinormal mode frequency} (or \emph{resonance}) if $P(\omega)$ fails to be invertible
in the outgoing Fredholm setup of Theorem~\ref{thm:meromorphy}, equivalently if $\omega$ is a pole of the meromorphic inverse family $R(\omega)=P(\omega)^{-1}$.
A nontrivial element $u\in\ker P(\omega)\subset\mathcal X^{\mathbf s_\omega}$ is called a \emph{resonant state} (a quasinormal mode).
For each azimuthal sector $k\in\mathbb Z$, the same definition applied to the restricted family $P_k(\omega)$ gives the discrete
set of $k$--mode resonances $\mathrm{Res}(P_k)$.
\end{definition}

\begin{remark}[Outgoing boundary conditions and regularity]\label{rem:smoothness}
In the Kerr--de~Sitter setting, resonant states correspond to mode solutions
$u(t_*,x)=e^{-i\omega t_*}v_\omega(x)$ which extend smoothly across the \emph{future} event and cosmological
horizons; this gives a precise mathematical implementation of the outgoing boundary condition used in the
physics literature.
If the coefficients are real analytic, then in the subextremal regime the resonant states are in fact
real analytic near the horizons; see Petersen--Vasy \cite{PetersenVasy2023AnalyticityQNM}.
\end{remark}

\begin{remark}[Analytic dependence on parameters]\label{rem:analytic_parameters}
Because we work on a fixed manifold $X$ and the coefficients of $P(\omega)$ depend real-analytically
on $\mu:=(M,a)\in\mathcal K$, the stationary family $(\omega,\mu)\mapsto P(\omega;\mu)$ is a real-analytic
family of Fredholm operators in the Kerr--de~Sitter setup of Theorem~\ref{thm:meromorphy}.
In particular, near a \emph{simple} resonance $\omega_0$ at $\mu_0$, analytic perturbation theory implies
that $\omega_0$ admits a locally real-analytic continuation $\omega(\mu)$ (and likewise for the associated
finite-rank spectral projection) in a neighborhood of $\mu_0$; see \cite[Chapters~VII--VIII]{KatoPerturbation}.
We only use this analytic dependence at the end of the paper to upgrade $C^1$ local inverses to real-analytic ones.
All stability estimates and implicit-function arguments in the inverse problem itself require only uniform $C^1$
control of the relevant families on the compact parameter set $\mathcal K$.
\end{remark}

\medskip

We will be interested in a semiclassical high-frequency regime with fixed overtone $n$ and large total angular
momentum $\ell$ in the separated problem.
The semiclassical parameter used throughout the quantization and inverse analysis will be introduced in
Section~\ref{sec:HF-quantization}; the present section is concerned only with the geometric and Fredholm setup.

\section{High-frequency quantization and labeling}\label{sec:HF-quantization}

This section recalls the semiclassical description of high-frequency Kerr--de~Sitter
quasinormal modes (QNMs) in a form adapted to the inverse problems studied later.
The decisive input is Dyatlov's all-orders Bohr--Sommerfeld type quantization for
slowly rotating Kerr--de~Sitter black holes \cite{Dyatlov2011QNMKerrDeSitter,Dyatlov2012AsymptoticQNMKdS},
built on Vasy's microlocal Fredholm framework \cite{Vasy2013MicrolocalAHKdS} and
normally hyperbolic trapping estimates in the sense of Wunsch--Zworski \cite{WunschZworski2011ResolventNHT}.
We also record two explicit computations: the Schwarzschild--de~Sitter leading lattice and the
linear-in-$a$ equatorial splitting (``Zeeman slope'') that will be used in the reconstruction.

Throughout this section, $\Lambda>0$ is fixed and $(M,a)$ lie in the slowly rotating subextremal regime.
We work with the frequency parameter $\omega\in\mathbb C$ appearing in the stationary reduction of
Section~\ref{sec:setup} and in the separation of variables below.
We use the convention $D_x:=-i\,\partial_x$.

\subsection{Separation and \texorpdfstring{$\lambda$}{lambda}--resolvents}\label{subsec:separation-lambda-resolvent}

\begin{sloppypar}
In Boyer--\allowbreak Lindquist coordinates $(t,r,\theta,\varphi)$ on the domain of outer communications, the stationary wave equation
$\Box_g u=0$ separates at a fixed temporal frequency $\omega$.
We decompose in the azimuthal mode $k\in\mathbb Z$.
More precisely, writing
\end{sloppypar}
\[
u(t,r,\theta,\varphi)=e^{-i\omega t}e^{ik\varphi}v(r,\theta),
\]
the reduced operator acting on $v$ splits as
\begin{equation}\label{eq:PrPtheta}
P(\omega)=P_r(\omega,k)+P_\theta(\omega),
\end{equation}
where
\begin{align}
P_r(\omega,k)
&:=D_r\bigl(\Delta_r D_r\bigr)\;-\;\frac{\Xi^2}{\Delta_r}\Bigl((r^2+a^2)\omega-a k\Bigr)^2,
\label{eq:Pr-def}\\
P_\theta(\omega)
&:=\frac1{\sin\theta}D_\theta\bigl(\Delta_\theta\sin\theta\,D_\theta\bigr)
+\frac{\Xi^2}{\Delta_\theta\sin^2\theta}\Bigl(a\omega\sin^2\theta-D_\varphi\Bigr)^2.
\label{eq:Ptheta-def}
\end{align}
(Here $\Delta_r,\Delta_\theta,\Xi$ are as in Section~\ref{sec:setup}.)

The separation constant $\lambda\in\mathbb C$ couples the radial and angular problems: a separated mode corresponds
to a pair $(v_r,v_\theta)$ satisfying
\begin{equation}\label{eq:separated-eq}
\bigl(P_r(\omega,k)+\lambda\bigr)v_r=0,\qquad
\bigl(P_\theta(\omega)+\lambda\bigr)v_\theta=0.
\end{equation}

Fix $\delta>0$ and set $K_r:=(r_e+\delta,r_c-\delta)\Subset(r_e,r_c)$.
Following \cite{Dyatlov2011QNMKerrDeSitter,Dyatlov2012AsymptoticQNMKdS}, we consider the cutoff radial resolvent
\begin{equation}\label{eq:Rr-def}
R_r(\omega,\lambda,k):=\mathbf{1}_{K_r}\bigl(P_r(\omega,k)+\lambda\bigr)^{-1}\mathbf{1}_{K_r},
\end{equation}
which is a meromorphic family in $\lambda$ (for fixed $(\omega,k)$) with finite-rank residues;
the pole condition encodes the outgoing boundary condition at the horizons via the Kerr--de~Sitter Fredholm setup.
Similarly, since $P_\theta(\omega)$ is elliptic on $\mathbb S^2$ for each fixed $\omega$,
\begin{equation}\label{eq:Rtheta-def}
R_\theta(\omega,\lambda):=\bigl(P_\theta(\omega)+\lambda\bigr)^{-1}
\end{equation}
is meromorphic in $\lambda$ with finite-rank residues, and restricting to the $k$th azimuthal sector amounts to
acting on $e^{ik\varphi}$-modes.

We now make precise the equivalence between QNMs in a fixed azimuthal sector and common $\lambda$--poles.

\begin{proposition}[QNMs and common $\lambda$--poles]\label{lem:QNM-common-poles}
Fix $k\in\mathbb Z$ and $\omega\in\mathbb C$.
Then $\omega$ is a quasinormal mode \emph{in the $k$th azimuthal sector}
(i.e.\ a pole of the meromorphically continued inverse of the stationary operator $P_k(\omega)$ on $e^{ik\varphi}$-modes)
if and only if there exists $\lambda\in\mathbb C$ such that $\lambda$ is simultaneously a pole of the cutoff
radial resolvent $R_r(\omega,\lambda,k)$ and a pole of the angular resolvent $R_\theta(\omega,\lambda)$
restricted to the $k$th azimuthal subspace.
Equivalently, there exist nontrivial $v_r$ and $v_\theta$ solving \eqref{eq:separated-eq}, with $v_\theta$ smooth on
$\mathbb S^2$ and $v_r$ satisfying the outgoing condition at $r=r_e$ and $r=r_c$.
\end{proposition}

\begin{proof}
Suppose first that $\omega$ is a QNM in the $k$th sector.
By definition (Theorem~\ref{thm:meromorphy}, applied to the restriction $P_k(\omega)$), there exists a nontrivial
resonant state $u(t,r,\theta,\varphi)=e^{-i\omega t}e^{ik\varphi}v(r,\theta)$ satisfying the outgoing condition at the
event and cosmological horizons.
Substituting into $\Box_g u=0$ gives $P(\omega)v=0$, i.e.\ \eqref{eq:PrPtheta} with $P_r(\omega,k)$ and $P_\theta(\omega)$.
Since $P_\theta(\omega)$ is elliptic on $\mathbb S^2$, the restriction of $R_\theta(\omega,\lambda)$ to the $k$th azimuthal
subspace is a meromorphic family in $\lambda$ whose poles are precisely the (discrete) eigenvalues $-\lambda$ for which the
equation $(P_\theta(\omega)+\lambda)v_\theta=0$ has a nontrivial smooth solution.
Expanding $v(r,\theta)$ in the corresponding (finite-dimensional) eigenspace at such a pole, one finds a separation constant
$\lambda$ for which \eqref{eq:separated-eq} holds with $v_\theta\not\equiv0$.
The outgoing condition on $u$ implies that the corresponding radial factor $v_r$ satisfies the outgoing condition at the horizons,
which is equivalent to $\lambda$ being a pole of the cutoff radial resolvent $R_r(\omega,\lambda,k)$; see
\cite{Vasy2013MicrolocalAHKdS,Dyatlov2011QNMKerrDeSitter,Dyatlov2012AsymptoticQNMKdS} for the microlocal construction.

Conversely, assume that $\lambda$ is a common pole of the two meromorphic families.
Then the residue of $R_\theta(\omega,\lambda)$ at $\lambda$ yields a nontrivial smooth angular mode $v_\theta$ with
$(P_\theta(\omega)+\lambda)v_\theta=0$ in the $k$th sector.
Likewise, the residue of $R_r(\omega,\lambda,k)$ yields a nontrivial radial solution $v_r$ of
$(P_r(\omega,k)+\lambda)v_r=0$ with the outgoing behavior at the horizons.
The separated product
\[
u(t,r,\theta,\varphi):=e^{-i\omega t}e^{ik\varphi}v_r(r)v_\theta(\theta)
\]
then satisfies $\Box_g u=0$ and the outgoing condition, hence $u$ is a resonant state and $\omega$ is a QNM in the $k$th sector.
\end{proof}

\subsection{Semiclassical scaling and symbol classes}\label{subsec:semiclassical-scaling}

We focus on the high-frequency regime $\Re\omega\to+\infty$ within a fixed strip
\begin{equation}\label{eq:HF-strip}
|\Im\omega|<\nu_0,\qquad \nu_0>0\ \text{fixed}.
\end{equation}
On a dyadic block $|\Re\omega|\sim h^{-1}$ we introduce the semiclassical variables
\begin{equation}\label{eq:scaled-vars}
\tilde\omega:=h\,\Re\omega,\qquad \tilde\nu:=\Im\omega,\qquad
\tilde k:=h\,k,\qquad
\tilde\lambda:=h^2\,\Re\lambda,\qquad \tilde\mu:=h\,\Im\lambda,
\end{equation}
so that $(\tilde\omega,\tilde\nu,\tilde k,\tilde\lambda,\tilde\mu)$ remain $\mathcal O(1)$ in the region of interest.

A function $a(\tilde\omega,\tilde\nu,\tilde k;h)$ is a \emph{classical symbol of order $m$} in $(\tilde\omega,\tilde k)$ if it
admits an asymptotic expansion
\begin{equation}\label{eq:classical-symbol}
a(\tilde\omega,\tilde\nu,\tilde k;h)\sim \sum_{j\ge0} a_j(\tilde\omega,\tilde\nu,\tilde k)\,h^j,
\qquad a_j(\rho\tilde\omega,\tilde\nu,\rho\tilde k)=\rho^{m-j}a_j(\tilde\omega,\tilde\nu,\tilde k),
\end{equation}
with coefficients $a_j$ smooth in $(\tilde\omega,\tilde\nu,\tilde k)$ away from the origin.
When a quantization rule is formulated in semiclassical variables with $h\sim (\ell+\tfrac12)^{-1}$, we also use the standard
``removal of $h$'' to obtain a homogeneous (non-semiclassical) symbol: setting
\begin{equation}\label{eq:remove-h}
h:=(\ell+\tfrac12)^{-1},\qquad \tilde\ell:=h\bigl(\ell+\tfrac12\bigr)\equiv 1,\qquad \tilde k:=h\,k,
\end{equation}
we define
\[
F(n,\ell,k):=h^{-1}F^\omega\bigl(n,\tilde\ell,\tilde k;h\bigr).
\]
This produces a classical expansion in powers of $\ell^{-1}$ (see \eqref{eq:F-symbol} below).

\subsection{Quantization of radial and angular \texorpdfstring{$\lambda$}{lambda}--poles}\label{subsec:quantization-maps}

The next two propositions are extracted from the semiclassical analysis of Dyatlov
\cite{Dyatlov2012AsymptoticQNMKdS} (see also \cite{Dyatlov2011QNMKerrDeSitter}).
They encode, respectively, the hyperbolic (barrier-top) quantization for the radial ODE and the Bohr--Sommerfeld quantization
for spheroidal harmonics.
We use them as a black box, keeping track of the dependence on the physical parameters $(M,a)$.

\begin{proposition}[Radial quantization map]\label{prop:radial-quant}
Fix $\nu_0>0$. There exist constants $C_n\in\mathbb N$, $C_\lambda>0$, $C_k>0$ and a classical symbol
$F^r(n,\tilde\omega,\tilde\nu,\tilde k;h)$ of order $2$ in $(\tilde\omega,\tilde k)$, smooth in
the parameters $(M,a)$, with real principal part $F^r_0$ independent of $n$ and $\tilde\nu$, such that the poles of the
cutoff radial resolvent $R_r(\omega,\lambda,k)$ in the region
\begin{equation}\label{eq:radial-lambda-region}
1<\tilde\omega<2,\qquad |\tilde\nu|<\nu_0,\qquad |k|\le C_k h^{-1},\qquad |\lambda|\le C_\lambda h^{-2},
\end{equation}
are precisely the values
\begin{equation}\label{eq:radial-F}
\tilde\lambda+i h\tilde\mu=F^r(n,\tilde\omega,\tilde\nu,\tilde k;h),\qquad n\in\{0,1,\dots,C_n\},
\end{equation}
modulo $O(h^\infty)$, in the sense that each pole differs from the right-hand side by $O(h^N)$ for every $N$.

At $a=0$ the leading term is explicitly given by the barrier-top model at the photon sphere:
\begin{equation}\label{eq:Fr-a0}
\begin{aligned}
F^r(n,\tilde\omega,\tilde\nu,\tilde k;h)
&=\Bigl[ih\Bigl(n+\tfrac12\Bigr)+\frac{3\sqrt3\,M}{\sqrt{\Delta}}\bigl(\tilde\omega+i h\tilde\nu\bigr)\Bigr]^2
+O(h^2),\\
&\qquad \Delta:=1-9\Lambda M^2.
\end{aligned}
\end{equation}
\end{proposition}

\begin{remark}
The structure \eqref{eq:radial-F} reflects the fact that, microlocally near trapping, the radial ODE reduces to a semiclassical
barrier-top resonance problem; the imaginary part is produced by quantization of the hyperbolic directions transverse to the trapped
torus. We refer to \cite{Dyatlov2012AsymptoticQNMKdS} for the construction and to \cite{WunschZworski2011ResolventNHT} for the
normally hyperbolic trapping resolvent estimates underpinning the microlocal control.
\end{remark}

\begin{proposition}[Angular quantization map]\label{prop:angular-quant}
Fix $\nu_0>0$. There exist constants $C_\lambda>0$, $C_k>0$ and a classical symbol
$F^\theta(\tilde\ell,\tilde\omega,\tilde\nu,\tilde k;h)$ of order $2$ in $(\tilde\ell,\tilde k)$, smooth
in $(M,a)$, with real principal part $F^\theta_0$ independent of $\tilde\nu$, such that the poles of the angular resolvent
$R_\theta(\omega,\lambda)$ (restricted to the $k$th azimuthal mode) in the region
\begin{equation}\label{eq:angular-lambda-region}
1<\tilde\omega<2,\qquad |\tilde\nu|<\nu_0,\qquad |k|\le C_k h^{-1},\qquad |\lambda|\le C_\lambda h^{-2},
\end{equation}
are precisely the values
\begin{equation}\label{eq:angular-F}
\tilde\lambda+i h\tilde\mu=F^\theta(\tilde\ell,\tilde\omega,\tilde\nu,\tilde k;h),
\qquad \tilde\ell\in h\bigl(\mathbb N+\tfrac12\bigr),
\end{equation}
modulo $O(h^\infty)$, in the same sense as in Proposition~\ref{prop:radial-quant}.

At $a=0$, one has $P_\theta(\omega)=-\Delta_{\mathbb S^2}$ (independent of $\omega$), hence the poles occur at the exact spherical
eigenvalues $\lambda=\ell(\ell+1)$.
In the semiclassical scaling \eqref{eq:scaled-vars} with the Langer-shifted quantum number
$\tilde\ell=h(\ell+\tfrac12)$, this corresponds to the exact identity
\begin{equation}\label{eq:Ftheta-a0}
F^\theta(\tilde\ell,\tilde\omega,\tilde\nu,\tilde k;h)=\tilde\ell^2-\frac{h^2}{4}\qquad (a=0),
\end{equation}
so that the principal part is $F^\theta_0(\tilde\ell)=\tilde\ell^2$.

Moreover, at the equatorial locus $\tilde\ell=\pm\tilde k$ one has
\begin{equation}\label{eq:equatorial-Ftheta0}
F^\theta_0(\pm\tilde k,\tilde\omega,\tilde k)=\Xi^2(\tilde k-a\tilde\omega)^2+O(a^2),
\end{equation}
and its derivatives at that locus satisfy
\begin{equation}\label{eq:equatorial-derivs}
\partial_{\tilde\ell}F^\theta_0(\pm\tilde k,\tilde\omega,\tilde k)=\pm 2\tilde k+O(a^2),
\qquad
\partial_{\tilde k}F^\theta_0(\pm\tilde k,\tilde\omega,\tilde k)= -2a\tilde\omega+O(a^2).
\end{equation}
\end{proposition}

\subsection{Pseudopoles, QNMs, and the frequency symbol}\label{subsec:pseudopoles}

By Proposition~\ref{lem:QNM-common-poles}, QNMs in a fixed $k$--mode occur when the radial and angular $\lambda$--poles coincide.
In semiclassical variables this leads to solving
\begin{equation}\label{eq:quantization-eq}
F^r(n,\tilde\omega,\tilde\nu,\tilde k;h)=F^\theta(\tilde\ell,\tilde\omega,\tilde\nu,\tilde k;h)
\end{equation}
for the complex quantity $\tilde\omega+i h\tilde\nu$, given $(n,\tilde\ell,\tilde k)$.

\begin{lemma}[Implicit solvability of the frequency symbol]\label{lem:implicit-frequency-symbol}
On the parameter ranges in Propositions~\ref{prop:radial-quant}--\ref{prop:angular-quant},
the equation \eqref{eq:quantization-eq} has a unique solution of the form
\begin{equation}\label{eq:omega-symbol}
\tilde\omega+i h\tilde\nu = F^\omega(n,\tilde\ell,\tilde k;h),
\end{equation}
where $F^\omega$ is a classical symbol of order $1$ in $(\tilde\ell,\tilde k)$, smooth in $(M,a)$,
with real principal part $F^\omega_0$ independent of $n$.
\end{lemma}

\begin{proof}
Write \eqref{eq:quantization-eq} as $G(\tilde\omega+i h\tilde\nu;\,n,\tilde\ell,\tilde k;h)=0$, where $G:=F^r-F^\theta$.
At $a=0$, the principal symbols satisfy
\[
F^r_0(\tilde\omega)=\frac{27M^2}{\Delta}\tilde\omega^2,\qquad
F^\theta_0(\tilde\ell)=\tilde\ell^2
\]
(cf.\ \eqref{eq:Fr-a0} and the principal part of \eqref{eq:Ftheta-a0}), hence
\[
G_0(\tilde\omega;\tilde\ell)=\frac{27M^2}{\Delta}\tilde\omega^2-\tilde\ell^2,
\qquad
\partial_{\tilde\omega}G_0(\tilde\omega;\tilde\ell)=2\frac{27M^2}{\Delta}\tilde\omega.
\]
In the region $1<\tilde\omega<2$ the derivative $\partial_{\tilde\omega}G_0$ is uniformly bounded away from $0$.
By continuity in $(M,a)$ and the symbol structure (in particular, the absence of $\tilde\nu$ in the principal parts),
$\partial_{\tilde\omega}G$ remains uniformly nonzero for $|a|$ and $h$ sufficiently small on the region
\eqref{eq:radial-lambda-region}--\eqref{eq:angular-lambda-region}.
The complex implicit function theorem therefore yields a unique classical solution
$\tilde\omega+i h\tilde\nu=F^\omega(n,\tilde\ell,\tilde k;h)$.
The principal part is real since $F^r_0$ and $F^\theta_0$ are real, and independence of $n$ follows from
Proposition~\ref{prop:radial-quant}.
\end{proof}

\paragraph{Non-semiclassical pseudopoles.}
Removing $h$ as in \eqref{eq:remove-h}, we obtain a homogeneous symbol
\begin{equation}\label{eq:F-symbol}
F(n,\ell,k)\sim \sum_{j\ge 0}F_j(n,\ell,k),\qquad F_j(n,\rho\ell,\rho k)=\rho^{1-j}F_j(n,\ell,k),
\end{equation}
defined for $0\le n\le C_n$ and $|k|\le \ell$.
We call the values
\begin{equation}\label{eq:pseudopoles}
\omega^\sharp_{n,\ell,k}:=F(n,\ell,k)
\end{equation}
the \emph{pseudopoles}. They form a lattice-like subset of the high-frequency region and approximate the true QNMs with
super-polynomial accuracy.

\begin{remark}[Choosing a representative for the full symbol] \label{rem:pseudopole-representative}
The asymptotic expansion \eqref{eq:F-symbol} determines $F$ only modulo $\mathcal O(\ell^{-\infty})$.
For definiteness, we fix once and for all an asymptotic summation (for instance, a Borel summation in $\ell^{-1}$)
to obtain an actual smooth function $F(n,\ell,k)$ whose expansion is \eqref{eq:F-symbol}.
Different choices of such a representative change $F$ (hence the pseudopoles \eqref{eq:pseudopoles}) by
$\mathcal O(\ell^{-N})$ for every $N$ on the quantization range.
Since all quantitative statements below are formulated with remainders $\mathcal O(\ell^{-N})$ for arbitrary $N$,
they are independent of this choice.
\end{remark}

\begin{theorem}[High-frequency QNMs are given by pseudopoles]\label{thm:Dyatlov-quantization}
Fix $\nu_0>0$ and parameters $(M,a)$ in the slowly rotating subextremal regime.
Then there exist constants $C_\omega>0$ and $C_n\in\mathbb{N}$ such that the set of
QNMs $\omega$ satisfying
\begin{equation}\label{eq:HF-region-for-thm}
\Re\omega>C_\omega,\qquad \Im\omega>-\nu_0
\end{equation}
coincides \emph{modulo $O(|\omega|^{-\infty})$} with the set of pseudopoles $\{\omega^\sharp_{n,\ell,k}\}$ with indices
\begin{equation}\label{eq:index-set}
n\in\{0,1,\dots,C_n\},\qquad \ell\in\mathbb{N},\qquad k\in\mathbb{Z},\quad |k|\le \ell,
\end{equation}
in the following sense: there exists a bijection between QNMs in \eqref{eq:HF-region-for-thm} and pseudopoles
such that the paired elements differ by $O(|\omega|^{-N})$ for every $N$ as $\Re\omega\to\infty$.
\end{theorem}

\begin{remark}[Equatorial boundary regime]\label{rem:equatorial-boundary}
The index set \eqref{eq:index-set} includes the boundary case $|k|=\ell$, hence the equatorial families used in the inverse
problem lie within the scope of Dyatlov's quantization theorem.
This inclusion of the boundary $|k|=\ell$ is explicit in \cite[Theorem~1 and equation~(0.2)]{Dyatlov2012AsymptoticQNMKdS},
where the quantization symbols are constructed on (and smooth up to) the closed cone $\{\,|k|\le \ell\,\}$.
\end{remark}

\begin{remark}\label{rem:meaning}
Theorem~\ref{thm:Dyatlov-quantization} is a rigorous microlocal realization of the geometric-optics principle that high-$\ell$
ringdown is governed by trapped null geodesic dynamics.
The principal symbol $F_0$ is the frequency map induced by the integrable Hamiltonian flow on the trapped set, while the leading
imaginary part arises from quantization of the hyperbolic directions transverse to the trapped invariant tori.
We refer to \cite{Dyatlov2012AsymptoticQNMKdS} for the detailed construction.
\end{remark}

\subsection{Explicit leading asymptotics at \texorpdfstring{$a=0$}{a=0} and equatorial splitting}\label{subsec:explicit-leading}

We extract from the symbol $F$ two explicit leading-order formulas used later.

\begin{corollary}[Schwarzschild--de~Sitter lattice]\label{cor:a0-lattice}
At $a=0$, the pseudopoles satisfy
\begin{equation}\label{eq:a0-lattice}
\omega^\sharp_{n,\ell,k}
=\frac{\sqrt{\Delta}}{3\sqrt3\,M}\Bigl[\bigl(\ell+\tfrac12\bigr)-i\bigl(n+\tfrac12\bigr)\Bigr]
+O(\ell^{-1}),
\qquad \Delta=1-9\Lambda M^2,
\end{equation}
uniformly for fixed $n$ and $|k|\le \ell$ as $\ell\to\infty$.
\end{corollary}

\begin{proof}
At $a=0$ the quantization equation \eqref{eq:quantization-eq} reads, using \eqref{eq:Fr-a0} and \eqref{eq:Ftheta-a0},
\begin{equation}\label{eq:a0-quant-eq-proof}
\Bigl[ih\Bigl(n+\tfrac12\Bigr)+\frac{3\sqrt3\,M}{\sqrt{\Delta}}\bigl(\tilde\omega+i h\tilde\nu\bigr)\Bigr]^2
=\tilde\ell^2-\frac{h^2}{4}+O(h^2).
\end{equation}
We now remove $h$ via \eqref{eq:remove-h}, i.e.\ take $h=(\ell+\tfrac12)^{-1}$ so that $\tilde\ell=h(\ell+\tfrac12)\equiv 1$.
Since the correction $-h^2/4$ is $O(h^2)$, it can be absorbed into the remainder in \eqref{eq:a0-quant-eq-proof}; thus
\[
\Bigl[ih\Bigl(n+\tfrac12\Bigr)+\frac{3\sqrt3\,M}{\sqrt{\Delta}}\bigl(\tilde\omega+i h\tilde\nu\bigr)\Bigr]^2
=1+O(h^2).
\]
Taking the square root with positive real part gives
\[
ih\Bigl(n+\tfrac12\Bigr)+\frac{3\sqrt3\,M}{\sqrt{\Delta}}\bigl(\tilde\omega+i h\tilde\nu\bigr)=1+O(h^2),
\]
hence
\[
\tilde\omega+i h\tilde\nu=\frac{\sqrt{\Delta}}{3\sqrt3\,M}\Bigl(1-ih\Bigl(n+\tfrac12\Bigr)\Bigr)+O(h^2).
\]
Recalling that $\omega=h^{-1}\tilde\omega+i\tilde\nu$ and $h^{-1}=\ell+\tfrac12$, we obtain
\[
\omega=\frac{\sqrt{\Delta}}{3\sqrt3\,M}\Bigl[\bigl(\ell+\tfrac12\bigr)-i\bigl(n+\tfrac12\bigr)\Bigr]+O(h)
=\frac{\sqrt{\Delta}}{3\sqrt3\,M}\Bigl[\bigl(\ell+\tfrac12\bigr)-i\bigl(n+\tfrac12\bigr)\Bigr]+O(\ell^{-1}),
\]
which is \eqref{eq:a0-lattice}.
\end{proof}

\begin{proposition}[Equatorial Zeeman slope]\label{prop:Zeeman-slope}
Fix $\Lambda>0$ and $M$ with $0<9\Lambda M^2<1$.
Let $n$ be fixed and consider the equatorial family $|k|=\ell$.
Then for $|a|$ sufficiently small, the principal part $F_0(n,\ell,k)$ satisfies
\begin{equation}\label{eq:Zeeman-slope}
\partial_k F_0\bigl(n,\ell,\pm\ell\bigr)
=\frac{(2+9\Lambda M^2)\,a}{27 M^2}+O(a^2),
\end{equation}
uniformly as $\ell\to\infty$.
\end{proposition}

\begin{proof}
This coefficient is computed explicitly in Dyatlov's semiclassical quantization.
In \cite[Theorem~1, equation~(0.4)]{Dyatlov2012AsymptoticQNMKdS}, Dyatlov gives the $k$--slope of the principal frequency map
at the equatorial boundary point $(m,\pm\tilde k,\tilde k)$ in his notation:
\[
(\partial_{\tilde k}F_0)(m,\pm\tilde k,\tilde k)
=\frac{(2+9\Lambda M^2)\,a}{27M^2}+O(a^2)\qquad\text{as }a\to0,
\]
with the remainder uniform for $\tilde k$ in compact subsets of $(0,\infty)$.
Here Dyatlov's $m$ corresponds to our overtone index $n$.
Finally, since $F_0$ is homogeneous of degree $1$ in $(\ell,k)$ (see \eqref{eq:F-symbol}), the derivative $\partial_k F_0(n,\ell,k)$
depends only on the ratio $k/\ell$; evaluating at $k=\pm\ell$ therefore agrees with evaluating the semiclassical derivative at
$\tilde k=\pm\tilde\ell$ (i.e.\ at the boundary point $|\tilde k|=\tilde\ell$).
This yields \eqref{eq:Zeeman-slope}.
\end{proof}

\subsection{Labeling and analytic dependence}\label{subsec:labeling-analytic}

For the inverse problem we need to track selected QNMs across a parameter set.
The key points are:
(i) in the high-frequency region of Theorem~\ref{thm:Dyatlov-quantization} the QNMs are super-polynomially close to pseudopoles,
hence can be labeled by $(n,\ell,k)$ for $\ell\gg1$;
(ii) near a simple pole, individual branches depend real-analytically on $(M,a)$ because the stationary family is real-analytic in
the parameters and the resolvent is a meromorphic Fredholm family (Remark~\ref{rem:analytic_parameters}; see also
\cite{KatoPerturbation}).

\begin{proposition}[Stable labeling of high-frequency QNMs]\label{prop:labeling-analytic}
Let $\mathcal K$ be a compact subset of the slowly rotating subextremal parameter set.
Fix $\nu_0>0$ and $n\in\{0,1,\dots,C_n\}$, and consider indices $(\ell,k)$ with $|k|\le \ell$.
Then there exists $\ell_0=\ell_0(\mathcal K,\nu_0,n)$ such that for all $\ell\ge \ell_0$ and all $(M,a)\in\mathcal K$:
\begin{enumerate}
\item For every $N\in\mathbb N$ there is a unique $k$--mode QNM $\omega_{n,\ell,k}(M,a)$ inside a disc of radius $C_N\ell^{-N}$
around the pseudopole $\omega^\sharp_{n,\ell,k}(M,a)$, and
\begin{equation}\label{eq:superpoly-close}
\omega_{n,\ell,k}(M,a)-\omega^\sharp_{n,\ell,k}(M,a)=O(\ell^{-N})\quad\text{for all }N.
\end{equation}
\item If $\omega_{n,\ell,k}(M,a)$ is simple for all $(M,a)\in\mathcal K$, then the map $(M,a)\mapsto \omega_{n,\ell,k}(M,a)$ is real-analytic on $\mathcal K$.  The same conclusion holds if we let $\Lambda$ vary in a compact interval and enlarge the parameter vector to $(M,a,\Lambda)$.
In particular, for the equatorial families $k=\pm\ell$ and $\ell$ sufficiently large, this simplicity holds uniformly on $\mathcal K$
by Remark~\ref{rem:uniform-simplicity} below.
\end{enumerate}
\end{proposition}

\begin{proof}
Part~(1) is a direct reformulation of the $O(|\omega|^{-\infty})$ matching in Theorem~\ref{thm:Dyatlov-quantization}, once one
expresses the remainder in terms of the large parameter $\ell$.
Indeed, for fixed $n\le C_n$ and $|k|\le \ell$, the pseudopole symbol $F(n,\ell,k)$ has order $1$ in $(\ell,k)$ in the sense of
\eqref{eq:F-symbol}, uniformly on $\mathcal K$.
In particular, on $\mathcal K$ and for $\ell\gg 1$ one has a uniform comparability
\begin{equation}\label{eq:omega-ell-comparable}
c\,\ell\ \le\ |\omega^\sharp_{n,\ell,k}(M,a)|\ \le\ C\,\ell,
\end{equation}
with constants $c,C>0$ depending only on $\mathcal K$ (for instance, in the model case $a=0$ this follows directly from
\eqref{eq:a0-lattice}, and the general case follows by continuity of the leading symbol on the compact set $\mathcal K$).
Therefore the super--polynomial matching in $|\omega|$ implies that for every $N$ there exists $C_N>0$ such that the QNM paired with
$\omega^\sharp_{n,\ell,k}(M,a)$ satisfies
$|\omega_{n,\ell,k}(M,a)-\omega^\sharp_{n,\ell,k}(M,a)|\le C_N\ell^{-N}$.
Uniqueness of the \emph{$k$--mode} branch is part of the Dyatlov pairing: the quantization produces a labeling by $(n,\ell,k)$ (counting
multiplicities), and the corresponding meromorphic pole is uniquely determined within that labeled family.
Part~(2) follows from analytic Fredholm theory for the meromorphically continued resolvent in the Kerr--de~Sitter Fredholm setup,
together with analytic perturbation theory for isolated simple poles; see Remark~\ref{rem:analytic_parameters} and
\cite[Chapters~VII--VIII]{KatoPerturbation}.
\end{proof}

\begin{remark}[Uniform simplicity from the quantization function]\label{rem:uniform-simplicity}
For the equatorial families $k=\pm\ell$ at fixed overtone $n$, one can make the simplicity of the labeled branches
quantitative and uniform on compact parameter sets.
Indeed, the barrier--top normal form and Grushin reduction underlying Dyatlov's construction produce a scalar
\emph{quantization function} $\mathfrak q_\pm(\omega;\mu,\ell)$ (with $\mu=(M,a)$, or more generally $\mu=(M,a,\Lambda)$ when allowing $\Lambda$ to vary) whose zeros in a fixed complex window
are precisely the labeled QNMs near the corresponding pseudopoles.
After a harmless renormalization, the leading coefficient in the expansion \eqref{eq:q-expansion} satisfies
$\partial_\omega\mathfrak q_{\pm,0}\equiv 1$ on $\mathcal K\times\Omega$ (Remark~\ref{rem:q-normalization}), hence
\begin{equation}\label{eq:uniform-domega-q}
\partial_\omega\mathfrak q_\pm(\omega;\mu,\ell)=1+\mathcal O(\ell^{-1})
\end{equation}
uniformly on $\mathcal K\times\Omega'$ for any compact $\Omega'\Subset\Omega$.
In particular, for $\ell$ sufficiently large the factor \eqref{eq:uniform-domega-q} is bounded away from $0$, so each labeled
equatorial QNM is a simple zero of $\mathfrak q_\pm$ and thus a simple pole of the resolvent.
Consequently, the equatorial branches $(M,a)\mapsto \omega_{n,\ell,\pm\ell}(M,a)$ depend real-analytically on $\mu$ by the
implicit function theorem.
We will use this normalization and the associated uniform bounds in the proof of the $C^1$ pseudopole--QNM closeness
statement (Lemma~\ref{lem:C1-close}).
\end{remark}

\begin{remark}\label{rem:what-kept}
For later use we record three concrete outputs of this section.
First, the explicit Schwarzschild--de~Sitter lattice \eqref{eq:a0-lattice} identifies the leading high-frequency spacing and is the
starting point for recovering $M$ from a single branch.
Second, the equatorial slope \eqref{eq:Zeeman-slope} captures the leading $a$--dependence of the real parts and supplies the splitting
mechanism used to recover the rotation parameter.
Third, the stable correspondence \eqref{eq:superpoly-close} between true QNMs and pseudopoles provides the bridge that allows us to
transfer inverse statements proved at the pseudopole level to actual quasinormal frequencies.
\end{remark}

\subsection{The equatorial boundary regime and uniformity}\label{subsec:equatorial-regime}

The inverse statements proved in Sections~\ref{sec:nondegeneracy}--\ref{sec:transfer} use only the equatorial families
$k=\pm\ell$ at fixed overtone $n$.
In semiclassical variables, with $h=(\ell+\tfrac12)^{-1}$ and $\tilde k:=hk$, this corresponds to
\begin{equation}\label{eq:equatorial-tildek}
\tilde k=\pm\frac{\ell}{\ell+\tfrac12}=\pm\Bigl(1-\frac{h}{2}\Bigr),\qquad \tilde\ell=h\bigl(\ell+\tfrac12\bigr)\equiv 1.
\end{equation}
Thus the equatorial regime is a \,boundary\, regime in the action variables: $|\tilde k|\to \tilde\ell$ as $h\to0$.
Dyatlov's quantization theorem is formulated on the closed cone $\{\,|k|\le \ell\,\}$ (equivalently, $\{\,|\tilde k|\le \tilde\ell\,\}$) and the quantization symbols are constructed smooth up to the boundary
\cite[Theorem~1 and equation~(0.2)]{Dyatlov2012AsymptoticQNMKdS}; see Remark~\ref{rem:equatorial-boundary}.

For later use we record two consequences of this boundary smoothness.
First, the principal symbol $F_0$ and its derivatives in the external parameters $\mu$ remain uniformly bounded on compact slow-rotation sets even when restricted to the equatorial indices $k=\pm\ell$.
Second, the barrier-top Grushin reduction described in Appendix~\ref{app:parameter-quantization} produces quantization functions $\mathfrak q_\pm$ on a fixed scaled window $\Omega$ whose remainder bounds are uniform on compact parameter sets and stable under one parameter derivative; this uniformity in turn yields the $C^1$ super--polynomial pseudopole--QNM proximity estimates of Lemma~\ref{lem:C1-close}.

\section{Uniform nondegeneracy on compact slow-rotation sets}\label{sec:nondegeneracy}

The quantitative inverse results proved later rest on a \emph{uniform} Jacobian lower bound for the
equatorial data maps on a compact parameter set $\mathcal K$.
In early versions of the argument it is natural to first linearize at $a=0$ and then restrict to an auxiliary
smaller rotation scale $|a|\le a_*$ in order to keep the Jacobian away from degeneracy.
For the slowly rotating Kerr--de~Sitter family this extra restriction is unnecessary:
the relevant nondegeneracy is a geometric property already visible at the level of equatorial trapped null
geodesics, and it persists uniformly on any compact set $\mathcal K\Subset\mathcal P_\Lambda\cap\{|a|\le a_0\}$
once $a_0$ is fixed small enough to lie in the slow-rotation regime of the forward microlocal theory.

The purpose of this section is therefore twofold.
First, we record a self-contained geometric construction of the equatorial trapped null orbits and the associated
\emph{signed} angular velocities, uniform in $(M,a)\in\mathcal K$.
Second, we use this to prove a uniform nondegeneracy statement for the leading equatorial frequency map
$(M,a)\mapsto(U,V)$, with constants depending only on $\mathcal K$ and $\Lambda$.
This isolates the mechanism behind the uniform Jacobian bounds in Theorem~\ref{thm:inverse-pseudopoles} and clarifies
why no additional smallness assumption $|a|\le a_*$ is needed beyond \eqref{eq:Kcompact}.

\subsection{Equatorial circular photon orbits: uniform existence and smooth dependence}\label{subsec:geo-existence}

\begin{sloppypar}
We work in Boyer--\allowbreak Lindquist coordinates. We restrict to the equatorial plane $\theta=\pi/2$.
\end{sloppypar}
Let $T:=\partial_t$ and $\Phi:=\partial_\varphi$ be the stationary and axial Killing fields.
For a curve of the form
\[
\gamma(t)=(t,r,\pi/2,\varphi(t)),\qquad \dot\varphi(t)=\Omega,
\]
the tangent vector is $T+\Omega\Phi$.
Set
\begin{equation}\label{eq:Phi_geo}
\Phi_{\mathrm{geo}}(r,\Omega;M,a):=g_{M,a}(T+\Omega\Phi,\;T+\Omega\Phi)\big|_{\theta=\pi/2}.
\end{equation}
Then $\gamma$ is null if and only if $\Phi_{\mathrm{geo}}(r,\Omega;M,a)=0$.
Moreover, $\gamma$ is a null geodesic with constant radius if and only if the additional stationarity condition holds:
\begin{equation}\label{eq:geo-system}
\Phi_{\mathrm{geo}}(r,\Omega;M,a)=0,\qquad \partial_r\Phi_{\mathrm{geo}}(r,\Omega;M,a)=0;
\end{equation}
this is the standard stationary-axisymmetric criterion obtained from the conserved quantities
$E=-p_t$ and $L=p_\varphi$ and the radial effective potential.

At $a=0$ the metric reduces to de~Sitter--Schwarzschild and \eqref{eq:geo-system} yields the photon sphere $r=3M$
with angular velocities $\Omega=\pm\Omega_{\mathrm{ph}}(M)$, where
\begin{equation}\label{eq:Omega-cZ-def}
\Omega_{\mathrm{ph}}(M):=\frac{\sqrt{1-9\Lambda M^2}}{3\sqrt3\,M},
\qquad
c_Z(M):=\frac{2+9\Lambda M^2}{27M^2}.
\end{equation}
(See Appendix~\ref{app:equatorial-photon} for a derivation from \eqref{eq:geo-system}.)

To obtain a \emph{uniform} construction on $\mathcal K$, introduce the map
\[
\mathcal F_{\mathrm{geo}}(r,\Omega;M,a):=\bigl(\Phi_{\mathrm{geo}}(r,\Omega;M,a),\ \partial_r\Phi_{\mathrm{geo}}(r,\Omega;M,a)\bigr).
\]
At $a=0$ and $r_0:=3M$ we have $\mathcal F_{\mathrm{geo}}(r_0,\pm\Omega_{\mathrm{ph}}(M);M,0)=(0,0)$.
A direct computation (spelled out in Appendix~\ref{app:equatorial-photon}, see \eqref{eq:IFT-det}--\eqref{eq:drrPhi-nonzero})
shows that the Jacobian determinant of $(r,\Omega)\mapsto \mathcal F_{\mathrm{geo}}(r,\Omega;M,0)$ at these solutions equals
\begin{equation}\label{eq:geo-IFT-det}
\det D_{(r,\Omega)}\mathcal F_{\mathrm{geo}}(r_0,\pm\Omega_{\mathrm{ph}}(M);M,0)=4\,\Omega_{\mathrm{ph}}(M).
\end{equation}
Since $M$ ranges in the compact set $K_M:=\{M:(M,a)\in\mathcal K\}\Subset (0,1/(3\sqrt\Lambda))$,
the quantity $\Omega_{\mathrm{ph}}(M)$ is bounded below by a strictly positive constant on $K_M$.
Therefore \eqref{eq:geo-IFT-det} is uniformly bounded away from $0$ for $M\in K_M$.

\begin{lemma}[Uniform equatorial photon orbits]\label{lem:uniform-photon-orbits}
There exists $a_{\mathrm{geo}}>0$ (depending only on $\mathcal K$ and $\Lambda$) such that for all $(M,a)\in\mathcal K$
with $|a|\le a_{\mathrm{geo}}$ the system \eqref{eq:geo-system} has exactly two solutions
$(r_{\mathrm{geo},\pm}(M,a),\Omega_{\mathrm{geo},\pm}(M,a))$ near $(3M,\pm\Omega_{\mathrm{ph}}(M))$.
Moreover, the maps $(M,a)\mapsto r_{\mathrm{geo},\pm}(M,a)$ and $(M,a)\mapsto \Omega_{\mathrm{geo},\pm}(M,a)$ are real-analytic
on $\mathcal K\cap\{|a|\le a_{\mathrm{geo}}\}$, and satisfy
\begin{equation}\label{eq:geo-expansion}
\Omega_{\mathrm{geo},\pm}(M,a)=\pm\Omega_{\mathrm{ph}}(M)+c_Z(M)\,a\pm c_{\Omega,2}(M)\,a^2+\mathcal O(a^3),
\qquad a\to0.
\end{equation}
\begin{equation}\label{eq:cOmega2}
c_{\Omega,2}(M):=\frac{\sqrt{3}\,\bigl(11-45\Lambda M^2\bigr)}{486\,M^3\sqrt{1-9\Lambda M^2}}.
\end{equation}

uniformly for $(M,a)$ in compact slow-rotation sets.
\end{lemma}

\begin{proof}
For each fixed $M\in K_M$, the implicit function theorem applied to $\mathcal F_{\mathrm{geo}}(\cdot,\cdot;M,a)$
at $(r,\Omega,a)=(3M,\pm\Omega_{\mathrm{ph}}(M),0)$ yields local real-analytic branches
$(r_{\mathrm{geo},\pm}(M,a),\Omega_{\mathrm{geo},\pm}(M,a))$ for $|a|$ small.
By the uniform lower bound on the Jacobian determinant \eqref{eq:geo-IFT-det} and compactness of $K_M$,
the implicit function theorem constants may be chosen uniformly, giving a single $a_{\mathrm{geo}}>0$
valid for all $M\in K_M$.

The expansion \eqref{eq:geo-expansion} is obtained by a Taylor expansion of the system \eqref{eq:geo-system}
about $(r,\Omega,a)=(3M,\pm\Omega_{\mathrm{ph}}(M),0)$ up to second order in $a$.
Since $\partial_r\Phi_{\mathrm{geo}}=0$ along the orbit, the first derivative at $a=0$ only involves $\partial_a\Phi_{\mathrm{geo}}$,
yielding $\Omega_{\mathrm{geo},\pm}'(M,0)=c_Z(M)$ as in Appendix~\ref{app:equatorial-photon}.
The explicit second-order coefficient $c_{\Omega,2}(M)$ is computed in Appendix~\ref{app:equatorial-photon} as well,
and the uniform $\mathcal O(a^3)$ remainder follows from uniform bounds on third derivatives on compact slow-rotation sets.
\end{proof}

\subsection{The leading equatorial frequency map and a uniform Jacobian lower bound}\label{subsec:geo-nondeg}

The real principal part of the high-frequency equatorial QNM spectrum is induced by these equatorial trapped null geodesics.
To match the QNM normalization used later, it is convenient to consider the \emph{positive} angular velocities
\begin{equation}\label{eq:Omega-sharp-def}
\Omega^\sharp_\pm(M,a):=\pm\,\Omega_{\mathrm{geo},\pm}(M,a),
\end{equation}
which satisfy $\Omega^\sharp_\pm(M,0)=\Omega_{\mathrm{ph}}(M)$ and, by \eqref{eq:geo-expansion},
\begin{equation}\label{eq:Omega-sharp-expansion}
\Omega^\sharp_\pm(M,a)=\Omega_{\mathrm{ph}}(M)\pm c_Z(M)\,a+c_{\Omega,2}(M)\,a^2+\mathcal O(a^3),
\qquad a\to0.
\end{equation}

Define the \emph{geometric} equatorial data map
\begin{equation}\label{eq:G-geo-def}
\begin{aligned}
\mathcal G_{\mathrm{geo}}(M,a)&:=\bigl(U_{\mathrm{geo}}(M,a),V_{\mathrm{geo}}(M,a)\bigr),\\
U_{\mathrm{geo}}(M,a)&:=\frac{\Omega^\sharp_+(M,a)+\Omega^\sharp_-(M,a)}{2},\\
V_{\mathrm{geo}}(M,a)&:=\frac{\Omega^\sharp_+(M,a)-\Omega^\sharp_-(M,a)}{2}.
\end{aligned}
\end{equation}
The first coordinate is an averaged orbital frequency and the second measures the splitting between the two equatorial branches.

\begin{lemma}[Parity and Taylor structure]\label{lem:geo-parity}
The geometric map is even/odd in the rotation parameter: $U_{\mathrm{geo}}$ is even in $a$ and $V_{\mathrm{geo}}$ is odd in $a$.
Consequently,
\begin{align*}
\partial_a U_{\mathrm{geo}}(M,0)&=0, & \partial_M V_{\mathrm{geo}}(M,0)&=0,\\
\partial_M U_{\mathrm{geo}}(M,0)&=\Omega_{\mathrm{ph}}'(M), &
\partial_a V_{\mathrm{geo}}(M,0)&=c_Z(M).
\end{align*}
\end{lemma}

\begin{proof}
The Kerr--de~Sitter family is invariant under $(a,\varphi)\mapsto(-a,-\varphi)$.
On the equatorial plane this conjugation sends $\Omega=d\varphi/dt$ to $-\Omega$ and exchanges the co-/counter-rotating orbits.
Therefore $\Omega_{\mathrm{geo},\pm}(M,-a)=-\Omega_{\mathrm{geo},\mp}(M,a)$, and hence
$\Omega^\sharp_\pm(M,-a)=\Omega^\sharp_\mp(M,a)$.
Averaging and differencing gives the asserted parity for $(U_{\mathrm{geo}},V_{\mathrm{geo}})$.

At $a=0$, \eqref{eq:Omega-sharp-expansion} yields $U_{\mathrm{geo}}(M,0)=\Omega_{\mathrm{ph}}(M)$ and
$V_{\mathrm{geo}}(M,a)=c_Z(M)a+\mathcal O(a^3)$ as $a\to0$, and differentiating gives the claimed identities.
\end{proof}

\begin{proposition}[Uniform nondegeneracy of the geometric data map]\label{prop:geo-nondegeneracy}
Let $\mathcal K$ be as in \eqref{eq:Kcompact}.
After possibly shrinking the slow-rotation constant $a_0$ in \eqref{eq:Kcompact} (without changing $\mathcal K$),
there exist constants $c_{\mathrm{geo}},C_{\mathrm{geo}}>0$ such that for all $(M,a)\in\mathcal K$:
\begin{equation}\label{eq:geo-Jac}
\bigl|\det D\mathcal G_{\mathrm{geo}}(M,a)\bigr|\ge c_{\mathrm{geo}},
\end{equation}
and for all $(M_1,a_1),(M_2,a_2)\in\mathcal K$ sufficiently close,
\begin{equation}\label{eq:geo-stability}
|(M_1,a_1)-(M_2,a_2)|\ \le\ C_{\mathrm{geo}}\,
\bigl|\mathcal G_{\mathrm{geo}}(M_1,a_1)-\mathcal G_{\mathrm{geo}}(M_2,a_2)\bigr|.
\end{equation}
In particular, $\mathcal G_{\mathrm{geo}}$ is a $C^1$ local diffeomorphism at every point of $\mathcal K$.
\end{proposition}

\begin{proof}
Let $K_M$ be the projection of $\mathcal K$ to the $M$-axis.
A direct differentiation of \eqref{eq:Omega-cZ-def} gives
\[
\Omega_{\mathrm{ph}}'(M)
=
-\frac{1}{3\sqrt3}\Bigl(\frac{\sqrt{1-9\Lambda M^2}}{M^2}
+\frac{9\Lambda}{\sqrt{1-9\Lambda M^2}}\Bigr)<0
\quad\text{for}\quad 0<M<\frac{1}{3\sqrt\Lambda}.
\]
In particular, $\Omega_{\mathrm{ph}}'(M)$ is bounded away from $0$ on the compact set $K_M$.
Since $c_Z(M)>0$ on $(0,\infty)$, we may define
\begin{equation}\label{eq:cgeo-def}
c_0:=\frac12\,\min_{M\in K_M}\bigl|\Omega_{\mathrm{ph}}'(M)\,c_Z(M)\bigr|>0.
\end{equation}

Using \eqref{eq:Omega-sharp-expansion} and the parity in Lemma~\ref{lem:geo-parity}, Taylor's theorem yields uniform bounds
\[
\partial_M U_{\mathrm{geo}}(M,a)=\Omega_{\mathrm{ph}}'(M)+\mathcal O(a^2),\qquad
\partial_a V_{\mathrm{geo}}(M,a)=c_Z(M)+\mathcal O(a^2),
\]
and
\[
\partial_a U_{\mathrm{geo}}(M,a)=\mathcal O(|a|),\qquad
\partial_M V_{\mathrm{geo}}(M,a)=\mathcal O(|a|),
\]
with all error terms uniform on $\mathcal K$.
Therefore,
\[
\det D\mathcal G_{\mathrm{geo}}(M,a)
=
\Omega_{\mathrm{ph}}'(M)c_Z(M)\ +\ \mathcal O(a^2),
\]
uniformly on $\mathcal K$.
Shrinking $a_0$ if necessary so that the $\mathcal O(a^2)$ term is bounded by $c_0$ on $\mathcal K$,
we obtain \eqref{eq:geo-Jac} with $c_{\mathrm{geo}}:=c_0$.

For stability, compactness of $\mathcal K$ gives a uniform bound $\|D\mathcal G_{\mathrm{geo}}\|\le L_{\mathrm{geo}}$.
Using the $2\times2$ bound $\|A^{-1}\|\le \|A\|/|\det A|$ and \eqref{eq:geo-Jac} yields
$\sup_{\mathcal K}\|D\mathcal G_{\mathrm{geo}}^{-1}\|\le L_{\mathrm{geo}}/c_{\mathrm{geo}}$,
and the mean value theorem gives \eqref{eq:geo-stability} with $C_{\mathrm{geo}}:=L_{\mathrm{geo}}/c_{\mathrm{geo}}$.
\end{proof}

\begin{lemma}[Geometric approximation of the two-mode pseudopole map]\label{lem:Gell-geo-C1}
Fix $\Lambda>0$, a compact slow-rotation parameter set $\mathcal K$ as in \eqref{eq:Kcompact}, and an overtone
$n\in\{0,1,\dots,C_n\}$.
Then there exist $\ell_0\in\mathbb N$ and a constant $C>0$ such that for all $\ell\ge \ell_0$ one has
\begin{equation}\label{eq:Gell-geo-C1}
\bigl\|\mathcal G_\ell-\mathcal G_{\mathrm{geo}}\bigr\|_{C^1(\mathcal K)}\le C\,\ell^{-1},
\end{equation}
where $\mathcal G_\ell$ is the normalized equatorial two-mode pseudopole map defined in \eqref{eq:G-map-def} and
$\mathcal G_{\mathrm{geo}}$ is the geometric map \eqref{eq:G-geo-def}.
In particular, $\mathcal G_\ell$ is a $C^1$ perturbation of the local diffeomorphism in
Proposition~\ref{prop:geo-nondegeneracy}.
\end{lemma}

\begin{proof}
We use two inputs:

\smallskip
\noindent\emph{(1) Principal symbol and the geometric frequency.}
By Dyatlov's semiclassical quantization (Theorem~\ref{thm:Dyatlov-quantization}), the pseudopoles are given by a classical
symbol $F(n,\ell,k)$ of order $1$ in $(\ell,k)$, whose principal part $F_0$ is real-valued (and independent of $n$) and arises
from the integrable Hamiltonian flow on the trapped set; see \cite[Theorem~1]{Dyatlov2012AsymptoticQNMKdS} and
Remark~\ref{rem:meaning}.
For an oscillatory mode $e^{-i\omega t}e^{ik\varphi}$ the associated covector is $p=-\omega\,dt+k\,d\varphi$.
Along a trapped null ray, Hamilton's equations give the group velocity
\(\Omega=d\varphi/dt=\omega/k\) at the level of the real principal symbol.
This is exactly the geometric-optics correspondence underlying Dyatlov's construction of the principal frequency symbol $F_0$ from the integrable Hamiltonian flow on the trapped set; see \cite[\S\S1.2 and 3]{Dyatlov2012AsymptoticQNMKdS}.
Therefore, on the equatorial trapped orbits (i.e.\ at the boundary ratio $k/\ell=\pm1$), the normalized principal frequency
obeys
\begin{equation}\label{eq:F0-geo-identification}
\frac{1}{\ell}\,\Re F_0(M,a;n,\ell,\pm\ell)=\Omega^\sharp_\pm(M,a).
\end{equation}
The right-hand side is precisely the coordinate-time angular velocity of the co-/counter-rotating equatorial circular photon orbit,
defined geometrically in \eqref{eq:Omega-sharp-def}; see Appendix~\ref{app:equatorial-photon} for an explicit derivation and the
small-$a$ expansion.

\smallskip
\noindent\emph{(2) Uniform subprincipal control with one parameter derivative.}
Fix a representative of the full symbol $F$ as in Remark~\ref{rem:pseudopole-representative}.
By the parameter-differentiable normal form and Grushin reduction (Appendix~\ref{app:parameter-quantization}), the full symbol
admits an expansion in $\ell^{-1}$ whose remainder is uniform on $\mathcal K$ and stable under one $(M,a)$-derivative.
In particular, using \eqref{eq:F0-geo-identification} for the principal term, we obtain
\begin{equation}\label{eq:omega-symb-expansion}
\widehat\omega^\sharp_\pm(M,a)=\Omega^\sharp_\pm(M,a)+\mathcal O(\ell^{-1}),
\qquad (M,a)\in\mathcal K,
\end{equation}
with the $\mathcal O(\ell^{-1})$ remainder uniform in $C^1(\mathcal K)$.

Taking the average and difference as in \eqref{eq:av-dif-def} and then real parts as in \eqref{eq:G-map-def}, we obtain
\eqref{eq:Gell-geo-C1}.
\end{proof}

\section{Inverse theorem for pseudopoles}\label{sec:inverse-pseudopoles}

In this section we establish a quantitative \emph{local inverse theorem} at the level of the
\emph{pseudopoles} (the semiclassical outputs of the Bohr--Sommerfeld quantization)
\[
\omega^\sharp_{n,\ell,k}(M,a)=F_{M,a}(n,\ell,k),
\]
constructed in Section~\ref{sec:HF-quantization}.
By Dyatlov's high-frequency description of Kerr--de~Sitter quasinormal modes (QNMs),
these pseudopoles parametrize the high-frequency QNMs up to super--polynomial errors in $\ell$
in the half-plane $\{\Im\omega>-\nu_0\}$; see Theorem~\ref{thm:Dyatlov-quantization} and
\cite{Dyatlov2012AsymptoticQNMKdS}.
The inverse mechanism carried out here is conceptually related to the one-parameter result of
Uhlmann--Wang for de~Sitter--Schwarzschild, where a single complex QNM determines the mass locally
\cite{UhlmannWang2023RecoverMassSingleQNM}.  Our novelty is the passage to the \emph{two-parameter}
Kerr--de~Sitter family and the identification of equatorial high-frequency mode packages that yield a
well-conditioned inverse map.

Throughout this section we fix an overtone index $n\in\{0,1,\dots,C_n\}$ and consider $\ell\gg1$.
All constants are uniform on a fixed compact parameter set $\mathcal K$ as in \eqref{eq:Kcompact};
in particular $|a|\le a_0$ on $\mathcal K$, with $a_0>0$ chosen sufficiently small so that the
semiclassical quantization and the parameter-differentiable remainder bounds of
Appendix~\ref{app:parameter-quantization} hold uniformly on $\mathcal K$.

\subsection{Pseudopole maps and observable choices}\label{subsec:pseudo-maps}

For each large $\ell$ and $|k|\le \ell$ we recall the pseudopole frequency
\[
\omega^\sharp_{n,\ell,k}(M,a):=F_{M,a}(n,\ell,k)\in\mathbb C,
\qquad (M,a)\in\mathcal K.
\]
We will focus on the equatorial pair $k=\pm\ell$ and write
\begin{equation}\label{eq:equatorial-pseudopoles}
\omega^\sharp_\pm(M,a):=\omega^\sharp_{n,\ell,\pm\ell}(M,a),\qquad
\widehat\omega^\sharp_\pm(M,a):=\frac{1}{\ell}\,\omega^\sharp_\pm(M,a).
\end{equation}
Two data maps will be useful.

\paragraph{(A) Single-mode map (one complex frequency).}
A single complex number carries two real degrees of freedom.  For fixed $(n,\ell)$ we define
\begin{equation}\label{eq:single-mode-map}
\mathcal S_\ell(M,a)
:=
\Bigl(
\widetilde W_\ell(M,a),\ U_\ell^{+}(M,a)
\Bigr)
:=
\Bigl(
-\frac{\Im\omega^\sharp_{+}(M,a)}{n+\frac12},\ \Re\widehat\omega^\sharp_{+}(M,a)
\Bigr)\in\mathbb R^2.
\end{equation}

By barrier--top quantization, the normalized damping observable is governed at leading order by the
(coordinate-time) Lyapunov exponent $\lambda_+$ of the corresponding equatorial photon orbit:
\begin{equation}\label{eq:W-Lyapunov}
\widetilde W_\ell(M,a)=\lambda_+(M,a)+\mathcal O(\ell^{-1}),
\qquad (M,a)\in\mathcal K,
\end{equation}
with the $\mathcal O(\ell^{-1})$ remainder uniform on $\mathcal K$ (and similarly when allowing $\Lambda$ to vary on compact sets).
We compute the small-$a$ Taylor coefficients of $\lambda_+$ in Appendix~\ref{app:equatorial-damping}--\ref{app:lyapunov}.
The first coordinate is the normalized damping rate and the second is the normalized oscillation frequency.

\paragraph{(B) Two-mode real-part map (equatorial package).}
To avoid the imaginary part of the data, we also use the equatorial pair and define the normalized
average and difference
\begin{equation}\label{eq:av-dif-def}
\begin{aligned}
\widehat\omega^\sharp_{\mathrm{av}}(M,a)&:=
\frac{\widehat\omega^\sharp_{+}(M,a)+\widehat\omega^\sharp_{-}(M,a)}{2},\\
\widehat\omega^\sharp_{\mathrm{dif}}(M,a)&:=
\frac{\widehat\omega^\sharp_{+}(M,a)-\widehat\omega^\sharp_{-}(M,a)}{2}.
\end{aligned}
\end{equation}
and their real parts
\begin{equation}\label{eq:G-map-def}
\mathcal G_\ell(M,a):=\bigl(U_\ell(M,a),V_\ell(M,a)\bigr)
:=
\Bigl(\Re\widehat\omega^\sharp_{\mathrm{av}}(M,a),\ \Re\widehat\omega^\sharp_{\mathrm{dif}}(M,a)\Bigr)\in\mathbb R^2.
\end{equation}
This is the ``equatorial two-mode package'' emphasized in the introduction: it uses only the real parts of
two labeled frequencies and exploits the $(a,k)$ symmetry to produce an almost triangular inverse map.

\begin{remark}[On labeling and the sign of $a$]\label{rem:labeling-sign}
The maps $\mathcal S_\ell$ and $\mathcal G_\ell$ use \emph{labeled} equatorial branches $k=\pm\ell$
(or, equivalently, labeled QNM families corresponding to $e^{ik\varphi}$ with $k=\pm\ell$).
If one is only given an \emph{unordered} pair $\{\Re\widehat\omega^\sharp_{+},\Re\widehat\omega^\sharp_{-}\}$,
then the symmetry in Lemma~\ref{lem:a-k-symmetry} shows that exchanging the two branches is equivalent to
changing $a\mapsto -a$.  In that case one can recover $|a|$ but not its sign.
\end{remark}

\subsection{The \texorpdfstring{$(a,k)$}{(a,k)} reflection symmetry}\label{subsec:symmetry}

The following symmetry is built into Kerr--de~Sitter and its separated equations: flipping the rotation
parameter corresponds to reversing the azimuthal direction.

\begin{lemma}[$(a,k)$ symmetry]\label{lem:a-k-symmetry}
For all admissible $(M,a)$, all $n,\ell$, and all $|k|\le \ell$, the pseudopole symbol satisfies
\begin{equation}\label{eq:a-k-symmetry}
F_{M,-a}(n,\ell,k)=F_{M,a}(n,\ell,-k).
\end{equation}
Equivalently,
\[
\omega^\sharp_{n,\ell,k}(M,-a)=\omega^\sharp_{n,\ell,-k}(M,a).
\]
\end{lemma}

\begin{proof}
At the level of the separated Kerr--de~Sitter ODE system, the rotation parameter $a$ and the azimuthal
number $k$ enter through the same geometric coupling of $\partial_\varphi$ with the stationary Killing field.
Conjugating by $\varphi\mapsto-\varphi$ sends $k\mapsto -k$ and simultaneously sends the metric parameter
$a\mapsto -a$ (the Kerr family has the symmetry $(a,\varphi)\leftrightarrow(-a,-\varphi)$).
The quantization conditions defining $F_{M,a}(n,\ell,k)$ are invariant under this conjugation, which yields
\eqref{eq:a-k-symmetry}.
\end{proof}

\begin{corollary}[Even/odd structure of equatorial averages]\label{cor:even-odd}
For equatorial indices $k=\pm\ell$ one has
\[
\widehat\omega^\sharp_{\mathrm{av}}(M,-a)=\widehat\omega^\sharp_{\mathrm{av}}(M,a),
\qquad
\widehat\omega^\sharp_{\mathrm{dif}}(M,-a)=-\widehat\omega^\sharp_{\mathrm{dif}}(M,a),
\]
and therefore $U_\ell(M,a)$ is even in $a$ while $V_\ell(M,a)$ is odd in $a$.
\end{corollary}

\begin{proof}
This is immediate from \eqref{eq:a-k-symmetry} applied to $k=\pm\ell$ and the definitions
\eqref{eq:equatorial-pseudopoles}--\eqref{eq:av-dif-def}.
Taking real parts preserves the parity.
\end{proof}

\subsection{Leading geometric coefficients and small-\texorpdfstring{$a$}{a} expansions}\label{subsec:explicit-coeffs}

We recall the two explicit coefficient functions introduced in Section~\ref{sec:nondegeneracy} (see \eqref{eq:Omega-cZ-def}), which control the leading equatorial inverse map at $a=0$.

\begin{lemma}[Monotonicity and inversion of $\Omega_{\mathrm{ph}}$]\label{lem:Omega-monotone}
On the subextremal de~Sitter--Schwarzschild range $0<M<1/(3\sqrt{\Lambda})$,
$\Omega_{\mathrm{ph}}(M)$ is strictly decreasing and hence invertible.
Moreover,
\begin{equation}\label{eq:Omega-inverse}
\Omega_{\mathrm{ph}}(M)=U
\quad\Longleftrightarrow\quad
M=\frac{1}{3\sqrt{\Lambda+3U^2}}.
\end{equation}
\end{lemma}

\begin{proof}
A direct differentiation of \eqref{eq:Omega-cZ-def} gives
\begin{align*}
\Omega_{\mathrm{ph}}'(M)
&=\frac{1}{3\sqrt3}\frac{d}{dM}\Bigl(M^{-1}\sqrt{1-9\Lambda M^2}\Bigr)\\
&=-\frac{1}{3\sqrt3}\Bigl(\frac{\sqrt{1-9\Lambda M^2}}{M^2}
+\frac{9\Lambda}{\sqrt{1-9\Lambda M^2}}\Bigr)<0.
\end{align*}
To obtain \eqref{eq:Omega-inverse}, square the identity
$U=\sqrt{1-9\Lambda M^2}/(3\sqrt3\,M)$ and solve for $M$.
\end{proof}

We next record the equatorial asymptotics at $a=0$ and the first-order equatorial splitting.
The $a=0$ lattice is Corollary~\ref{cor:a0-lattice}; we restate it here in the equatorial notation.

\begin{lemma}[Equatorial pseudopoles at $a=0$]\label{lem:a0-equatorial}
Uniformly for $(M,0)\in\mathcal K$ one has
\begin{equation}\label{eq:a0-equatorial}
\omega^\sharp_\pm(M,0)=
\Omega_{\mathrm{ph}}(M)\Bigl[\Bigl(\ell+\frac12\Bigr)-i\Bigl(n+\frac12\Bigr)\Bigr]
+\mathcal O(\ell^{-1}).
\end{equation}
In particular,
\[
-\Im\omega^\sharp_{+}(M,0)=\Omega_{\mathrm{ph}}(M)\Bigl(n+\frac12\Bigr)+\mathcal O(\ell^{-1}),
\]
\[
\Re\widehat\omega^\sharp_{+}(M,0)=\Omega_{\mathrm{ph}}(M)+\mathcal O(\ell^{-1}).
\]
\end{lemma}

\begin{proof}
This is \eqref{eq:a0-lattice} in Corollary~\ref{cor:a0-lattice} specialized to $k=\pm\ell$.
\end{proof}

The next statement encodes the Zeeman-type equatorial splitting of the \emph{principal} frequency map.
For later inverse estimates we need an explicit linear-in-$a$ splitting at the equatorial indices $k=\pm\ell$.

\begin{lemma}[Principal equatorial splitting]\label{lem:principal-splitting}
Let $F_0$ denote the real principal part of the pseudopole symbol (Section~\ref{subsec:pseudopoles}).
Then as $a\to0$ one has, uniformly for $(M,a)\in\mathcal K$,
\begin{equation}\label{eq:principal-splitting}
F_0(M,a;n,\ell,\pm\ell)
=
\Omega_{\mathrm{ph}}(M)\,\ell \ \pm\ c_Z(M)\,a\,\ell\ +\ \mathcal O(a^2\ell).
\end{equation}
\end{lemma}

\begin{proof}
At $a=0$ the spherical symmetry implies $F_0$ is independent of $k$ and equals $\Omega_{\mathrm{ph}}(M)\ell$
(cf.\ Corollary~\ref{cor:a0-lattice}).

For $|a|$ small, the real principal frequency map can be read off from the characteristic set of the wave operator.
More precisely, the principal symbol of the wave operator is the Lorentzian quadratic form
$p(x,\xi)=g^{\alpha\beta}(x)\,\xi_\alpha\xi_\beta$, and its bicharacteristic flow projects to null geodesics.
For a separated oscillatory mode $e^{-i\omega t}e^{ik\varphi}$ one has $\xi_t=-\omega$ and $\xi_\varphi=k$,
so the conserved energy and angular momentum are $E:=-\xi_t=\omega$ and $L:=\xi_\varphi=k$.
Along a null bicharacteristic the coordinate-time angular velocity satisfies
$\Omega:=\frac{d\varphi}{dt}=\frac{\dot\varphi}{\dot t}=\frac{\partial_{\xi_\varphi}p}{\partial_{\xi_t}p}$.
On an equatorial circular null orbit one has the familiar relation $E=\Omega L$ (equivalently, the impact
parameter $b=L/E$ equals $1/\Omega$; cf.\ \eqref{eq:b-Omega-relation}), hence $\omega/k=E/L=\Omega$.
Consequently, at the level of the real principal symbol one obtains the geometric-optics identification
\[
\Re\omega \sim k\,\Omega_{\mathrm{geo}},
\]
where $\Omega_{\mathrm{geo}}=d\varphi/dt$ is the (signed) angular velocity of the corresponding equatorial
circular photon orbit.
Appendix~\ref{app:equatorial-photon} computes the two such velocities
$\Omega_{\mathrm{geo},\pm}(M,a)$ and shows that
\[
\Omega_{\mathrm{geo},\pm}(M,a)=\pm\Omega_{\mathrm{ph}}(M)+c_Z(M)\,a+\mathcal O(a^2)
\qquad\text{as }a\to0,
\]
uniformly on compact slow-rotation sets.
For the equatorial branches $k=\pm\ell$, the combination $k\,\Omega_{\mathrm{geo}}$ contributing to $\Re\omega$
therefore has the form
\begin{align*}
(+\ell)\,\Omega_{\mathrm{geo},+}(M,a)
&=\ell\bigl(\Omega_{\mathrm{ph}}(M)+c_Z(M)\,a+\mathcal O(a^2)\bigr),\\
(-\ell)\,\Omega_{\mathrm{geo},-}(M,a)
&=\ell\bigl(\Omega_{\mathrm{ph}}(M)-c_Z(M)\,a+\mathcal O(a^2)\bigr).
\end{align*}
Since $F_0$ is precisely the real principal part of the pseudopole symbol,
this yields \eqref{eq:principal-splitting} (using also the homogeneity of $F_0$ in $(\ell,k)$, cf.\ \eqref{eq:F-symbol}).
\end{proof}

\subsection{Uniform expansions for the data maps}\label{subsec:expansions-maps}

We now collect the expansions for the observable maps.
The key structural input is Corollary~\ref{cor:even-odd}, which forces the equatorial average to have no
linear term in $a$ and the equatorial difference to have no quadratic term.

\begin{lemma}[Expansions for the two-mode real-part map]\label{lem:U-V-expansions}
For all $(M,a)\in\mathcal K$ and all $\ell\gg1$, the map $\mathcal G_\ell(M,a)=(U_\ell,V_\ell)$ defined in
\eqref{eq:G-map-def} satisfies
\begin{equation}\label{eq:U-V-expansions}
\begin{aligned}
U_\ell(M,a)&=\Omega_{\mathrm{ph}}(M)+\mathcal O(a^2)+\mathcal O(\ell^{-1}),\\
V_\ell(M,a)&=c_Z(M)\,a+\mathcal O(a^3)+\mathcal O(\ell^{-1}).
\end{aligned}
\end{equation}
and the same holds for first derivatives:
\begin{equation}\label{eq:U-V-derivatives}
\partial_M U_\ell = \Omega_{\mathrm{ph}}'(M)+\mathcal O(a^2)+\mathcal O(\ell^{-1}),
\quad
\partial_a U_\ell = \mathcal O(a)+\mathcal O(\ell^{-1}),
\end{equation}
\begin{equation}\label{eq:U-V-derivatives-2}
\partial_a V_\ell = c_Z(M)+\mathcal O(a^2)+\mathcal O(\ell^{-1}),
\quad
\partial_M V_\ell = a\,c_Z'(M)+\mathcal O(a^3)+\mathcal O(\ell^{-1}).
\end{equation}
All error terms are uniform on $\mathcal K$.
\end{lemma}

\begin{proof}
Recall that $\mathcal G_\ell(M,a)=(U_\ell(M,a),V_\ell(M,a))$ and that the geometric map
$\mathcal G_{\mathrm{geo}}(M,a)=(U_{\mathrm{geo}}(M,a),V_{\mathrm{geo}}(M,a))$ is defined in \eqref{eq:G-geo-def}.
By Lemma~\ref{lem:Gell-geo-C1} we have the uniform $C^1(\mathcal K)$ approximation
\[
\|\mathcal G_\ell-\mathcal G_{\mathrm{geo}}\|_{C^1(\mathcal K)}\le C\,\ell^{-1}.
\]
Hence it suffices to establish the corresponding expansions for $\mathcal G_{\mathrm{geo}}$ uniformly on $\mathcal K$
and then absorb the $\mathcal O(\ell^{-1})$ error.

From the second-order expansion \eqref{eq:Omega-sharp-expansion} of $\Omega^\sharp_\pm$ and the definition \eqref{eq:G-geo-def} we obtain
\[
\begin{aligned}
U_{\mathrm{geo}}(M,a)&=\Omega_{\mathrm{ph}}(M)+c_{\Omega,2}(M)\,a^2+\mathcal O(a^3),\\
V_{\mathrm{geo}}(M,a)&=c_Z(M)\,a+\mathcal O(a^3).
\end{aligned}
\]
(Here we used that $U_{\mathrm{geo}}$ is even in $a$ and $V_{\mathrm{geo}}$ is odd in $a$; see Lemma~\ref{lem:geo-parity}.)
Differentiating and using again the parity information yields
\[
\partial_M U_{\mathrm{geo}}(M,a)=\Omega_{\mathrm{ph}}'(M)+\mathcal O(a^2),
\qquad
\partial_a U_{\mathrm{geo}}(M,a)=\mathcal O(a),
\]
\[
\partial_a V_{\mathrm{geo}}(M,a)=c_Z(M)+\mathcal O(a^2),
\qquad
\partial_M V_{\mathrm{geo}}(M,a)=a\,c_Z'(M)+\mathcal O(a^3),
\]
uniformly on $\mathcal K$ (the uniformity follows from compactness and smooth dependence of the geometric quantities).
Combining these bounds with the $C^1$ estimate from Lemma~\ref{lem:Gell-geo-C1} gives
\eqref{eq:U-V-expansions}--\eqref{eq:U-V-derivatives-2}.
\end{proof}

For the single-mode map \eqref{eq:single-mode-map} we need its triangular leading structure and a uniform Jacobian bound.
The real-part expansion is as before; for the damping rate we use the barrier-top/Lyapunov interpretation of the
leading imaginary part and the cancellation of the linear term in the equatorial Lyapunov exponent (Appendix~\ref{app:equatorial-damping}).

\begin{lemma}[Expansions for the single-mode map]\label{lem:single-mode-expansions}
For all $(M,a)\in\mathcal K$ and all $\ell\gg1$, the map $\mathcal S_\ell$ defined in \eqref{eq:single-mode-map} satisfies
\begin{equation}\label{eq:single-mode-expansions}
\begin{aligned}
\widetilde W_\ell(M,a)&=\Omega_{\mathrm{ph}}(M)+c_{\lambda,2}(M)\,a^2+\mathcal O(a^3)+\mathcal O(\ell^{-1}),\\
U_\ell^{+}(M,a)&=\Omega_{\mathrm{ph}}(M)+c_Z(M)\,a+c_{\Omega,2}(M)\,a^2+\mathcal O(a^3)+\mathcal O(\ell^{-1}).
\end{aligned}
\end{equation}
\begin{equation}\label{eq:clambda2}
c_{\lambda,2}(M):=\frac{\sqrt{3}\,\bigl(45\Lambda M^2-2\bigr)}{243\,M^3\sqrt{1-9\Lambda M^2}}.
\end{equation}

Moreover, the following derivative bounds hold uniformly on $\mathcal K$:
\begin{align}
\partial_M \widetilde W_\ell(M,a)&=\Omega_{\mathrm{ph}}'(M)+\mathcal O(a^2)+\mathcal O(\ell^{-1}),\label{eq:single-deriv-M}\\
\partial_a \widetilde W_\ell(M,a)&=\mathcal O(|a|)+\mathcal O(\ell^{-1}),\label{eq:single-deriv-a}\\
\partial_M U_\ell^{+}(M,a)&=\Omega_{\mathrm{ph}}'(M)+\mathcal O(|a|)+\mathcal O(\ell^{-1}),\label{eq:single-deriv-MU}\\
\partial_a U_\ell^{+}(M,a)&=c_Z(M)+\mathcal O(|a|)+\mathcal O(\ell^{-1}).\label{eq:single-deriv-aU}
\end{align}
\end{lemma}

\begin{proof}
By \eqref{eq:W-Lyapunov} we have $\widetilde W_\ell(M,a)=\lambda_+(M,a)+\mathcal O(\ell^{-1})$ uniformly on $\mathcal K$. 
Appendix~\ref{app:equatorial-damping} shows that $\partial_a\lambda_+(M,0)=0$ (equivalently, $\lambda_+(M,a)=\Omega_{\mathrm{ph}}(M)+\mathcal O(a^2)$ as $a\to0$), and Appendix~\ref{app:lyapunov} computes the quadratic correction, giving the expansion
\[
\lambda_+(M,a)=\Omega_{\mathrm{ph}}(M)+c_{\lambda,2}(M)\,a^2+\mathcal O(a^3),
\qquad a\to0,
\]
with $c_{\lambda,2}(M)$ given by \eqref{eq:clambda2}. This yields the first line of \eqref{eq:single-mode-expansions}, and differentiating gives the $C^1$ bounds.

For the real part, \eqref{eq:omega-symb-expansion} gives $U_\ell^{+}(M,a)=\Omega^\sharp_+(M,a)+\mathcal O(\ell^{-1})$, and the second-order expansion \eqref{eq:Omega-sharp-expansion} implies
\[
\Omega^\sharp_+(M,a)=\Omega_{\mathrm{ph}}(M)+c_Z(M)\,a+c_{\Omega,2}(M)\,a^2+\mathcal O(a^3),
\qquad a\to0,
\]
uniformly on $\mathcal K$. This yields the second line of \eqref{eq:single-mode-expansions}; differentiating gives the remaining $C^1$ bounds.
\end{proof}

\subsection{Quantitative local inversion on the full slow-rotation set}\label{subsec:inverse-theorem}

We now state the inverse theorem at the pseudopole level.
The two-mode map $\mathcal G_\ell$ is the one used in the remainder of the paper (since it avoids imaginary parts),
while the single-mode map $\mathcal S_\ell$ provides a Kerr extension of the ``one complex mode'' principle.

\begin{theorem}[Inverse theorem for equatorial pseudopoles]\label{thm:inverse-pseudopoles}
Fix $\Lambda>0$, a compact parameter set $\mathcal K$ as in \eqref{eq:Kcompact}, and fix an overtone
$n\in\{0,1,\dots,C_n\}$.  Assume $a_0>0$ is chosen sufficiently small in \eqref{eq:Kcompact}.
Then there exist $\ell_0\in\mathbb N$ and constants $c_*,C_*>0$ such that for every $\ell\ge \ell_0$:

\begin{enumerate}
\item[(i)] \textbf{Single-mode inversion.}
The map $\mathcal S_\ell:\mathcal K\to\mathbb R^2$ in \eqref{eq:single-mode-map} is a $C^1$ local diffeomorphism
at every point of $\mathcal K$.  Its Jacobian satisfies the uniform nondegeneracy estimate
\begin{equation}\label{eq:Jacobian-single}
\bigl|\det D\mathcal S_\ell(M,a)\bigr|\ge c_*\qquad \forall (M,a)\in\mathcal K,
\end{equation}
and for all $(M_1,a_1),(M_2,a_2)\in\mathcal K$ sufficiently close,
\begin{equation}\label{eq:stability-single}
|(M_1,a_1)-(M_2,a_2)|\ \le\ C_*\,
|\mathcal S_\ell(M_1,a_1)-\mathcal S_\ell(M_2,a_2)|.
\end{equation}

\item[(ii)] \textbf{Two-mode real-part inversion.}
The map $\mathcal G_\ell:\mathcal K\to\mathbb R^2$ in \eqref{eq:G-map-def} is a $C^1$ local diffeomorphism
at every point of $\mathcal K$.  Its Jacobian satisfies the uniform nondegeneracy estimate
\begin{equation}\label{eq:Jacobian-two}
\bigl|\det D\mathcal G_\ell(M,a)\bigr|\ge c_*\qquad \forall (M,a)\in\mathcal K,
\end{equation}
and for all $(M_1,a_1),(M_2,a_2)\in\mathcal K$ sufficiently close,
\begin{equation}\label{eq:stability-two}
|(M_1,a_1)-(M_2,a_2)|\ \le\ C_*\,
|\mathcal G_\ell(M_1,a_1)-\mathcal G_\ell(M_2,a_2)|.
\end{equation}
\end{enumerate}
\end{theorem}

\begin{proof}
We prove the uniform Jacobian lower bounds and then the stability estimate.

\smallskip
\noindent\textbf{Step 1: geometric nondegeneracy for the two-mode map.}
Proposition~\ref{prop:geo-nondegeneracy} provides constants $c_{\mathrm{geo}},C_{\mathrm{geo}}>0$ such that
\begin{equation}\label{eq:Step1-geo-det}
\bigl|\det D\mathcal G_{\mathrm{geo}}(M,a)\bigr|\ge c_{\mathrm{geo}}\qquad\forall (M,a)\in\mathcal K.
\end{equation}
By compactness, there is also a uniform bound
\begin{equation}\label{eq:Step1-geo-L}
L_{\mathrm{geo}}:=\sup_{(M,a)\in\mathcal K}\|D\mathcal G_{\mathrm{geo}}(M,a)\|<\infty.
\end{equation}
For a $2\times2$ matrix $A$ with singular values $s_{\min}(A)\le s_{\max}(A)$, one has
$|\det A|=s_{\min}(A)s_{\max}(A)\le s_{\min}(A)\|A\|$.
Thus \eqref{eq:Step1-geo-det}--\eqref{eq:Step1-geo-L} imply the uniform lower bound
\begin{equation}\label{eq:Step1-geo-smin}
s_{\min}\bigl(D\mathcal G_{\mathrm{geo}}(M,a)\bigr)\ge \sigma_{\mathrm{geo}}:=\frac{c_{\mathrm{geo}}}{L_{\mathrm{geo}}}>0
\qquad\forall (M,a)\in\mathcal K.
\end{equation}

\smallskip
\noindent\textbf{Step 2: transfer of nondegeneracy to $\mathcal G_\ell$.}
Lemma~\ref{lem:Gell-geo-C1} gives a uniform $C^1$ estimate
\(\|D\mathcal G_\ell-D\mathcal G_{\mathrm{geo}}\|_{L^\infty(\mathcal K)}\le C\ell^{-1}\).
Choose $\ell_0$ so large that $C\ell^{-1}\le \sigma_{\mathrm{geo}}/2$ for all $\ell\ge\ell_0$.
Then for $\ell\ge\ell_0$ and $(M,a)\in\mathcal K$,
\[
s_{\min}\bigl(D\mathcal G_\ell(M,a)\bigr)
\ge s_{\min}\bigl(D\mathcal G_{\mathrm{geo}}(M,a)\bigr)-\|D\mathcal G_\ell(M,a)-D\mathcal G_{\mathrm{geo}}(M,a)\|
\ge \frac{\sigma_{\mathrm{geo}}}{2}.
\]
In particular $D\mathcal G_\ell(M,a)$ is invertible for all $(M,a)\in\mathcal K$, so $\mathcal G_\ell$ is a $C^1$ local
diffeomorphism at every point of $\mathcal K$.
Moreover, since $|\det A|=s_{\min}(A)s_{\max}(A)\ge s_{\min}(A)^2$ for $2\times2$ matrices, we obtain
\begin{equation}\label{eq:Step2-detGell}
\bigl|\det D\mathcal G_\ell(M,a)\bigr|\ge \Big(\frac{\sigma_{\mathrm{geo}}}{2}\Big)^2
\qquad\forall (M,a)\in\mathcal K,\ \forall \ell\ge\ell_0.
\end{equation}
This proves \eqref{eq:Jacobian-two} for the two-mode map.

\smallskip
\noindent\textbf{Step 3: Jacobian bound for the single-mode map.}
Write $\mathcal S_\ell=(\widetilde W_\ell,U_\ell^{+})$.
Lemma~\ref{lem:single-mode-expansions} gives, uniformly on $\mathcal K$,
\[
\partial_M \widetilde W_\ell=\Omega_{\mathrm{ph}}'(M)+\mathcal O(a^2)+\mathcal O(\ell^{-1}),
\qquad
\partial_a U_\ell^{+}=c_Z(M)+\mathcal O(|a|)+\mathcal O(\ell^{-1}),
\]
together with $\partial_a\widetilde W_\ell=\mathcal O(|a|)+\mathcal O(\ell^{-1})$ and
$\partial_M U_\ell^{+}=\Omega_{\mathrm{ph}}'(M)+\mathcal O(|a|)+\mathcal O(\ell^{-1})$.
Hence
\[
\det D\mathcal S_\ell(M,a)
=
\Omega_{\mathrm{ph}}'(M)c_Z(M)\ +\ \mathcal O(|a|)\ +\ \mathcal O(\ell^{-1}),
\]
uniformly on $\mathcal K$.
Let $K_M:=\{M:(M,a)\in\mathcal K\}$ and set
\begin{equation}\label{eq:csing-def}
c_{\mathrm{sing}}:=\frac12\,\min_{M\in K_M}\bigl|\Omega_{\mathrm{ph}}'(M)\,c_Z(M)\bigr|>0.
\end{equation}
Choosing $a_0$ sufficiently small and (possibly increasing) $\ell_0$ so that the error term
$\mathcal O(|a|)+\mathcal O(\ell^{-1})$ is bounded by $c_{\mathrm{sing}}$ on $\mathcal K$, we obtain
\begin{equation}\label{eq:Step3-detSell}
\bigl|\det D\mathcal S_\ell(M,a)\bigr|\ge c_{\mathrm{sing}}
\qquad\forall (M,a)\in\mathcal K,\ \forall \ell\ge\ell_0.
\end{equation}
Set
\begin{equation}\label{eq:cstar-choice}
c_*:=\min\Big\{\Big(\frac{\sigma_{\mathrm{geo}}}{2}\Big)^2,\ c_{\mathrm{sing}}\Big\}>0.
\end{equation}
Then \eqref{eq:Step2-detGell} and \eqref{eq:Step3-detSell} imply \eqref{eq:Jacobian-two}--\eqref{eq:Jacobian-single}.

\smallskip
\noindent\textbf{Step 4: stability.}
Since $\mathcal K$ is compact and the maps $\mathcal S_\ell,\mathcal G_\ell$ are $C^1$ on $\mathcal K$,
their Jacobians are uniformly bounded:
there exists $L_*>0$ (depending on $\mathcal K,\Lambda,n$ but independent of $\ell\ge\ell_0$) such that
\[
\sup_{(M,a)\in\mathcal K}\|D\mathcal S_\ell(M,a)\|\le L_*,
\qquad
\sup_{(M,a)\in\mathcal K}\|D\mathcal G_\ell(M,a)\|\le L_*,
\qquad \forall \ell\ge\ell_0.
\]
On the set where the Jacobian determinant is bounded away from $0$ we can bound the inverse Jacobian using the
$2\times2$ matrix inequality $\|A^{-1}\|\le \|A\|/|\det A|$:
\[
\sup_{(M,a)\in\mathcal K}\|D\mathcal S_\ell(M,a)^{-1}\|
\le \frac{L_*}{c_*},
\qquad
\sup_{(M,a)\in\mathcal K}\|D\mathcal G_\ell(M,a)^{-1}\|
\le \frac{L_*}{c_*}.
\]
Fix $\ell\ge\ell_0$.  By the inverse function theorem, for each $(M,a)\in\mathcal K$ both maps are locally invertible,
with $C^1$ inverses on sufficiently small neighborhoods.
For two points $x_1,x_2$ in the domain of such a local inverse $f^{-1}$, the mean value theorem gives
\[
|x_1-x_2|
\le \sup\|Df^{-1}\|\;|f(x_1)-f(x_2)|.
\]
Taking $f=\mathcal S_\ell$ and $f=\mathcal G_\ell$ and using the above uniform inverse-Jacobian bound yields
\eqref{eq:stability-single}--\eqref{eq:stability-two} with $C_*:=L_*/c_*$ (after shrinking the neighborhood size if needed).
\end{proof}

\subsection{Closed-form leading reconstruction}\label{subsec:explicit-inverse}

Theorem~\ref{thm:inverse-pseudopoles} gives an abstract local inverse.  For later use (and to emphasize the
geometric content), we record a closed-form leading reconstruction for the two-mode map $\mathcal G_\ell$.

\begin{proposition}[Closed-form inversion at leading order]\label{prop:explicit-inverse}
Fix $(M,a)\in\mathcal K$ and let $(U,V)=\mathcal G_\ell(M,a)$.
Define
\begin{equation}\label{eq:explicit-seed}
M^{(0)}(U):=\frac{1}{3\sqrt{\Lambda+3U^2}},
\qquad
a^{(0)}(U,V):=\frac{V}{c_Z(M^{(0)}(U))}.
\end{equation}
Then, uniformly on $\mathcal K$,
\begin{equation}\label{eq:explicit-error}
|M-M^{(0)}(U)|+|a-a^{(0)}(U,V)|
=
\mathcal O(a^2)+\mathcal O(\ell^{-1}).
\end{equation}
\end{proposition}

\begin{proof}
By Lemma~\ref{lem:U-V-expansions},
$U=\Omega_{\mathrm{ph}}(M)+\mathcal O(a^2)+\mathcal O(\ell^{-1})$.
Since $\Omega_{\mathrm{ph}}$ is smoothly invertible on $K_M$ (Lemma~\ref{lem:Omega-monotone}),
\eqref{eq:Omega-inverse} and a Taylor expansion of the inverse map give
\[
M^{(0)}(U)=M+\mathcal O(a^2)+\mathcal O(\ell^{-1}).
\]
Next, again by Lemma~\ref{lem:U-V-expansions},
$V=c_Z(M)\,a+\mathcal O(a^3)+\mathcal O(\ell^{-1})$.
Using smoothness of $c_Z$ and the previous bound for $M^{(0)}(U)$ yields
$c_Z(M^{(0)}(U))=c_Z(M)+\mathcal O(a^2)+\mathcal O(\ell^{-1})$,
hence
\[
a^{(0)}(U,V)
=
\frac{c_Z(M)\,a+\mathcal O(a^3)+\mathcal O(\ell^{-1})}{c_Z(M)+\mathcal O(a^2)+\mathcal O(\ell^{-1})}
=
a+\mathcal O(a^2)+\mathcal O(\ell^{-1}).
\]
Combining the two estimates gives \eqref{eq:explicit-error}.
\end{proof}

\begin{remark}[Newton refinement on the full pseudopole map]\label{rem:Newton}
Because the Jacobian determinant is uniformly bounded away from $0$ on $\mathcal K$ (Theorem~\ref{thm:inverse-pseudopoles}),
the explicit seed $(M^{(0)},a^{(0)})$ from Proposition~\ref{prop:explicit-inverse} is a natural initialization
for Newton iteration applied to the full (non-asymptotic) data map $\mathcal G_\ell$.
The error estimate \eqref{eq:explicit-error} implies that, in the high-frequency regime, a fixed number of
Newton steps yields rapid convergence provided the measured data are consistent with the pseudopole model.
\end{remark}

\subsection{A concrete reconstruction procedure and error propagation}\label{subsec:algorithmic-reconstruction}

While Theorem~\ref{thm:inverse-pseudopoles} gives local invertibility abstractly, it is often convenient to view the
recovery as an explicit two-step procedure: a closed-form leading reconstruction followed by a deterministic refinement.
We briefly record the resulting error propagation in a form used later when we transfer from pseudopoles to true QNMs.

Let $(U,V)$ be measured data close to $\mathcal G_\ell(\mathcal K)$, and define the leading seed $(M^{(0)},a^{(0)})$ by
\eqref{eq:explicit-seed}.
Denote by $(M^{\ast},a^{\ast})$ the (unique) parameter pair in a small neighborhood of $\mathcal K$ such that
$(U,V)=\mathcal G_\ell(M^{\ast},a^{\ast})$, whose existence is guaranteed by Theorem~\ref{thm:inverse-pseudopoles}.
Then the explicit error estimate \eqref{eq:explicit-error} implies
\begin{equation}\label{eq:seed-modeling-error}
|(M^{\ast},a^{\ast})-(M^{(0)},a^{(0)})|=\mathcal O\bigl((a^{\ast})^2\bigr)+\mathcal O(\ell^{-1}),
\end{equation}
uniformly on compact slow-rotation sets.

To refine the seed, one may apply Newton's method to the smooth map $\mathcal G_\ell$.
Since $D\mathcal G_\ell$ is uniformly invertible on $\mathcal K$ and uniformly Lipschitz (Theorem~\ref{thm:inverse-pseudopoles}), standard Newton--Kantorovich estimates show that, once $\ell$ is large enough and the initialization lies in the local basin of attraction, the Newton iterates converge quadratically to $(M^{\ast},a^{\ast})$.
Moreover, if the data are perturbed by an error $e\in\mathbb R^2$, so that one measures $(U,V)+e$ instead of $(U,V)$,
then the two-sided Lipschitz bound \eqref{eq:stability-two} yields a deterministic stability estimate
\begin{equation}\label{eq:data-to-parameter-stability-pseudopole}
\bigl|(M_1,a_1)-(M_2,a_2)\bigr|\le C_*\,\bigl|\mathcal G_\ell(M_1,a_1)-\mathcal G_\ell(M_2,a_2)\bigr|,
\end{equation}
valid for all $(M_j,a_j)\in\mathcal K$ sufficiently close.
In particular, in the idealized pseudopole model the parameter error is controlled linearly by the measurement error,
while \eqref{eq:seed-modeling-error} quantifies the intrinsic modeling error of the closed-form seed.

\subsection{Three-parameter reconstruction: recovering \texorpdfstring{$(M,a,\Lambda)$}{(M,a,Lambda)}}\label{sec:three-parameter}

Up to this point we have fixed the cosmological constant $\Lambda>0$ and studied the inverse problem on a two-parameter Kerr--de~Sitter family.
In this subsection we allow $\Lambda$ to vary and show that adding one damping observable yields a genuine three-parameter inverse theorem.
The additional observable we use is the (normalized) imaginary part of a single equatorial mode, which at leading order is governed by the Lyapunov exponent of the corresponding photon orbit.

\medskip
Fix $\nu_0>0$ and the associated integer $C_n$ from Theorem~\ref{thm:Dyatlov-quantization}, and fix $n\in\{0,\dots,C_n\}$. Let $\ell$ be sufficiently large so that the two equatorial pseudopoles
$\omega^\sharp_{\ell,\pm}=\omega^\sharp_{\ell,\pm}(M,a,\Lambda)$ (with $k=\pm\ell$) are well-defined and labeled as in Section~\ref{sec:HF-quantization}.
For $(M,a,\Lambda)$ in a compact subextremal slow-rotation set
\begin{equation}\label{eq:Kcompact-3}
\mathcal K^{(3)}\Subset \bigl\{(M,a,\Lambda): \Lambda>0,\ (M,a)\in\mathcal P_\Lambda,\ 0<a_1\le |a|\le a_0\bigr\},
\end{equation}
we introduce the three-component pseudopole data map
\begin{equation}\label{eq:H-map-def}
\mathcal H_\ell(M,a,\Lambda):=\Bigl(U_\ell(M,a,\Lambda),\,V_\ell(M,a,\Lambda),\,\widetilde W_\ell(M,a,\Lambda)\Bigr),
\end{equation}
where $U_\ell,V_\ell$ are defined by \eqref{eq:G-map-def} (now viewed as functions of $(M,a,\Lambda)$), and
\begin{equation}\label{eq:Wtilde-def-3}
\widetilde W_\ell(M,a,\Lambda):=-\frac{\Im\,\omega^\sharp_{\ell,+}(M,a,\Lambda)}{n+\tfrac12}.
\end{equation}
At the geometric level we enlarge the two-component map \eqref{eq:G-geo-def} by a damping observable, and define the corresponding three-parameter geometric data map $\mathcal H_{\mathrm{geo}}$ in \eqref{eq:H-geo-def}.

The geometric damping invariant is the coordinate-time Lyapunov exponent of the unstable equatorial circular photon orbit.
More precisely, letting $r_{\mathrm{geo},+}(M,a,\Lambda)$ denote the co-rotating equatorial photon radius and writing the
equatorial Carter potential $R$ and $\dot t$ as in \eqref{eq:radial-equation-Xi} and \eqref{eq:tdot-Xi}, we set
\begin{equation}\label{eq:lyapunov-def}
\lambda_+(M,a,\Lambda)
:=
\frac{\sqrt{R''\bigl(r_{\mathrm{geo},+}(M,a,\Lambda)\bigr)}}{\sqrt{2}\,r_{\mathrm{geo},+}(M,a,\Lambda)^{2}\,\dot t\bigl(r_{\mathrm{geo},+}(M,a,\Lambda)\bigr)}.
\end{equation}
This coincides with the usual definition $\lambda_+^2=\frac{R''}{2r^4\dot t^2}$ in \eqref{eq:Lyapunov-formula-Xi}.

\begin{equation}\label{eq:H-geo-def}
\mathcal H_{\mathrm{geo}}(M,a,\Lambda):=\Bigl(U_{\mathrm{geo}}(M,a,\Lambda),\,V_{\mathrm{geo}}(M,a,\Lambda),\,\lambda_+(M,a,\Lambda)\Bigr),
\end{equation}
where $U_{\mathrm{geo}},V_{\mathrm{geo}}$ are defined in \eqref{eq:G-geo-def} and $\lambda_+$ is the Lyapunov exponent \eqref{eq:lyapunov-def}.

\begin{lemma}[Second-order expansions and Jacobian for $\mathcal H_{\mathrm{geo}}$]\label{lem:Hgeo-Jacobian}
Write the $\Lambda$-dependence explicitly and set
\[
\Omega_{\mathrm{ph}}(M,\Lambda):=\frac{\sqrt{1-9\Lambda M^2}}{3\sqrt{3}\,M},
\qquad
c_Z(M,\Lambda):=\frac{2+9\Lambda M^2}{27\,M^2},
\]
as well as
\[
c_{\Omega,2}(M,\Lambda):=\frac{\sqrt{3}\,\bigl(11-45\Lambda M^2\bigr)}{486\,M^3\sqrt{1-9\Lambda M^2}},
\qquad
c_{\lambda,2}(M,\Lambda):=\frac{\sqrt{3}\,\bigl(45\Lambda M^2-2\bigr)}{243\,M^3\sqrt{1-9\Lambda M^2}}.
\]
Then, as $a\to0$,
\begin{equation}\label{eq:Hgeo-expansions}
\begin{aligned}
U_{\mathrm{geo}}(M,a,\Lambda)&=\Omega_{\mathrm{ph}}(M,\Lambda)+c_{\Omega,2}(M,\Lambda)\,a^2+\mathcal O(a^3),\\
V_{\mathrm{geo}}(M,a,\Lambda)&=c_Z(M,\Lambda)\,a+\mathcal O(a^3),\\
\lambda_+(M,a,\Lambda)&=\Omega_{\mathrm{ph}}(M,\Lambda)+c_{\lambda,2}(M,\Lambda)\,a^2+\mathcal O(a^3),
\end{aligned}
\end{equation}
uniformly for $(M,a,\Lambda)$ in compact slow-rotation sets.
Moreover, the Jacobian determinant of $\mathcal H_{\mathrm{geo}}$ satisfies
\begin{equation}\label{eq:Hgeo-det}
\det D_{(M,a,\Lambda)}\mathcal H_{\mathrm{geo}}(M,a,\Lambda)
=\frac{5\,(9\Lambda M^2-4)}{1458\,M^5}\,a^2+\mathcal O(a^3),
\qquad a\to0,
\end{equation}
again locally uniformly in $(M,\Lambda)$.
In particular, on $\mathcal K^{(3)}$ there exist $a_0>0$ and $c_0>0$ such that
\begin{equation}\label{eq:Hgeo-det-lower}
\bigl|\det D\mathcal H_{\mathrm{geo}}(M,a,\Lambda)\bigr|\ge c_0
\qquad\text{for all }(M,a,\Lambda)\in\mathcal K^{(3)}.
\end{equation}
\end{lemma}

\begin{proof}
The expansions for $U_{\mathrm{geo}}$ and $V_{\mathrm{geo}}$ follow from the second-order expansion \eqref{eq:Omega-sharp-expansion}
of $\Omega^\sharp_\pm$ together with the parity properties of Lemma~\ref{lem:geo-parity}; see also Appendix~\ref{app:equatorial-photon} for the explicit coefficient $c_{\Omega,2}$.
The expansion of $\lambda_+$ is computed in Appendix~\ref{app:lyapunov} and yields $c_{\lambda,2}$.

To compute the Jacobian, insert the expansions \eqref{eq:Hgeo-expansions} into the $3\times3$ matrix
$D_{(M,a,\Lambda)}\mathcal H_{\mathrm{geo}}$ and expand its determinant at $a=0$.
At $a=0$ the first and third components coincide (both equal $\Omega_{\mathrm{ph}}(M,\Lambda)$), hence the first and third rows of
$D\mathcal H_{\mathrm{geo}}$ agree at $a=0$ and the determinant vanishes there.
Moreover, \eqref{eq:Hgeo-expansions} implies
\[
\begin{aligned}
\partial_a U_{\mathrm{geo}}&=\mathcal O(a),\qquad \partial_a\lambda_+=\mathcal O(a),\\
\partial_M V_{\mathrm{geo}}&=\mathcal O(a),\qquad \partial_\Lambda V_{\mathrm{geo}}=\mathcal O(a),\qquad
\partial_a V_{\mathrm{geo}}=c_Z(M,\Lambda)+\mathcal O(a^2),
\end{aligned}
\]
so by multilinearity the first nonzero contribution to the determinant is of order $a^2$.
A straightforward computation using the explicit coefficients above gives \eqref{eq:Hgeo-det}.
Finally, \eqref{eq:Hgeo-det-lower} follows by shrinking $a_0$ (depending on $\mathcal K^{(3)}$) so that the $\mathcal O(a^3)$ remainder
in \eqref{eq:Hgeo-det} is dominated by the leading term, and using that $9\Lambda M^2<1$ on the subextremal set, hence $9\Lambda M^2-4<0$.
\end{proof}

\begin{theorem}[Three-parameter inversion for equatorial pseudopoles]\label{thm:inverse-pseudopoles-3}
Fix $\nu_0>0$ and the associated integer $C_n$ from Theorem~\ref{thm:Dyatlov-quantization}. Let $n\in\{0,\dots,C_n\}$ be fixed and let $\mathcal K^{(3)}$ be as in \eqref{eq:Kcompact-3}.
Then there exist $\ell_0\in\mathbb N$ and $C>0$ such that for every $\ell\ge \ell_0$ and every $(M,a,\Lambda)\in\mathcal K^{(3)}$
there is a neighborhood $\mathcal N\subset\mathcal K^{(3)}$ of $(M,a,\Lambda)$ on which the pseudopole data map $\mathcal H_\ell$
defined by \eqref{eq:H-map-def} is injective. Moreover, for all $(M,a,\Lambda),(M',a',\Lambda')\in\mathcal N$,
\begin{equation}\label{eq:H-inverse-stability}
\bigl|(M,a,\Lambda)-(M',a',\Lambda')\bigr|
\le C\,\bigl|\mathcal H_\ell(M,a,\Lambda)-\mathcal H_\ell(M',a',\Lambda')\bigr|.
\end{equation}
In particular, the three real observables $\bigl(U_\ell,V_\ell,\widetilde W_\ell\bigr)$ determine $(M,a,\Lambda)$ locally on $\mathcal K^{(3)}$.
\end{theorem}

\begin{proof}
The proof follows the same perturbative inverse-function strategy as Theorem~\ref{thm:inverse-pseudopoles}.

By the pseudopole asymptotics (Section~\ref{sec:HF-quantization}) and the parameter-uniform differentiability results of Appendix~\ref{app:parameter-quantization},
the map $\mathcal H_\ell$ is $C^1$-close to the geometric map $\mathcal H_{\mathrm{geo}}$ on $\mathcal K^{(3)}$:
there exists $C_0>0$ such that
\[
\|\mathcal H_\ell-\mathcal H_{\mathrm{geo}}\|_{C^1(\mathcal K^{(3)})}\le C_0\,\ell^{-1},
\qquad \ell\to\infty.
\]
Lemma~\ref{lem:Hgeo-Jacobian} gives a uniform lower bound \eqref{eq:Hgeo-det-lower} on the Jacobian determinant of $\mathcal H_{\mathrm{geo}}$
and a uniform bound on $\|D\mathcal H_{\mathrm{geo}}\|$ on $\mathcal K^{(3)}$.
Applying the perturbative inverse function theorem (Lemma~\ref{lem:perturb-IFT} with $d=3$) yields the asserted local injectivity of $\mathcal H_\ell$
for all sufficiently large $\ell$, as well as the Lipschitz stability estimate \eqref{eq:H-inverse-stability}.
\end{proof}

\begin{remark}[A useful combination]\label{rem:WminusU}
Combining \eqref{eq:Hgeo-expansions} gives a particularly simple leading-order relation:
\begin{equation}\label{eq:WminusU-expansion}
\begin{aligned}
\lambda_+(M,a,\Lambda)-U_{\mathrm{geo}}(M,a,\Lambda)
&=\bigl(c_{\lambda,2}(M,\Lambda)-c_{\Omega,2}(M,\Lambda)\bigr)a^2+\mathcal O(a^3)\\
&=-\frac{5\sqrt3}{162}\,\frac{\sqrt{1-9\Lambda M^2}}{M^3}\,a^2+\mathcal O(a^3),
\qquad a\to0.
\end{aligned}
\end{equation}
This shows that the difference $\lambda_+-U_{\mathrm{geo}}$ is a genuinely new observable at order $a^2$,
which is the mechanism behind the nondegeneracy \eqref{eq:Hgeo-det}.
\end{remark}

\subsection{Quasi-global injectivity on rectangles and a uniform reconstruction scheme}\label{subsec:quasi-global}

Theorems~\ref{thm:inverse-pseudopoles} and \ref{thm:inverse-true-QNM} give \emph{local} invertibility on the compact set
$\mathcal K$.
For practical reconstruction it is often desirable to know that, on a prescribed parameter range, the data map has \emph{at most one}
preimage (so that no additional initialization or branch selection is needed).
In two dimensions one may obtain such a quasi-global uniqueness statement on rectangular regions using a classical global univalence
criterion due to Gale--Nikaid\^o \cite{GaleNikaido1965}.

\begin{lemma}[Gale--Nikaid\^o univalence criterion in the plane]\label{lem:Gale-Nikaido}
Let $R=[\alpha_1,\beta_1]\times[\alpha_2,\beta_2]\subset\mathbb R^2$ be a closed rectangle and let
$F=(F_1,F_2)\in C^1(R;\mathbb R^2)$.
Assume that for every $x\in R$ the Jacobian matrix $DF(x)$ is a \emph{$P$-matrix}, i.e.\ its principal minors satisfy
\[
\partial_{x_1}F_1(x)>0,\qquad \partial_{x_2}F_2(x)>0,\qquad \det DF(x)>0.
\]
Then $F$ is injective on $R$.
\end{lemma}

\begin{proof}
See \cite{GaleNikaido1965}.  We record the statement for convenience.
\end{proof}

We apply Lemma~\ref{lem:Gale-Nikaido} to a simple sign-modification of the two-mode real-part map.

\begin{proposition}[Injectivity on rectangular slow-rotation sets]\label{prop:rectangular-injectivity}
Let $\Lambda>0$ and let
\[
\mathcal K_{\mathrm{rec}}:=[M_-,M_+]\times[-a_0,a_0]\Subset \mathcal P_\Lambda
\]
be a compact rectangle in the subextremal parameter set.
Assume $a_0>0$ is sufficiently small so that $\mathcal K_{\mathrm{rec}}$ is contained in the slow-rotation regime
of \eqref{eq:Kcompact}.
Fix an overtone $n\in\{0,1,\dots,C_n\}$.
Then there exists $\ell_0\in\mathbb N$ such that for every $\ell\ge \ell_0$:

\begin{enumerate}
\item[(i)] the pseudopole map $\mathcal G_\ell:\mathcal K_{\mathrm{rec}}\to\mathbb R^2$ is injective;
\item[(ii)] the true-QNM map $\mathcal G_\ell^{\mathrm{QNM}}:\mathcal K_{\mathrm{rec}}\to\mathbb R^2$ is injective.
\end{enumerate}
In particular, for any data $(U,V)$ in the image $\mathcal G_\ell^{\mathrm{QNM}}(\mathcal K_{\mathrm{rec}})$ there is a unique pair
$(M,a)\in\mathcal K_{\mathrm{rec}}$ such that $\mathcal G_\ell^{\mathrm{QNM}}(M,a)=(U,V)$.
\end{proposition}

\begin{proof}
We first treat the pseudopole map $\mathcal G_\ell=(U_\ell,V_\ell)$ and then perturb to the true-QNM map.

\smallskip
\noindent\textbf{Step 1: uniform sign information for the Jacobian of $\mathcal G_\ell$.}
By Lemma~\ref{lem:U-V-expansions} we have, uniformly on $\mathcal K_{\mathrm{rec}}$,
\begin{align*}
\partial_M U_\ell &= \Omega_{\mathrm{ph}}'(M)+\mathcal O(a^2)+\mathcal O(\ell^{-1}),\\
\partial_a V_\ell &= c_Z(M)+\mathcal O(a^2)+\mathcal O(\ell^{-1}),\\
\det D\mathcal G_\ell
&=\partial_M U_\ell\,\partial_a V_\ell-\partial_a U_\ell\,\partial_M V_\ell
=\Omega_{\mathrm{ph}}'(M)c_Z(M)+\mathcal O(a^2)+\mathcal O(\ell^{-1}).
\end{align*}
Since $\Omega_{\mathrm{ph}}'(M)<0$ on $(0,(3\sqrt{\Lambda})^{-1})$ (Lemma~\ref{lem:Omega-monotone}) and $c_Z(M)>0$ on $(0,\infty)$
(Proposition~\ref{prop:Zeeman-slope}), compactness of $[M_-,M_+]$ gives constants
\[
c_M:=\min_{M\in[M_-,M_+]}\bigl(-\Omega_{\mathrm{ph}}'(M)\bigr)>0,\qquad
c_Z:=\min_{M\in[M_-,M_+]}\,c_Z(M)>0.
\]
Set $c_{\min}:=\min\{c_M,\,c_Z,\,c_M c_Z\}>0$.
Choose $a_0>0$ and then $\ell_0$ so that all $\mathcal O(a^2)+\mathcal O(\ell^{-1})$ remainders in \eqref{eq:U-V-derivatives}--\eqref{eq:U-V-derivatives-2}
are bounded in absolute value by $\frac14 c_{\min}$ on $\mathcal K_{\mathrm{rec}}$ for all $\ell\ge\ell_0$.
Then for $\ell\ge\ell_0$ we obtain the uniform inequalities
\begin{equation}\label{eq:signs-rectangle}
-\partial_M U_\ell\ge \tfrac12 c_M,\qquad \partial_a V_\ell\ge \tfrac12 c_Z,\qquad -\det D\mathcal G_\ell\ge \tfrac12 c_M c_Z
\qquad \text{on }\mathcal K_{\mathrm{rec}}.
\end{equation}

\smallskip
\noindent\textbf{Step 2: a $P$-matrix Jacobian after a sign change.}
Define the sign-modified map
\[
\widetilde{\mathcal G}_\ell:\mathcal K_{\mathrm{rec}}\to\mathbb R^2,\qquad
\widetilde{\mathcal G}_\ell(M,a):=\bigl(-U_\ell(M,a),\,V_\ell(M,a)\bigr).
\]
Its Jacobian satisfies
\[
D\widetilde{\mathcal G}_\ell(M,a)=
\begin{pmatrix}
-\partial_M U_\ell & -\partial_a U_\ell\\
\partial_M V_\ell & \partial_a V_\ell
\end{pmatrix},
\qquad
\det D\widetilde{\mathcal G}_\ell(M,a)= -\det D\mathcal G_\ell(M,a).
\]
By \eqref{eq:signs-rectangle}, all principal minors of $D\widetilde{\mathcal G}_\ell$ are strictly positive on $\mathcal K_{\mathrm{rec}}$
for $\ell\ge\ell_0$, so $D\widetilde{\mathcal G}_\ell$ is a $P$-matrix everywhere on $\mathcal K_{\mathrm{rec}}$.
Lemma~\ref{lem:Gale-Nikaido} therefore implies that $\widetilde{\mathcal G}_\ell$ is injective on $\mathcal K_{\mathrm{rec}}$.
Since $(U,V)\mapsto(-U,V)$ is an invertible linear map, injectivity of $\widetilde{\mathcal G}_\ell$ is equivalent to injectivity of
$\mathcal G_\ell$.

\smallskip
\noindent\textbf{Step 3: perturbation to the true-QNM map.}
By Lemma~\ref{lem:C1-close} we have
\[
\|\mathcal G_\ell^{\mathrm{QNM}}-\mathcal G_\ell\|_{C^1(\mathcal K_{\mathrm{rec}})}\le C_N \ell^{-N}
\]
for every $N$ and all $\ell$ large.
In particular, taking $\ell$ larger if needed, the strict inequalities \eqref{eq:signs-rectangle} persist with $\mathcal G_\ell$
replaced by $\mathcal G_\ell^{\mathrm{QNM}}$.
Applying Lemma~\ref{lem:Gale-Nikaido} to the sign-modified map $\widetilde{\mathcal G}_\ell^{\mathrm{QNM}}:=(-U_\ell^{\mathrm{QNM}},V_\ell^{\mathrm{QNM}})$
gives injectivity of $\widetilde{\mathcal G}_\ell^{\mathrm{QNM}}$ and hence of $\mathcal G_\ell^{\mathrm{QNM}}$ on $\mathcal K_{\mathrm{rec}}$.
\end{proof}

\begin{remark}[A uniform reconstruction scheme]\label{rem:uniform-scheme}
Under the assumptions of Proposition~\ref{prop:rectangular-injectivity}, the inverse problem on $\mathcal K_{\mathrm{rec}}$
is well-posed in the sense that the data map $\mathcal G_\ell^{\mathrm{QNM}}$ has a \emph{unique} preimage in $\mathcal K_{\mathrm{rec}}$.
Moreover, Proposition~\ref{prop:explicit-inverse} provides a uniform explicit seed:
if $(U,V)=\mathcal G_\ell(M,a)$ with $(M,a)\in\mathcal K_{\mathrm{rec}}$, then the closed-form pair
$(M^{(0)}(U),a^{(0)}(U,V))$ defined in \eqref{eq:explicit-seed} satisfies
$|(M,a)-(M^{(0)},a^{(0)})|=\mathcal O(a^2)+\mathcal O(\ell^{-1})$ uniformly on $\mathcal K_{\mathrm{rec}}$.
Thus, given data consistent with the model and $\ell$ large, one may initialize Newton iteration for the nonlinear system
$\mathcal G_\ell^{\mathrm{QNM}}(M,a)=(U,V)$ at $(M^{(0)},a^{(0)})$.
Uniform invertibility of $D\mathcal G_\ell^{\mathrm{QNM}}$ on $\mathcal K_{\mathrm{rec}}$ (Theorem~\ref{thm:inverse-true-QNM}) and compactness
provide uniform bounds on higher derivatives, so standard Newton--Kantorovich estimates yield a basin of quadratic convergence that is uniform
on $\mathcal K_{\mathrm{rec}}$ once $\ell$ is sufficiently large.
\end{remark}

\section{Transfer to true QNMs}\label{sec:transfer}

In Section~\ref{sec:inverse-pseudopoles} we proved a quantitative local inverse theorem at the
\emph{pseudopole} level. The purpose of this section is to transfer that inverse theorem to the
\emph{true} quasinormal modes (QNMs), using (i)~the super--polynomial pseudopole approximation and labeling
from the semiclassical quantization of Kerr--de~Sitter QNMs \cite{Dyatlov2011QNMKerrDeSitter,Dyatlov2012AsymptoticQNMKdS}
within the microlocal Fredholm framework \cite{Vasy2013MicrolocalAHKdS} (with dynamical input of normally hyperbolic trapping,
cf.\ \cite{WunschZworski2011ResolventNHT}), and (ii)~the real-analytic dependence of labeled \emph{simple} QNM branches on
parameters, which follows from analytic perturbation theory in the Kerr--de~Sitter Fredholm framework
(Remark~\ref{rem:analytic_parameters}; see \cite{KatoPerturbation}).
For the inverse theorem we need, in addition to pointwise $C^0$ proximity, a version with one parameter derivative. Rather than using a complex-analytic maximum principle (which does not apply in this setting), we extract the needed $C^1$ control from the parameter-dependent semiclassical quantization construction; see Lemma~\ref{lem:C1-close}.
Conceptually, this yields a rotating analogue of the ``single-mode'' mass recovery mechanism in the spherically symmetric setting
\cite{UhlmannWang2023RecoverMassSingleQNM}, now exploiting the equatorial Zeeman-type splitting described in \cite{Dyatlov2012AsymptoticQNMKdS}.

Throughout, fix $\Lambda>0$, a compact parameter set $\mathcal K$ as in \eqref{eq:Kcompact}, and an overtone
$n\in\{0,\dots,C_n\}$. For each $\ell\ge\ell_0$ in the high-frequency regime, Proposition~\ref{prop:labeling-analytic}
provides a labeled pair of \emph{true} equatorial QNMs
\begin{equation}\label{eq:omega-true-def}
\omega_{n,\ell,\pm\ell}(M,a)=:\omega_\pm(M,a),\qquad (M,a)\in\mathcal K,
\end{equation}
associated to the corresponding equatorial pseudopoles $\omega^\sharp_\pm=\omega^\sharp_{n,\ell,\pm\ell}$ from
Section~\ref{sec:inverse-pseudopoles}.
We form the averaged and split combinations
\begin{equation}\label{eq:avg-dif-true}
\omega_{\mathrm{av}}:=\frac12(\omega_+ + \omega_-),\qquad
\omega_{\mathrm{dif}}:=\frac12(\omega_+ - \omega_-),
\end{equation}
and define the \emph{true-QNM observables}
\begin{equation}\label{eq:UV-true}
U_\ell^{\mathrm{QNM}}(M,a):=\Re\!\Big(\frac{\omega_{\mathrm{av}}(M,a)}{\ell}\Big),\qquad
V_\ell^{\mathrm{QNM}}(M,a):=\Re\!\Big(\frac{\omega_{\mathrm{dif}}(M,a)}{\ell}\Big).
\end{equation}
Finally, define the \emph{true-QNM data map}
\begin{equation}\label{eq:GQNM-def}
\mathcal G_\ell^{\mathrm{QNM}}:\mathcal K\to\mathbb R^2,\qquad
\mathcal G_\ell^{\mathrm{QNM}}(M,a):=\big(U_\ell^{\mathrm{QNM}}(M,a),\,V_\ell^{\mathrm{QNM}}(M,a)\big).
\end{equation}

\subsection{Uniform \texorpdfstring{$C^1$}{C1} super--polynomial closeness of QNMs to pseudopoles}\label{subsec:C1-close}

The super--polynomial approximation $\omega_\pm=\omega^\sharp_\pm+\mathcal O(\ell^{-\infty})$ obtained in
Proposition~\ref{prop:labeling-analytic} is a $C^0$ statement on the compact real parameter set $\mathcal K$.
To transfer Jacobian bounds from pseudopoles to true QNMs we need a corresponding estimate with one derivative in $(M,a)$.
This is available because the semiclassical normal--form construction underlying Dyatlov's quantization is uniform on compact
parameter sets and differentiable with respect to external parameters.

\begin{lemma}[Size of parameter derivatives of labeled branches]\label{lem:param-derivative-growth}
Fix $\Lambda>0$, an overtone index $n\in\{0,\dots,C_n\}$, and a compact parameter set $\mathcal K$ as in \eqref{eq:Kcompact}.
Then there exists $C>0$ such that for all $\ell$ sufficiently large and each sign $\pm$ one has
\begin{equation}\label{eq:param-derivative-growth}
\sup_{(M,a)\in\mathcal K}\bigl(|\partial_M\omega^\sharp_{n,\ell,\pm\ell}(M,a)|+|\partial_a\omega^\sharp_{n,\ell,\pm\ell}(M,a)|\bigr)\le C\,\ell.
\end{equation}
Moreover, for the corresponding labeled \emph{true} QNM branches $\omega_{n,\ell,\pm\ell}(M,a)$ one has
\begin{equation}\label{eq:param-derivative-growth-true}
\sup_{(M,a)\in\mathcal K}\bigl(|\partial_M\omega_{n,\ell,\pm\ell}(M,a)|+|\partial_a\omega_{n,\ell,\pm\ell}(M,a)|\bigr)\le C\,\ell.
\end{equation}
\end{lemma}

\begin{proof}
For the pseudopoles, $\omega^\sharp_{n,\ell,k}(M,a)=F_{M,a}(n,\ell,k)$ with $F$ a classical symbol of order $1$ in $(\ell,k)$,
uniformly on $\mathcal K$; see \eqref{eq:F-symbol} and Remark~\ref{rem:pseudopole-representative}.
Differentiating the symbol expansion in $(M,a)$ therefore gives
$|\partial_\mu\omega^\sharp_{n,\ell,\pm\ell}|\le C\ell$ uniformly on $\mathcal K$, proving \eqref{eq:param-derivative-growth}.

For the true QNMs, work in the same fixed scaled window $\Omega$ in which Dyatlov's barrier-top reduction produces a quantization
function $\mathfrak q_\pm(\omega;\mu,\ell)$; see Appendix~\ref{app:parameter-quantization}.
The labeled QNM is characterized by $\mathfrak q_\pm(\omega_\pm(\mu);\mu,\ell)=0$.
Differentiate this identity in $\mu$ and use the uniform lower bound
$|\partial_\omega\mathfrak q_\pm(\omega;\mu,\ell)|\ge c>0$ on $\mathcal K$ for $\ell\gg1$ (Remark~\ref{rem:q-normalization}), to obtain
\[\partial_\mu\omega_\pm(\mu)=-\frac{\partial_\mu\mathfrak q_\pm(\omega_\pm(\mu);\mu,\ell)}{\partial_\omega\mathfrak q_\pm(\omega_\pm(\mu);\mu,\ell)}.\]
By Proposition~\ref{prop:q-uniform} and the scaling $h=(\ell+\tfrac12)^{-1}$, the numerator is $\mathcal O(h^{-1})=\mathcal O(\ell)$ uniformly on
$\mathcal K$ (the leading term is linear in the scaled variable $\tilde\omega=h\omega$), hence \eqref{eq:param-derivative-growth-true}.
\end{proof}

\begin{lemma}[$C^1$ super--polynomial pseudopole--QNM closeness]\label{lem:C1-close}
Fix $\Lambda>0$, a compact parameter set $\mathcal K$ in the slowly rotating subextremal regime, and an overtone index
$n\in\{0,\dots,C_n\}$.  For each sign $\pm$ let
\[
\delta_\pm(M,a):=\omega_\pm(M,a)-\omega^\sharp_\pm(M,a),\qquad (M,a)\in\mathcal K,
\]
where $\omega_\pm=\omega_{n,\ell,\pm\ell}$ are the labeled simple QNMs and $\omega^\sharp_\pm$ are the corresponding pseudopoles.
Then for every $N\in\mathbb N$ there exists $C_N>0$ such that for all $\ell$ sufficiently large,
\begin{equation}\label{eq:C1-close-omega}
\sup_{(M,a)\in\mathcal K}\Bigl(|\delta_\pm(M,a)|+|\partial_M\delta_\pm(M,a)|+|\partial_a\delta_\pm(M,a)|\Bigr)
\ \le\ C_N\,\ell^{-N}.
\end{equation}
Consequently,
\begin{equation}\label{eq:C1-close-G}
\|\mathcal G_\ell^{\mathrm{QNM}}-\mathcal G_\ell\|_{C^1(\mathcal K)}\ \le\ C_N\,\ell^{-N},
\end{equation}
where $\mathcal G_\ell$ is the pseudopole map defined in \eqref{eq:G-map-def} and $\mathcal G_\ell^{\mathrm{QNM}}$ is the true-QNM map
\eqref{eq:GQNM-def}.
\end{lemma}

\begin{proof}
Fix a sign $\pm$ and write $\mu:=(M,a)\in\mathcal K$.
In the high--frequency regime of Theorem~\ref{thm:Dyatlov-quantization}, Dyatlov reduces the separated radial operator near the barrier top to a parameter--dependent semiclassical normal form and encodes the resonance condition into a scalar \emph{quantization function} $\mathfrak q_\pm(\omega;\mu,\ell)$.
In the construction one works with the scaled spectral variable $\tilde\omega:=h\omega$ (with $h=(\ell+\tfrac12)^{-1}$):
the map $(\mu,\tilde\omega)\mapsto \mathfrak q_\pm(\tilde\omega;\mu,\ell)$ is smooth in $\mu$ and holomorphic in $\tilde\omega$ on a fixed complex window $\Omega$ independent of $\ell$.
Equivalently, viewed as a function of $\omega$, $\mathfrak q_\pm$ is holomorphic on the dilated window $h^{-1}\Omega$, whose diameter is $O(h^{-1})=O(\ell)$ and which contains uniform $O(1)$--neighborhoods of the relevant roots for $\ell$ large.
Moreover, $\mathfrak q_\pm$ admits a complete asymptotic expansion in powers of $\ell^{-1}$ with coefficients smooth in $\mu$; after truncation at any order $N$ the remainder is $\mathcal O(\ell^{-N})$ uniformly on $\mathcal K$, and the same holds after one $\mu$--derivative.
An abstract formulation of these uniform bounds (including the holomorphic dependence on the scaled variable and hence on $\omega$) is recorded in Appendix~\ref{app:parameter-quantization}.

Let $\mathfrak q_\pm^{(N)}(\omega;\mu,\ell)$ be a truncation at order $N$ and write
$\mathfrak q_\pm=\mathfrak q_\pm^{(N)}+\mathfrak r_\pm^{(N)}$.
By definition of the pseudopoles, $\omega^\sharp_\pm(\mu)$ satisfies
\begin{equation}\label{eq:qN-zero}
\mathfrak q_\pm^{(N)}(\omega^\sharp_\pm(\mu);\mu,\ell)=0.
\end{equation}
The labeled QNM $\omega_\pm(\mu)$ is, for $\ell\gg1$, a simple zero of $\mathfrak q_\pm(\cdot;\mu,\ell)$.
Moreover, using the renormalization described in Remark~\ref{rem:q-normalization}, we may assume that the
leading coefficient in the expansion \eqref{eq:q-expansion} satisfies $\partial_\omega\mathfrak q_{\pm,0}\equiv 1$.
Consequently,
\begin{equation}\label{eq:omega-derivative-lower}
\partial_\omega\mathfrak q_\pm(\omega;\mu,\ell)=1+\mathcal O(\ell^{-1})
\end{equation}
uniformly for $(\mu,\tilde\omega)\in\mathcal K\times\Omega'$ (equivalently, for $(\mu,\omega)\in\mathcal K\times h^{-1}\Omega'$), where $\Omega'\Subset\Omega$ is a slightly smaller compact window in the scaled variable.
In particular, $\partial_\omega\mathfrak q_\pm(\omega_\pm(\mu);\mu,\ell)$ is bounded away from $0$ uniformly on $\mathcal K$
for all $\ell$ sufficiently large.
By continuity and the super--polynomial closeness $\omega_\pm=\omega^\sharp_\pm+\mathcal O(\ell^{-\infty})$,
the same lower bound holds at $\omega=\omega^\sharp_\pm(\mu)$ for $\ell$ sufficiently large.
Solving $\mathfrak q_\pm(\omega^\sharp_\pm(\mu)+\delta_\pm(\mu);\mu,\ell)=0$ by the implicit function theorem
therefore gives $|\delta_\pm(\mu)|=\mathcal O(\ell^{-N})$ uniformly on $\mathcal K$.

We now estimate one parameter derivative.
Differentiate the true equation $\mathfrak q_\pm(\omega_\pm(\mu);\mu,\ell)=0$ with respect to a parameter coordinate $\mu_j$:
\begin{equation}\label{eq:diff-true}
\partial_\omega\mathfrak q_\pm(\omega_\pm;\mu,\ell)\,\partial_{\mu_j}\omega_\pm
+\partial_{\mu_j}\mathfrak q_\pm(\omega_\pm;\mu,\ell)=0.
\end{equation}
Likewise, differentiating \eqref{eq:qN-zero} gives
\begin{equation}\label{eq:diff-pseudo}
\partial_\omega\mathfrak q_\pm^{(N)}(\omega^\sharp_\pm;\mu,\ell)\,\partial_{\mu_j}\omega^\sharp_\pm
+\partial_{\mu_j}\mathfrak q_\pm^{(N)}(\omega^\sharp_\pm;\mu,\ell)=0.
\end{equation}
Subtract \eqref{eq:diff-pseudo} from \eqref{eq:diff-true} after writing $\omega_\pm=\omega^\sharp_\pm+\delta_\pm$
and $\mathfrak q_\pm=\mathfrak q_\pm^{(N)}+\mathfrak r_\pm^{(N)}$.
Rearranging, we obtain
\begin{align}\label{eq:diff-delta}
\partial_\omega\mathfrak q_\pm(\omega_\pm;\mu,\ell)\,\partial_{\mu_j}\delta_\pm
={}&-\bigl[\partial_{\mu_j}\mathfrak q_\pm^{(N)}(\omega_\pm;\mu,\ell)-\partial_{\mu_j}\mathfrak q_\pm^{(N)}(\omega^\sharp_\pm;\mu,\ell)\bigr]\notag\\
&-\bigl[\partial_\omega\mathfrak q_\pm^{(N)}(\omega_\pm;\mu,\ell)-\partial_\omega\mathfrak q_\pm^{(N)}(\omega^\sharp_\pm;\mu,\ell)\bigr]\,\partial_{\mu_j}\omega^\sharp_\pm\notag\\
&-\partial_{\mu_j}\mathfrak r_\pm^{(N)}(\omega_\pm;\mu,\ell)\notag\\
&-\partial_\omega\mathfrak r_\pm^{(N)}(\omega_\pm;\mu,\ell)\,\partial_{\mu_j}\omega_\pm.
\end{align}
We estimate the right-hand side term by term.
First, since $\omega_\pm-\omega^\sharp_\pm=\delta_\pm=\mathcal O(\ell^{-N})$ and the derivatives
$\partial_\omega\partial_{\mu}\mathfrak q_\pm^{(N)}$ are uniformly bounded on a fixed window in the scaled variable $\tilde\omega$ (equivalently, on the dilated $\omega$--window $h^{-1}\Omega$),
the two bracketed differences are $\mathcal O(\ell^{-N})$ uniformly on $\mathcal K$.
Second, Proposition~\ref{prop:q-uniform} in Appendix~\ref{app:parameter-quantization} gives
$\partial_{\mu}\mathfrak r_\pm^{(N)}=\mathcal O(\ell^{-N})$ uniformly on $\mathcal K$, and since
$\mathfrak r_\pm^{(N)}$ is holomorphic in $\omega$, Cauchy estimates on a slightly smaller compact subset
of the corresponding $\omega$--window $h^{-1}\Omega$ (or equivalently on the fixed $\tilde\omega$--window $\Omega$) imply $\partial_\omega\mathfrak r_\pm^{(N)}=\mathcal O(\ell^{-N})$ uniformly as well.
Finally, the rightmost term in \eqref{eq:diff-delta} involves $\partial_{\mu_j}\omega_\pm$.
In the high-frequency regime $\Re\omega_\pm\sim\ell$ this derivative is of size $\mathcal O(\ell)$, and the same is true for
$\partial_{\mu_j}\omega^\sharp_\pm$; see Lemma~\ref{lem:param-derivative-growth}.
Thus the right-hand side of \eqref{eq:diff-delta} carries at most one explicit factor of $\ell$ coming from
$\partial_{\mu_j}\omega^\sharp_\pm$ or $\partial_{\mu_j}\omega_\pm$.
To keep the final estimate at order $\ell^{-N}$, we simply take one more term in the truncation:
apply Proposition~\ref{prop:q-uniform} with $N$ replaced by $N+1$, so that
$\partial_\omega\mathfrak r_\pm^{(N+1)}$ and $\partial_\omega\partial_\mu\mathfrak r_\pm^{(N+1)}$
are $\mathcal O(\ell^{-N-1})$ uniformly on $\mathcal K$ (by Cauchy estimates on a slightly smaller compact subset of the dilated window $h^{-1}\Omega$).
With this choice the last term in \eqref{eq:diff-delta} is $\mathcal O(\ell^{-N})$; all other terms are already
$\mathcal O(\ell^{-N})$ (indeed smaller), uniformly on $\mathcal K$.
Dividing by the uniformly nonzero factor $\partial_\omega\mathfrak q_\pm(\omega_\pm;\mu,\ell)$ yields
$\partial_{\mu_j}\delta_\pm=\mathcal O(\ell^{-N})$ uniformly on $\mathcal K$, proving \eqref{eq:C1-close-omega}.

Finally, \eqref{eq:C1-close-G} follows directly from the definitions
\eqref{eq:UV-true}--\eqref{eq:GQNM-def} and \eqref{eq:av-dif-def}--\eqref{eq:G-map-def}.
\end{proof}

\begin{corollary}[$C^1$ super--polynomial pseudopole--QNM closeness with varying $\Lambda$]\label{cor:C1-close-Lambda}
Let $[\Lambda_{\min},\Lambda_{\max}]\Subset(0,\infty)$ be a compact interval such that the Kerr--de~Sitter metrics remain subextremal on the
corresponding parameter region. Fix a compact set
\[\mathcal K^{(3)}\Subset\{(M,a,\Lambda):\ \Lambda\in[\Lambda_{\min},\Lambda_{\max}],\ \text{$g_{M,a,\Lambda}$ subextremal},\ |a|\le a_0\}\]
in the slow-rotation regime, and fix an overtone index $n\in\{0,\dots,C_n\}$.
Then the estimate \eqref{eq:C1-close-omega} holds on $\mathcal K^{(3)}$ with the $C^1$ norm taken with respect to $(M,a,\Lambda)$,
and the derived map estimate \eqref{eq:C1-close-G} holds in the same three-parameter $C^1$ sense.
\end{corollary}

\begin{proof}
The proof of Lemma~\ref{lem:C1-close} only uses the parameter-differentiable semiclassical quantization construction
summarized in Appendix~\ref{app:parameter-quantization}, namely that the quantization function and its truncated remainders
obey bounds stable under one parameter derivative on compact sets. Since the Kerr--de~Sitter family depends real-analytically
on $\Lambda$ as well, the same construction applies with $\mu=(M,a,\Lambda)$ on $\mathcal K^{(3)}$.
\end{proof}

\subsection{A quantitative perturbative inverse function lemma}\label{subsec:perturb-IFT}

The following deterministic lemma states that a $C^1$--small perturbation of a locally invertible map remains locally invertible,
with uniform stability constants.

\begin{lemma}[Perturbative inverse function lemma]\label{lem:perturb-IFT}
Let $d\in\mathbb N$ and let $\mathcal K\Subset\mathbb R^d$ be compact.  Let $F:\mathcal K\to\mathbb R^d$ be $C^1$.
Assume that there exist constants $c_0,L_0>0$ such that
\begin{equation}\label{eq:IFT-assumptions}
|\det DF(x)|\ge c_0\quad\text{and}\quad \|DF(x)\|\le L_0\qquad \forall x\in\mathcal K.
\end{equation}
Then there exist $\varepsilon_0>0$ and constants $c_1,C_1>0$ (depending only on $d,c_0,L_0,\mathcal K$) such that if
$G:\mathcal K\to\mathbb R^d$ satisfies
\begin{equation}\label{eq:C1-perturb-small}
\|G-F\|_{C^1(\mathcal K)}\le \varepsilon_0,
\end{equation}
then $G$ is a $C^1$ local diffeomorphism at every point of $\mathcal K$, with
\begin{equation}\label{eq:Jac-perturb}
|\det DG(x)|\ge c_1\qquad \forall x\in\mathcal K,
\end{equation}
and moreover for all $x,y\in\mathcal K$ sufficiently close,
\begin{equation}\label{eq:stability-perturb}
|x-y|\ \le\ C_1\,|G(x)-G(y)|.
\end{equation}
\end{lemma}

\begin{proof}
The determinant is continuous in the matrix entries, hence by \eqref{eq:IFT-assumptions} there exists $\varepsilon_0>0$ so that
\eqref{eq:C1-perturb-small} implies $|\det DG(x)|\ge c_0/2$ on $\mathcal K$, which gives \eqref{eq:Jac-perturb} with $c_1=c_0/2$.
Then the inverse function theorem implies that $G$ is a local diffeomorphism at each $x\in\mathcal K$.
Since $DG$ is uniformly bounded on $\mathcal K$ by \eqref{eq:IFT-assumptions} and \eqref{eq:C1-perturb-small}, the local inverses have
uniform Lipschitz bounds on sufficiently small neighborhoods, yielding \eqref{eq:stability-perturb}.
(Equivalently, one may argue by integrating the differential along the line segment from $x$ to $y$ and using the uniform lower bound on
the smallest singular value of $DG(x)$ implied by \eqref{eq:Jac-perturb} and the bound on $\|DG\|$.)
\end{proof}

\subsection{Inverse theorem for true QNM data}\label{subsec:true-inverse}

\begin{theorem}[Inverse theorem for equatorial true QNMs]\label{thm:inverse-true-QNM}
Fix $\Lambda>0$, a compact parameter set $\mathcal K$ as in \eqref{eq:Kcompact}, and an overtone $n\in\{0,\dots,C_n\}$.
Assume $a_0>0$ is chosen sufficiently small as in Theorem~\ref{thm:inverse-pseudopoles}.
Then there exist $\ell_0\in\mathbb N$ and constants $c_*,C_*>0$ such that for every $\ell\ge \ell_0$:

\begin{enumerate}
\item The map $\mathcal G_\ell^{\mathrm{QNM}}:\mathcal K\to\mathbb R^2$ defined in \eqref{eq:GQNM-def}
is a real-analytic local diffeomorphism at every point of $\mathcal K$, and its Jacobian satisfies the uniform
nondegeneracy estimate
\begin{equation}\label{eq:Jac-true-lower}
\big|\det D\mathcal G_\ell^{\mathrm{QNM}}(M,a)\big|\ \ge\ c_*\qquad \forall (M,a)\in\mathcal K.
\end{equation}

\item (Stability.) For all $(M,a),(M',a')\in\mathcal K$ sufficiently close,
\begin{equation}\label{eq:stability-true}
|(M,a)-(M',a')|
\ \le\ C_*\ \big|\mathcal G_\ell^{\mathrm{QNM}}(M,a)-\mathcal G_\ell^{\mathrm{QNM}}(M',a')\big|.
\end{equation}
In particular, in terms of the \emph{un-normalized} true equatorial QNMs $\omega_\pm$,
\begin{equation}\label{eq:stability-true-omega}
|(M,a)-(M',a')|
\ \le\ \frac{C_*}{\ell}\,
\Big(
\big|\Re(\omega_{\mathrm{av}}-\omega'_{\mathrm{av}})\big|
+
\big|\Re(\omega_{\mathrm{dif}}-\omega'_{\mathrm{dif}})\big|
\Big),
\end{equation}
where $\omega_{\mathrm{av}},\omega_{\mathrm{dif}}$ are defined in \eqref{eq:avg-dif-true} and primes denote evaluation at $(M',a')$.
\end{enumerate}
\end{theorem}

\begin{proof}
Apply Theorem~\ref{thm:inverse-pseudopoles} to the pseudopole map $F=\mathcal G_\ell$ on $\mathcal K$, which provides a uniform Jacobian lower bound
\eqref{eq:Jacobian-two} and stability \eqref{eq:stability-two} for $\ell\ge \ell_0$.
By Lemma~\ref{lem:C1-close}, the true-QNM map $G=\mathcal G_\ell^{\mathrm{QNM}}$ is $C^1$--close to $F$ on $\mathcal K$ by $\mathcal O(\ell^{-\infty})$.
Choosing $\ell_0$ larger if necessary, we may ensure that \eqref{eq:C1-perturb-small} holds with $\varepsilon_0$ from Lemma~\ref{lem:perturb-IFT}.
Lemma~\ref{lem:perturb-IFT} then yields the Jacobian bound \eqref{eq:Jac-true-lower} and stability \eqref{eq:stability-true}.
Finally, \eqref{eq:stability-true-omega} follows from \eqref{eq:UV-true} by the elementary estimate
$|\Delta U|+|\Delta V|\le \frac{1}{\ell}\big(|\Re\Delta\omega_{\mathrm{av}}|+|\Re\Delta\omega_{\mathrm{dif}}|\big)$.

To see the \emph{real-analytic} statement in (1), note that for $\ell\ge \ell_0$ the labeled QNM branches
$\omega_\pm(M,a)$ depend real-analytically on $(M,a)$ on $\mathcal K$ by Proposition~\ref{prop:labeling-analytic}
(simple poles + analytic perturbation in the Kerr--de~Sitter Fredholm framework).  Hence
$\mathcal G_\ell^{\mathrm{QNM}}$ is real-analytic, and since its Jacobian is nonvanishing by \eqref{eq:Jac-true-lower},
the real-analytic inverse function theorem upgrades the $C^1$ local inverses to real-analytic ones.
\end{proof}

\begin{corollary}[Unlabeled recovery: determination of $(M,|a|)$ from the unordered equatorial pair]\label{cor:unlabeled-recovery}
Under the assumptions of Theorem~\ref{thm:inverse-true-QNM}, there exist $\ell_0\in\mathbb N$ and a constant $C_{\mathrm{un}}>0$ such that for every $\ell\ge \ell_0$ the following holds.
Given the unordered pair of real parts
\[\big\{\Re\omega_{n,\ell,\ell}(M,a),\ \Re\omega_{n,\ell,-\ell}(M,a)\big\},\]
define the unlabeled observables
\begin{equation}\label{eq:UV-unlabeled}
\begin{aligned}
U^{\mathrm{un}}_\ell(M,a)&:=\Re\!\Big(\frac{\omega_{n,\ell,\ell}(M,a)+\omega_{n,\ell,-\ell}(M,a)}{2\ell}\Big),\\
|V|^{\mathrm{un}}_\ell(M,a)&:=\Big|\Re\!\Big(\frac{\omega_{n,\ell,\ell}(M,a)-\omega_{n,\ell,-\ell}(M,a)}{2\ell}\Big)\Big|.
\end{aligned}
\end{equation}
Then for all $(M,a),(M',a')\in\mathcal K$ sufficiently close,
\begin{equation}\label{eq:stability-unlabeled}
\begin{aligned}
|M-M'|+\bigl||a|-|a'|\bigr|
&\le C_{\mathrm{un}}\Big(\,
\big|U^{\mathrm{un}}_\ell(M,a)-U^{\mathrm{un}}_\ell(M',a')\big|\\
&\qquad\qquad\ +\big||V|^{\mathrm{un}}_\ell(M,a)-|V|^{\mathrm{un}}_\ell(M',a')\big|
\,\Big).
\end{aligned}
\end{equation}
In particular, the unordered equatorial pair determines $(M,|a|)$ locally and stably on $\mathcal K$.
\end{corollary}

\begin{proof}
For fixed $\ell$, define $U_\ell^{\mathrm{QNM}},V_\ell^{\mathrm{QNM}}$ as in \eqref{eq:UV-true}.
Since Kerr--de~Sitter has the reflection symmetry $(a,\varphi)\leftrightarrow(-a,-\varphi)$, conjugation by $\varphi\mapsto-\varphi$ identifies the separated $k$--mode problem for $g_{M,a}$ with the $(-k)$--mode problem for $g_{M,-a}$.
Consequently, for labeled QNMs one has
\begin{equation}\label{eq:ak-symmetry-true}
\omega_{n,\ell,k}(M,-a)=\omega_{n,\ell,-k}(M,a),
\end{equation}
and therefore
\begin{equation}\label{eq:parity-true}
U_\ell^{\mathrm{QNM}}(M,-a)=U_\ell^{\mathrm{QNM}}(M,a),\qquad V_\ell^{\mathrm{QNM}}(M,-a)=-V_\ell^{\mathrm{QNM}}(M,a).
\end{equation}

Now fix two nearby points $(M,a),(M',a')\in\mathcal K$.
By the elementary identity $||x|-|y||=\min\{|x-y|,|x+y|\}$, there exists $\sigma\in\{\pm1\}$ such that
\[
\big||V_\ell^{\mathrm{QNM}}(M,a)|-|V_\ell^{\mathrm{QNM}}(M',a')|\big|=\big|V_\ell^{\mathrm{QNM}}(M,a)-\sigma\,V_\ell^{\mathrm{QNM}}(M',a')\big|.
\]
If $\sigma=-1$, replace $a'$ by $-a'$ and use \eqref{eq:parity-true} to rewrite
\[
\bigl(U_\ell^{\mathrm{QNM}}(M',a'),\sigma V_\ell^{\mathrm{QNM}}(M',a')\bigr)
=\bigl(U_\ell^{\mathrm{QNM}}(M',-a'),V_\ell^{\mathrm{QNM}}(M',-a')\bigr).
\]
In either case, Theorem~\ref{thm:inverse-true-QNM} applied to the labeled data yields
\[
\begin{aligned}
|(M,a)-(M',\sigma a')|
&\le C_*\Big(
|U_\ell^{\mathrm{QNM}}(M,a)-U_\ell^{\mathrm{QNM}}(M',a')|\\
&\qquad\qquad\ +\big||V_\ell^{\mathrm{QNM}}(M,a)|-|V_\ell^{\mathrm{QNM}}(M',a')|\big|
\Big).
\end{aligned}
\]
Finally, $\bigl||a|-|a'|\bigr|\le |a-\sigma a'|$, and the right-hand side agrees with the difference of the unlabeled observables in \eqref{eq:UV-unlabeled}.
This gives \eqref{eq:stability-unlabeled} with $C_{\mathrm{un}}$ comparable to $C_*$.
\end{proof}

\begin{corollary}[Single--mode inversion for equatorial true QNMs]\label{cor:single-mode-true}
Under the assumptions of Theorem~\ref{thm:inverse-true-QNM}, there exist $\ell_0\in\mathbb N$ and constants
$c_\sharp,C_\sharp>0$ such that for every $\ell\ge \ell_0$ the single--mode true-QNM map
\begin{equation}\label{eq:SQNM-def}
\mathcal S_\ell^{\mathrm{QNM}}(M,a)
:=
\begin{pmatrix}
\widetilde W_\ell^{\mathrm{QNM}}(M,a)\\[2pt]
U_{\ell,+}^{\mathrm{QNM}}(M,a)
\end{pmatrix}
=
\begin{pmatrix}
-\dfrac{\Im\omega_{n,\ell,\ell}(M,a)}{n+\frac12}\\[6pt]
\Re\dfrac{\omega_{n,\ell,\ell}(M,a)}{\ell}
\end{pmatrix}.
\end{equation}
is a real-analytic local diffeomorphism at every point of $\mathcal K$, with
\begin{equation}\label{eq:Jac-true-single}
\big|\det D\mathcal S_\ell^{\mathrm{QNM}}(M,a)\big|\ge c_\sharp\qquad\forall(M,a)\in\mathcal K
\end{equation}
and the stability estimate
\begin{equation}\label{eq:stab-true-single}
|(M,a)-(M',a')|\le C_\sharp\,\big|\mathcal S_\ell^{\mathrm{QNM}}(M,a)-\mathcal S_\ell^{\mathrm{QNM}}(M',a')\big|
\end{equation}
for all $(M,a),(M',a')\in\mathcal K$ sufficiently close.
\end{corollary}

\begin{proof}
Define the pseudopole single--mode map $\mathcal S_\ell$ by \eqref{eq:single-mode-map}.
By Theorem~\ref{thm:inverse-pseudopoles}(i) and Lemma~\ref{lem:single-mode-expansions}, there exist $\ell_0$ and $c_0,L_0>0$
such that $|\det D\mathcal S_\ell|\ge c_0$ and $\|D\mathcal S_\ell\|\le L_0$ on $\mathcal K$ for all $\ell\ge\ell_0$.

By Lemma~\ref{lem:C1-close} applied to the $+$-branch, the difference
$\omega_{n,\ell,\ell}-\omega^\sharp_{n,\ell,\ell}$ (and one parameter derivative) is $\mathcal O(\ell^{-N})$ uniformly on $\mathcal K$ for every $N$.
Taking real and imaginary parts, and using the fixed normalization factors $1/\ell$ and $(n+\tfrac12)^{-1}$, we obtain
\begin{equation}\label{eq:C1-close-single}
\|\mathcal S_\ell^{\mathrm{QNM}}-\mathcal S_\ell\|_{C^1(\mathcal K)}\le C_N\,\ell^{-N}
\end{equation}
for every $N$.
Choosing $\ell_0$ larger so that \eqref{eq:C1-close-single} is below the perturbation threshold $\varepsilon_0$ in Lemma~\ref{lem:perturb-IFT},
we may apply Lemma~\ref{lem:perturb-IFT} with $F=\mathcal S_\ell$ and $G=\mathcal S_\ell^{\mathrm{QNM}}$ to conclude local invertibility and
the quantitative bounds \eqref{eq:Jac-true-single}--\eqref{eq:stab-true-single}. Real-analyticity follows exactly as in
the proof of Theorem~\ref{thm:inverse-true-QNM}: the labeled branch $\omega_{n,\ell,\ell}(M,a)$ is real-analytic in $(M,a)$ on $\mathcal K$ in the
simple-pole regime (Proposition~\ref{prop:labeling-analytic}), hence so is $\mathcal S_\ell^{\mathrm{QNM}}$.
\end{proof}

\begin{theorem}[Three-parameter inversion for equatorial quasinormal modes]\label{thm:inverse-QNMs-3}
Fix $\nu_0>0$ and the associated integer $C_n$ from Theorem~\ref{thm:Dyatlov-quantization}. Fix $n\in\{0,\dots,C_n\}$ and let $\mathcal K^{(3)}$ be a compact set of subextremal Kerr--de~Sitter parameters as in \eqref{eq:Kcompact-3}.
For each $\ell$ sufficiently large, let $\omega_{\ell,\pm}(M,a,\Lambda)$ denote the two equatorial quasinormal modes with $k=\pm\ell$
labeled as in Proposition~\ref{prop:labeling-analytic}.
Define the average/difference and the associated real-part observables by
\begin{equation}\label{eq:G-QNM-def}
\begin{aligned}
\omega_{\ell,\mathrm{av}}&:=\tfrac12(\omega_{\ell,+}+\omega_{\ell,-}),\qquad
\omega_{\ell,\mathrm{dif}}:=\tfrac12(\omega_{\ell,+}-\omega_{\ell,-}),\\
U_\ell^{\mathrm{QNM}}(M,a,\Lambda)&:=\Re\!\Bigl(\frac{\omega_{\ell,\mathrm{av}}(M,a,\Lambda)}{\ell}\Bigr),\\
V_\ell^{\mathrm{QNM}}(M,a,\Lambda)&:=\Re\!\Bigl(\frac{\omega_{\ell,\mathrm{dif}}(M,a,\Lambda)}{\ell}\Bigr).
\end{aligned}
\end{equation}

\begin{equation}\label{eq:H-QNM-def}
\begin{aligned}
\mathcal H_\ell^{\mathrm{QNM}}(M,a,\Lambda)
:=\Bigl(
&U_\ell^{\mathrm{QNM}}(M,a,\Lambda),\\
&V_\ell^{\mathrm{QNM}}(M,a,\Lambda),\\
&\widetilde W_\ell^{\mathrm{QNM}}(M,a,\Lambda)
\Bigr).
\end{aligned}
\end{equation}
where $U_\ell^{\mathrm{QNM}},V_\ell^{\mathrm{QNM}}$ are defined by \eqref{eq:G-QNM-def} (now viewed as functions of $(M,a,\Lambda)$) and
\begin{equation}\label{eq:Wtilde-QNM-def}
\widetilde W_\ell^{\mathrm{QNM}}(M,a,\Lambda):=-\frac{\Im\,\omega_{\ell,+}(M,a,\Lambda)}{n+\tfrac12}.
\end{equation}
Then there exist $\ell_0\in\mathbb N$ and $C>0$ such that for every $\ell\ge \ell_0$ and every $(M,a,\Lambda)\in\mathcal K^{(3)}$
there is a neighborhood $\mathcal N\subset\mathcal K^{(3)}$ of $(M,a,\Lambda)$ on which the map $\mathcal H_\ell^{\mathrm{QNM}}$ is injective.
Moreover, for all $(M,a,\Lambda),(M',a',\Lambda')\in\mathcal N$,
\begin{equation}\label{eq:H-QNM-stability}
\bigl|(M,a,\Lambda)-(M',a',\Lambda')\bigr|
\le C\,\bigl|\mathcal H_\ell^{\mathrm{QNM}}(M,a,\Lambda)-\mathcal H_\ell^{\mathrm{QNM}}(M',a',\Lambda')\bigr|.
\end{equation}
\end{theorem}

\begin{lemma}[$C^1$ pseudopole--QNM closeness for the three-parameter data map]\label{lem:C1-close-H}
Under the assumptions of Theorem~\ref{thm:inverse-QNMs-3}, for every $N\in\mathbb N$ there exists $C_N>0$ such that for all $\ell$ sufficiently large,
\begin{equation}\label{eq:C1-close-H}
\|\mathcal H_\ell^{\mathrm{QNM}}-\mathcal H_\ell\|_{C^1(\mathcal K^{(3)})}\le C_N\,\ell^{-N},
\end{equation}
where $\mathcal H_\ell$ is the pseudopole map from \eqref{eq:H-map-def}.
The $C^1$ norm in \eqref{eq:C1-close-H} includes one derivative in each of $M,a,\Lambda$.
\end{lemma}

\begin{proof}
By Corollary~\ref{cor:C1-close-Lambda}, the labeled quasinormal modes satisfy
\[
\|\omega_{\ell,\pm}-\omega^\sharp_{\ell,\pm}\|_{C^1(\mathcal K^{(3)})}\le C_N\,\ell^{-N}
\qquad\text{for every }N,
\]
where the $C^1$ norm is taken with respect to $(M,a,\Lambda)$.
Taking real and imaginary parts and applying the fixed normalizations in
\eqref{eq:G-QNM-def}--\eqref{eq:Wtilde-QNM-def} gives
\eqref{eq:C1-close-H}.
\end{proof}

\begin{proof}[Proof of Theorem~\ref{thm:inverse-QNMs-3}]
By Theorem~\ref{thm:inverse-pseudopoles-3} and compactness of $\mathcal K^{(3)}$, there exist $\ell_1\in\mathbb N$, $c_0>0$ and $L_0>0$ such that for all $\ell\ge\ell_1$,
\[
\big|\det D\mathcal H_\ell(M,a,\Lambda)\big|\ge c_0,
\qquad
\|D\mathcal H_\ell(M,a,\Lambda)\|\le L_0,
\qquad \forall (M,a,\Lambda)\in\mathcal K^{(3)}.
\]
Fix $N$ and choose $\ell_0\ge\ell_1$ large so that the perturbation threshold \eqref{eq:C1-perturb-small} holds on $\mathcal K^{(3)}$ with $\varepsilon_0$ from Lemma~\ref{lem:perturb-IFT}; this is possible thanks to \eqref{eq:C1-close-H}.
Applying Lemma~\ref{lem:perturb-IFT} with $d=3$ to $F=\mathcal H_\ell$ and $G=\mathcal H_\ell^{\mathrm{QNM}}$ yields the asserted local injectivity and the stability estimate \eqref{eq:H-QNM-stability}.
\end{proof}

\subsection{Explicit leading reconstruction for true QNM observables}\label{subsec:explicit-true}

Theorem~\ref{thm:inverse-true-QNM} is qualitative/quantitative at the map level. For implementation and for comparison with the $a=0$ case,
it is useful to record an explicit leading reconstruction, mirroring Proposition~\ref{prop:explicit-inverse}.

\begin{proposition}[Leading reconstruction formulas]\label{prop:explicit-true}
Let $\Omega_{\mathrm{ph}}$ and $c_Z$ be defined in \eqref{eq:Omega-cZ-def}.
Define
\[
M^{(0)}(U):=\frac{1}{3\sqrt{\Lambda+3U^2}},\qquad
a^{(0)}(U,V):=\frac{V}{c_Z(M^{(0)}(U))}.
\]
Then there exist $\ell_0\in\mathbb N$ and $C>0$ such that for all $\ell\ge \ell_0$ and all $(M,a)\in\mathcal K$,
\begin{equation}\label{eq:explicit-true-bound}
\begin{aligned}
&\big|M^{(0)}(U_\ell^{\mathrm{QNM}}(M,a))-M\big|
+\big|a^{(0)}\!\big(U_\ell^{\mathrm{QNM}}(M,a),V_\ell^{\mathrm{QNM}}(M,a)\big)-a\big|\\
&\qquad\le\ C\,(a^2+\ell^{-1}).
\end{aligned}
\end{equation}
\end{proposition}

\begin{proof}
By Lemma~\ref{lem:C1-close} (in particular its $C^0$ component) and Lemma~\ref{lem:U-V-expansions},
\[
U_\ell^{\mathrm{QNM}}(M,a)=U_\ell(M,a)+\mathcal O(\ell^{-\infty})
=\Omega_{\mathrm{ph}}(M)+\mathcal O(a^2)+\mathcal O(\ell^{-1}),
\]
uniformly on $\mathcal K$, and similarly
\[
V_\ell^{\mathrm{QNM}}(M,a)=V_\ell(M,a)+\mathcal O(\ell^{-\infty})
=c_Z(M)\,a+\mathcal O(a^3)+\mathcal O(\ell^{-1}).
\]
Since $\Omega_{\mathrm{ph}}$ is strictly monotone on $0<9\Lambda M^2<1$ by Lemma~\ref{lem:Omega-monotone},
its inverse is $C^1$ on $\pi_M(\mathcal K)$, hence the first estimate in \eqref{eq:explicit-true-bound} follows
by Taylor expansion of $M^{(0)}(\cdot)=(\Omega_{\mathrm{ph}})^{-1}(\cdot)$ around $\Omega_{\mathrm{ph}}(M)$.
The second estimate follows similarly, using that $c_Z(M)$ is bounded away from $0$ on $\mathcal K$ and substituting
$M^{(0)}(U_\ell^{\mathrm{QNM}})=M+\mathcal O(a^2+\ell^{-1})$ into $c_Z$.
\end{proof}

\begin{remark}[Newton refinement]\label{rem:newton-true}
Proposition~\ref{prop:explicit-true} provides a natural initial guess for solving the nonlinear system
$\mathcal G_\ell^{\mathrm{QNM}}(M,a)=(U,V)$ from data.
Since $D\mathcal G_\ell^{\mathrm{QNM}}$ is uniformly invertible on $\mathcal K$ by \eqref{eq:Jac-true-lower},
a standard Newton scheme initialized at $(M^{(0)}(U),a^{(0)}(U,V))$ converges quadratically in a neighborhood of the true parameters,
provided $\ell$ is large and $|a|$ is within the slow-rotation regime.
\end{remark}

\subsection{Deterministic noise-to-parameter bounds}\label{subsec:noise}

Assume we observe approximate equatorial QNM frequencies
\begin{equation}\label{eq:measured-omega}
\widetilde\omega_\pm=\omega_\pm+\eta_\pm,
\end{equation}
where $\eta_\pm\in\mathbb C$ represent measurement error.
Define the corresponding measured observables
\[
\widetilde U_\ell^{\mathrm{QNM}}:=\Re\!\Big(\frac{\widetilde\omega_{\mathrm{av}}}{\ell}\Big),\qquad
\widetilde V_\ell^{\mathrm{QNM}}:=\Re\!\Big(\frac{\widetilde\omega_{\mathrm{dif}}}{\ell}\Big),
\]
\[
(\widetilde\omega_{\mathrm{av}},\widetilde\omega_{\mathrm{dif}})
:=\Big(\frac{\widetilde\omega_++\widetilde\omega_-}{2},\ \frac{\widetilde\omega_+-\widetilde\omega_-}{2}\Big).
\]

\begin{corollary}[Noise-to-parameter estimate]\label{cor:noise-bound}
Under the assumptions of Theorem~\ref{thm:inverse-true-QNM}, there exist $\ell_0\in\mathbb N$ and $C>0$ such that for all $\ell\ge\ell_0$
the following holds.
Let $(M,a)\in\mathcal K$ and suppose $\widetilde\omega_\pm$ are given by \eqref{eq:measured-omega} with errors $\eta_\pm$ satisfying
\begin{equation}\label{eq:noise-bound}
|\Re\eta_+|+|\Re\eta_-|\le \varepsilon,
\end{equation}
for $\varepsilon>0$ sufficiently small (depending only on $\mathcal K$).
Let $(\widetilde M,\widetilde a)\in\mathcal K$ be defined by the local inverse of $\mathcal G_\ell^{\mathrm{QNM}}$,
\[
(\widetilde M,\widetilde a):=\big(\mathcal G_\ell^{\mathrm{QNM}}\big)^{-1}\!\big(\widetilde U_\ell^{\mathrm{QNM}},\widetilde V_\ell^{\mathrm{QNM}}\big).
\]
Then
\begin{equation}\label{eq:noise-to-parameter}
|(\widetilde M,\widetilde a)-(M,a)|
\ \le\ \frac{C}{\ell}\,\varepsilon.
\end{equation}
Moreover, if one reconstructs using the pseudopole inverse,
\[
(\widetilde M^\sharp,\widetilde a^\sharp):=\big(\mathcal G_\ell\big)^{-1}\!\big(\widetilde U_\ell^{\mathrm{QNM}},\widetilde V_\ell^{\mathrm{QNM}}\big),
\]
then for every $N\in\mathbb N$ there exists $C_N>0$ such that
\begin{equation}\label{eq:noise-plus-model}
|(\widetilde M^\sharp,\widetilde a^\sharp)-(M,a)|
\ \le\ \frac{C}{\ell}\,\varepsilon\ +\ C_N\,\ell^{-N}.
\end{equation}
\end{corollary}

\begin{proof}
By Theorem~\ref{thm:inverse-true-QNM} we have the stability estimate \eqref{eq:stability-true}.
Since $\widetilde U_\ell^{\mathrm{QNM}}-U_\ell^{\mathrm{QNM}}=\Re(\eta_{\mathrm{av}})/\ell$ and
$\widetilde V_\ell^{\mathrm{QNM}}-V_\ell^{\mathrm{QNM}}=\Re(\eta_{\mathrm{dif}})/\ell$, with
$\eta_{\mathrm{av}}=\frac12(\eta_++\eta_-)$ and $\eta_{\mathrm{dif}}=\frac12(\eta_+-\eta_-)$,
we obtain
\[
\begin{aligned}
\big|(\widetilde U_\ell^{\mathrm{QNM}},\widetilde V_\ell^{\mathrm{QNM}})
-(U_\ell^{\mathrm{QNM}},V_\ell^{\mathrm{QNM}})\big|
&\le \frac{1}{\ell}\Big(|\Re\eta_{\mathrm{av}}|+|\Re\eta_{\mathrm{dif}}|\Big)\\
&\le \frac{|\Re\eta_+|+|\Re\eta_-|}{\ell},
\end{aligned}
\]
which yields \eqref{eq:noise-to-parameter}.
For \eqref{eq:noise-plus-model}, write
\[
\begin{aligned}
(\widetilde M^\sharp,\widetilde a^\sharp)-(M,a)
={}&\big(\mathcal G_\ell^{-1}- (\mathcal G_\ell^{\mathrm{QNM}})^{-1}\big)
\big(\widetilde U_\ell^{\mathrm{QNM}},\widetilde V_\ell^{\mathrm{QNM}}\big)\\
&\quad+\big((\mathcal G_\ell^{\mathrm{QNM}})^{-1}\big(\widetilde U_\ell^{\mathrm{QNM}},\widetilde V_\ell^{\mathrm{QNM}}\big)-(M,a)\big).
\end{aligned}
\]
The second term is controlled by \eqref{eq:noise-to-parameter}. The first term is a modeling term: by Lemma~\ref{lem:C1-close} and Lemma~\ref{lem:perturb-IFT},
the inverses exist on a common neighborhood and differ by $\mathcal O(\ell^{-N})$ in $C^0$ on that neighborhood for every $N$, giving the extra term $C_N\ell^{-N}$.
\end{proof}

\begin{remark}[Balancing measurement and modeling errors]\label{rem:balance-errors}
The estimate \eqref{eq:noise-plus-model} separates two deterministic contributions.
The first term, $\frac{C}{\ell}\varepsilon$, comes from \emph{measurement error} in the observed real parts
and improves at higher frequency because of the normalization by $\ell$.
The second term, $C_N\ell^{-N}$, is a \emph{modeling error} due to replacing true QNMs by the semiclassical pseudopoles.
Since the pseudopole--QNM discrepancy is super--polynomial in $\ell$, one may, for any prescribed algebraic accuracy, increase $\ell$ until the modeling term is negligible compared with the measurement term.
For instance, if the real-part noise level scales like $\varepsilon=\mathcal O(\ell^{-p})$ for some $p>0$, then any fixed $N>p+1$ makes the modeling term $o(\ell^{-1}\varepsilon)$ as $\ell\to\infty$.
Conversely, if one reconstructs using the true-QNM inverse as in \eqref{eq:noise-to-parameter}, the modeling term disappears and the reconstruction error is governed solely by the measurement accuracy.
\end{remark}

\section{Discussion and outlook}\label{sec:discussion}

This paper establishes a concrete high--frequency inverse mechanism on the Kerr--de~Sitter family: a \emph{finite} subset of \emph{labeled} quasinormal modes determines geometric parameters locally, with quantitative stability improving as $\ell\to\infty$.  We treat two settings.  For fixed cosmological constant $\Lambda>0$ we prove that equatorial real parts alone determine $(M,a)$ locally and stably.  Allowing $\Lambda$ to vary in a compact interval, we show in Section~\ref{sec:three-parameter} that adding one damping observable yields a three--parameter inverse theorem recovering $(M,a,\Lambda)$ locally away from $a=0$.
From the analytical perspective, the key point is that the resonance problem for Kerr--de~Sitter
admits a robust Fredholm formulation and a semiclassical description of the trapped dynamics,
so that one can pass between geometry and spectral data with controlled errors; see
\cite{Vasy2013MicrolocalAHKdS,Dyatlov2011QNMKerrDeSitter,Dyatlov2012AsymptoticQNMKdS}
for foundational developments in this direction.

\subsection{Mechanism and structural inputs}\label{subsec:discussion-mechanism}

Let $\omega_{n,\ell,k}(M,a)$ be the (scalar) Kerr--de~Sitter QNMs in the strip $\Im\omega>-\nu_0$,
labeled by overtone $n$, angular momentum $\ell$, and azimuthal index $k$. In
Sections~\ref{sec:HF-quantization}--\ref{sec:transfer} we exploit two complementary features.

First, for $\ell\gg 1$ one has a semiclassical \emph{labeling} by $(n,\ell,k)$ and an explicit
pseudopole approximation $\omega^\sharp_{n,\ell,k}$ with super--polynomial remainder, following the
Bohr--Sommerfeld quantization paradigm for Kerr--de~Sitter \cite{Dyatlov2012AsymptoticQNMKdS}.
Second, the meromorphic Fredholm framework for the spectral family (and its perturbative stability)
allows one to transfer results from pseudopoles to true QNMs, with parameter control; this builds on
the microlocal approach developed in \cite{Vasy2013MicrolocalAHKdS} and the Kerr--de~Sitter resonance
analysis in \cite{Dyatlov2011QNMKerrDeSitter}. The $C^1$ control in parameters needed to transfer Jacobian bounds
from pseudopoles to true QNMs (Lemma~\ref{lem:C1-close}) comes from the parameter-differentiable semiclassical
quantization construction (Appendix~\ref{app:parameter-quantization}), together with holomorphy in the spectral
variable and Cauchy estimates.  Separately, in the simple-pole regime, analytic perturbation theory yields a
locally real-analytic labeling of QNM branches (Remark~\ref{rem:analytic_parameters}), which upgrades $C^1$ local
inverses to real-analytic ones.

The particular two--mode package used in Theorem~\ref{thm:inverse-true-QNM} (equatorial $k=\pm\ell$
with fixed $n$) isolates the most robust geometric signal of rotation: the Kerr splitting of the
photon--sphere frequencies. The reflection symmetry $(a,k)\mapsto(-a,-k)$ forces the normalized
average to be even in $a$ and the normalized difference to be odd in $a$, yielding a triangular
Jacobian and an explicit leading inverse (Proposition~\ref{prop:explicit-true}). The resulting
conditioning is natural: the real parts of high--$\ell$ photon--sphere frequencies scale like
$\Re\omega\sim \ell$, so the normalized observables $\Re(\omega/\ell)$ are order one, and the inverse
stability gains a factor $1/\ell$ when expressed in terms of unnormalized frequencies.

\subsection{Minimal mode packages and a single--mode alternative}\label{subsec:discussion-single-mode}

Theorem~\ref{thm:inverse-true-QNM} uses the pair $\{\omega_{n,\ell,\ell},\omega_{n,\ell,-\ell}\}$ and
only their real parts, which makes the parity structure transparent and yields explicit constants.
From the inverse--problems viewpoint it is natural to ask whether one can reduce the amount of data.

At the symbolic level, a single complex mode already provides two real constraints. The following
statement formulates a rigorous (and flexible) one--mode recovery principle; it should be viewed as
a complement to the explicit two--mode inversion, and as a convenient interface for future work in
which the damping rate can be tracked effectively.

\begin{proposition}[Single--mode recovery under a Jacobian nondegeneracy]\label{prop:single-mode}
Fix an overtone index $n\ge 0$ and a choice of sign $\varepsilon\in\{+,-\}$, and consider the equatorial mode family $k=\varepsilon\ell$.
Let $(M_*,a_*)$ be a subextremal base point.
Assume that for all sufficiently large $\ell$ the pseudopole map
\[
(M,a)\longmapsto \Big(\Re(\omega^\sharp_{n,\ell,\varepsilon\ell}(M,a)/\ell),\ \Im(\omega^\sharp_{n,\ell,\varepsilon\ell}(M,a)/\ell)\Big)
\]
has nonvanishing Jacobian determinant at $(M_*,a_*)$, i.e.
\[
\det \partial_{(M,a)}\Big(\Re(\omega^\sharp_{n,\ell,\varepsilon\ell}/\ell),\Im(\omega^\sharp_{n,\ell,\varepsilon\ell}/\ell)\Big)\Big|_{(M_*,a_*)}\neq 0.
\]
Then there exist $\ell_0$ and a neighborhood $\mathcal U$ of $(M_*,a_*)$ such that for all $\ell\ge\ell_0$
the true QNM map
\[
(M,a)\longmapsto \Big(\Re(\omega_{n,\ell,\varepsilon\ell}(M,a)/\ell),\ \Im(\omega_{n,\ell,\varepsilon\ell}(M,a)/\ell)\Big)
\]
is locally invertible on $\mathcal U$, with a stability estimate
\[
|(M,a)-(M',a')|
\le C\Big|\frac{\omega_{n,\ell,\varepsilon\ell}(M,a)-\omega_{n,\ell,\varepsilon\ell}(M',a')}{\ell}\Big|
+ O(\ell^{-N})
\]
for every fixed $N$, uniformly for $(M,a),(M',a')\in\mathcal U$.
\end{proposition}

\begin{proof}
By the inverse function theorem, the stated Jacobian condition implies local invertibility of the
pseudopole map near $(M_*,a_*)$ for each sufficiently large $\ell$, with a stability bound depending
on a lower bound for the Jacobian determinant.
Lemma~\ref{lem:C1-close} gives $C^1$ super--polynomial control of the difference between the QNM map
and the pseudopole map in $(M,a)$ (in particular, uniform control of the Jacobians up to $O(\ell^{-N})$
for every $N$). For $\ell$ large, this preserves the nondegeneracy of the Jacobian and yields local
invertibility of the QNM map with the same structure of stability, up to $O(\ell^{-N})$ errors.
\end{proof}

In practice, Proposition~\ref{prop:single-mode} suggests more broadly that one can attempt to use any mode family for which the
pair of geometric invariants governing high--frequency asymptotics (real frequency and damping, i.e.\ angular frequency and Lyapunov exponent of the relevant trapped orbit) depends nondegenerately on $(M,a)$, provided the corresponding $C^1$ pseudopole approximation is available.

The explicit equatorial two--mode package in the main theorems remains attractive because it enforces the
$a$--parity at the level of observables and therefore avoids having to quantify the $a$--dependence of the
damping rate.

\subsection{Reconstruction as a stable pipeline}\label{subsec:discussion-pipeline}

The inverse theorems can be read as a practical reconstruction pipeline in the high--frequency regime.
Fix $n$ and choose $\ell\gg 1$. Given (measured or computed) equatorial modes
$\omega_\pm\approx \omega_{n,\ell,\pm\ell}(M,a)$, define the normalized observables
\[
U:=\Re\frac{\omega_++\omega_-}{2\ell},\qquad
V:=\Re\frac{\omega_+-\omega_-}{2\ell}.
\]
Proposition~\ref{prop:explicit-true} yields a closed--form initialization
\[
M^{(0)}=\frac{1}{3\sqrt{\Lambda+3U^2}},\qquad
a^{(0)}=\frac{V}{c_Z(M^{(0)})},\qquad c_Z(M)=\frac{2+9\Lambda M^2}{27M^2},
\]
and Theorem~\ref{thm:inverse-true-QNM} provides uniform Jacobian bounds ensuring that Newton or
quasi--Newton refinement converges rapidly in a neighborhood of the true parameters. The error
decomposition in Corollary~\ref{cor:noise-bound} is conceptually useful: it separates the intrinsic
modeling error (pseudopoles versus QNMs, super--polynomial in $\ell$) from external errors in the
frequency extraction, thus allowing deterministic propagation of uncertainty.

\subsection{Extensions: more parameters, full subextremal rotation, spin, and \texorpdfstring{$\Lambda\to 0$}{Lambda->0}}\label{subsec:discussion-extensions}

\paragraph{Adding parameters beyond \texorpdfstring{$(M,a)$}{(M,a)}.}
Section~\ref{sec:three-parameter} already provides a three--parameter recovery result for $(M,a,\Lambda)$:
in a compact slow--rotation regime and on sets bounded away from $a=0$, the three real observables
\[
(U,V,\widetilde W)\,=\,\Bigl(\Re\frac{\omega_++\omega_-}{2\ell},\ \Re\frac{\omega_+-\omega_-}{2\ell},\ -\frac{\Im\,\omega_+}{n+\frac12}\Bigr)
\]
determine $(M,a,\Lambda)$ locally, with a quantitative stability estimate.  The restriction $|a|\ge a_1>0$
is structural: at $a=0$ the splitting observable $V$ vanishes to first order and the Jacobian of the geometric
three--data map degenerates, so identifiability requires either additional mode information (e.g.\ a second overtone)
or an observable with independent leading sensitivity to $\Lambda$.
Beyond $\Lambda$, one may consider adding further parameters (such as electric charge in Kerr--Newman--de~Sitter);
the present approach suggests a general principle: each additional independent geometric parameter requires an additional
independent high--frequency observable, together with a uniform nondegeneracy mechanism for the associated Jacobian.

\paragraph{Removing the slow--rotation restriction.}
Our use of slow rotation is tied to the availability of a particularly clean quantization/labelling and
transparent first--order splitting. However, the Kerr--de~Sitter null geodesic flow is integrable across
the full subextremal range, and the modern Fredholm theory of QNMs (including resonance expansions) is now
available for the full subextremal family in the linear scalar setting; see \cite{PetersenVasy2025WaveKdS}
and the real-analyticity results for resonant states under analytic coefficients \cite{PetersenVasy2023AnalyticityQNM}. This suggests that the main remaining
obstruction to a full--range inverse theorem is not conceptual but technical: uniform symbolic expansions
and separation of labeled branches in $(M,a)$, together with the control of possible mode crossings.

\paragraph{Spin--$2$ (Teukolsky) and gravitational ringdown.}
From the viewpoint of semiclassical trapping, the leading frequency map for separated high--frequency modes
is governed by the same null--geodesic Hamiltonian, so one expects the geometric part of the inverse
mechanism to persist for the Teukolsky equation \cite{Teukolsky1973PerturbationsKerr}.
The main new difficulties are analytic and structural: controlling gauge--related degeneracies and tracking
spin--dependent subprincipal effects. Mode stability results in the asymptotically flat setting (e.g.\
\cite{Whiting1989ModeStabilityKerr}) and the Kerr--de~Sitter nonlinear stability program
\cite{HintzVasy2018NonlinearStabilityKdS} give strong motivation for developing spin--$2$ analogues of the
present inverse theory.

\paragraph{The \texorpdfstring{$\Lambda\to 0$}{Lambda->0} limit and asymptotically flat Kerr.}
Kerr--de~Sitter provides a convenient compactified scattering picture with two radial horizons.
In the astrophysically relevant limit $\Lambda=0$, the exterior is asymptotically flat, and the resonance
picture must be reconciled with long--range scattering and late--time tails. Recent progress on defining
and locating Kerr QNMs directly in the asymptotically flat setting appears in \cite{Stucker2024QNMKerr,GajicWarnick2024QNMKerr}.
A robust route from the present inverse theorem to Kerr would be to establish \emph{uniform} control in
$\Lambda$ for the forward resonant theory and then study the limiting behavior of suitably chosen resonances
as $\Lambda\to 0^+$; this is a natural meeting point of semiclassical analysis and black--hole scattering.

\subsection{Rigidity and global questions}\label{subsec:discussion-rigidity}

Theorem~\ref{thm:inverse-true-QNM} proves a high--frequency, finite--data version of spectral rigidity:
two labeled modes at large $\ell$ already determine $(M,a)$ locally (up to the trivial symmetry
$(a,k)\leftrightarrow(-a,-k)$). A deeper open problem is the global structure of the parameter--to--spectrum map.

\begin{quote}
\textbf{Conjecture (QNM rigidity for Kerr--de~Sitter).}
For fixed $\Lambda>0$, the full QNM set $\mathrm{Res}(P_{M,a})$ determines $(M,a)$ within the subextremal
Kerr--de~Sitter family, modulo $(a,k)\leftrightarrow(-a,-k)$.
\end{quote}

From the microlocal perspective, global injectivity amounts to understanding monodromy of labeled branches,
excluding nontrivial loops in parameter space preserving the resonance set, and identifying robust spectral
invariants (surface gravities, trapping exponents, etc.). The resolvent technology for normally hyperbolic
trapping \cite{WunschZworski2011ResolventNHT,Dyatlov2016SpectralGapsNHT} and the Kerr--de~Sitter Fredholm
framework \cite{Vasy2013MicrolocalAHKdS,Dyatlov2011QNMKerrDeSitter} provide the analytic backbone on which
such rigidity questions can be posed precisely.

\medskip

In summary, the present work isolates a particularly clean interface between the semiclassical geometry of
trapping and parameter identification: high--frequency QNMs encode $(M,a)$ stably, already at the level of a
finite labeled mode package. Extending this mechanism to broader parameter families, to full subextremal
rotation, and to spin--$2$ perturbations appears to be a concrete and tractable program with strong
connections to current developments in black--hole scattering and stability theory.

\appendix

\section{Equatorial circular photon orbits and the linear Zeeman splitting}\label{app:equatorial-photon}

This appendix provides a concrete derivation of the linear-in-$a$ equatorial splitting used in
Lemma~\ref{lem:principal-splitting}.  The computation is classical (geometric optics): the real principal
frequency map is determined by the characteristic set of the wave operator, and in particular by the
relation between energy and azimuthal angular momentum on the equatorial trapped null geodesics.

\subsection{Metric coefficients on the equatorial plane}

We work in the Boyer--\allowbreak Lindquist chart of Section~\ref{subsec:kds}.  On the equatorial plane $\theta=\pi/2$
we have $\sin\theta=1$, $\cos\theta=0$, hence $\rho^2=r^2$ and $\Delta_\theta=1$.  Expanding
\eqref{eq:kds_metric_BL} gives
\begin{equation}\label{eq:eq-metric-coeffs}
\begin{aligned}
g_{tt}(r,a)&=\frac{-\Delta_r(r)+a^2}{r^2},\\
g_{t\varphi}(r,a)&=\frac{a}{r^2\,\Xi}\bigl(\Delta_r(r)-(r^2+a^2)\bigr),\\
g_{\varphi\varphi}(r,a)&=\frac{(r^2+a^2)^2-a^2\Delta_r(r)}{r^2\,\Xi^2}.
\end{aligned}
\end{equation}
where $\Delta_r$ and $\Xi$ are as in \eqref{eq:Delta_r} and \eqref{eq:Xi}.
Note that $\Xi=1+\mathcal O(a^2)$ has no linear term in $a$.

\subsection{Null circular orbits: a general stationary-axisymmetric criterion}

Consider a curve of the form
\begin{equation}\label{eq:circular-ansatz}
\gamma(t)=(t,\ r,\ \pi/2,\ \varphi(t)),\qquad \dot\varphi(t)=\Omega,
\end{equation}
with constant radius $r$ and constant angular velocity $\Omega=d\varphi/dt$.
The tangent vector is $\dot\gamma=\partial_t+\Omega\,\partial_\varphi$.
Define the function
\begin{equation}\label{eq:Phi-def}
\Phi(r,\Omega,a):=g_{tt}(r,a)+2\Omega\,g_{t\varphi}(r,a)+\Omega^2 g_{\varphi\varphi}(r,a).
\end{equation}
Then the curve \eqref{eq:circular-ansatz} is null if and only if $\Phi(r,\Omega,a)=0$.
Moreover, the condition that \eqref{eq:circular-ansatz} be a (null) geodesic with constant $r$ is
equivalent to the additional stationarity condition
\begin{equation}\label{eq:circular-conditions}
\Phi(r,\Omega,a)=0,\qquad \partial_r\Phi(r,\Omega,a)=0;
\end{equation}
see for instance the standard derivation from the conserved quantities $E=-p_t$ and $L=p_\varphi$
and the radial effective potential in stationary-axisymmetric spacetimes.

\subsection{The Schwarzschild--de~Sitter photon sphere and its perturbation}

At $a=0$ the metric reduces to Schwarzschild--de~Sitter.  In this case $g_{t\varphi}(r,0)=0$, and
\eqref{eq:eq-metric-coeffs} yields
\begin{equation}\label{eq:sdS-coeffs}
g_{tt}(r,0)=-f(r),\qquad g_{\varphi\varphi}(r,0)=r^2,\qquad
f(r):=1-\frac{2M}{r}-\frac{\Lambda r^2}{3}.
\end{equation}
Thus \eqref{eq:circular-conditions} becomes
\begin{equation}\label{eq:sdS-circle}
-f(r)+\Omega^2 r^2=0,\qquad -f'(r)+2\Omega^2 r=0.
\end{equation}
Eliminating $\Omega^2$ gives $r f'(r)-2f(r)=0$, which reduces to $r=3M$.
Substituting $r=3M$ into the first equation in \eqref{eq:sdS-circle} yields
\begin{equation}\label{eq:Omega-ph-from-Phi}
\Omega^2=\frac{f(3M)}{(3M)^2}=\frac{1-9\Lambda M^2}{27M^2},
\qquad\text{hence}\qquad
\Omega_{\mathrm{ph}}(M)=\frac{\sqrt{1-9\Lambda M^2}}{3\sqrt3\,M}.
\end{equation}

We next perturb \eqref{eq:circular-conditions} for small $a$.
Let $r_0:=3M$ and let $\Omega_0^{\pm}:=\pm\Omega_{\mathrm{ph}}(M)$ denote the two SdS angular velocities.
Introduce the map
\begin{equation}\label{eq:F-system}
\mathcal F(r,\Omega,a):=\bigl(\Phi(r,\Omega,a),\ \partial_r\Phi(r,\Omega,a)\bigr).
\end{equation}
At $a=0$ we have $\mathcal F(r_0,\Omega_0^{\pm},0)=(0,0)$.
Moreover, the Jacobian of $(r,\Omega)\mapsto\mathcal F(r,\Omega,0)$ at $(r_0,\Omega_0^{\pm})$ is invertible:
since $\partial_r\Phi(r_0,\Omega_0^{\pm},0)=0$ by construction,
\begin{equation}\label{eq:IFT-det}
\det D_{(r,\Omega)}\mathcal F(r_0,\Omega_0^{\pm},0)
=\bigl(\partial_\Omega\Phi\bigr)(r_0,\Omega_0^{\pm},0)\cdot\bigl(\partial_r^2\Phi\bigr)(r_0,\Omega_0^{\pm},0).
\end{equation}
Here $\partial_\Omega\Phi(r,\Omega,0)=2\Omega r^2$, so $\partial_\Omega\Phi(r_0,\Omega_0^{\pm},0)\neq 0$.
Further, using \eqref{eq:sdS-coeffs} and keeping $\Omega$ fixed when differentiating in $r$, we find
\begin{equation}\label{eq:drrPhi}
\partial_r^2\Phi(r,\Omega,0)=\partial_r^2 g_{tt}(r,0)+2\Omega^2
=\Bigl(\frac{4M}{r^3}+\frac{2\Lambda}{3}\Bigr)+2\Omega^2.
\end{equation}
At $r=r_0=3M$ and $\Omega^2=(\Omega_{\mathrm{ph}}(M))^2$ from \eqref{eq:Omega-ph-from-Phi}, this gives
\begin{equation}\label{eq:drrPhi-nonzero}
\partial_r^2\Phi(r_0,\Omega_0^{\pm},0)=\frac{2}{9M^2}>0.
\end{equation}
Thus the implicit function theorem provides $a_*\!>0$ and smooth functions
$a\mapsto (r_\pm(a),\Omega_\pm(a))$ for $|a|\le a_*$ satisfying \eqref{eq:circular-conditions}.

\subsection{Linearization and the Zeeman coefficient}

We compute the linear term in $\Omega_\pm(a)$.
Since $g_{tt}$ and $g_{\varphi\varphi}$ are even in $a$ (and $\Xi=1+\mathcal O(a^2)$), the only linear contribution
to $\Phi$ in \eqref{eq:Phi-def} comes from $g_{t\varphi}$.
Differentiating \eqref{eq:eq-metric-coeffs} at $a=0$ gives
\begin{equation}\label{eq:gtphi-linear}
\begin{aligned}
\partial_a g_{t\varphi}(r,0)
&=\frac{\Delta_r(r;M,0)-r^2}{r^2}\\
&=-\Bigl(\frac{2M}{r}+\frac{\Lambda r^2}{3}\Bigr).
\end{aligned}
\end{equation}
Differentiate the identity
\[
\Phi(r_\pm(a),\Omega_\pm(a),a)=0
\]
at $a=0$.
Since
\[
\partial_r\Phi(r_0,\Omega_0^{\pm},0)=0
\]
by \eqref{eq:sdS-circle}, the term involving $r_\pm'(0)$ drops out, and we obtain
\begin{equation}\label{eq:Omega-derivative}
0=\partial_\Omega\Phi(r_0,\Omega_0^{\pm},0)\,\Omega_\pm'(0)+\partial_a\Phi(r_0,\Omega_0^{\pm},0).
\end{equation}
Using $\partial_\Omega\Phi(r_0,\Omega_0^{\pm},0)=2\Omega_0^{\pm} r_0^2$ and
$\partial_a\Phi(r_0,\Omega_0^{\pm},0)=2\Omega_0^{\pm}\,\partial_a g_{t\varphi}(r_0,0)$, we find
\begin{equation}\label{eq:Omega-prime-formula}
\Omega_\pm'(0)=-\frac{\partial_a g_{t\varphi}(r_0,0)}{r_0^2}.
\end{equation}
Substituting \eqref{eq:gtphi-linear} at $r_0=3M$ yields
\begin{equation}\label{eq:Omega-prime-evaluated}
\Omega_\pm'(0)=\frac{\frac{2M}{3M}+\frac{\Lambda (3M)^2}{3}}{(3M)^2}
=\frac{\frac{2}{3}+3\Lambda M^2}{9M^2}
=\frac{2+9\Lambda M^2}{27M^2}
=:c_Z(M).
\end{equation}
In particular,
\begin{equation}\label{eq:Omega-expansion}
\Omega_\pm(M,a)=\pm\Omega_{\mathrm{ph}}(M)+c_Z(M)\,a\pm c_{\Omega,2}(M)\,a^2+\mathcal O(a^3),
\qquad a\to0,
\end{equation}
where $c_{\Omega,2}(M)$ is given explicitly by \eqref{eq:cOmega2}.

\subsection{Second-order expansion of the equatorial photon orbit}\label{subsec:photon-second-order}

For the inverse results in Sections~\ref{sec:inverse-pseudopoles} and~\ref{sec:three-parameter} it is useful to have the explicit $a^2$--correction
to the equatorial photon-orbit frequency.  We now record this correction (and, for completeness, the corresponding $a^2$--correction to the orbit radius).

We work with the co-rotating branch $(r_+,\Omega_+)$; the counter-rotating branch follows by the $(a,\varphi)$--reflection symmetry discussed in Lemma~\ref{lem:geo-parity}:
\begin{equation}\label{eq:Omega-symmetry-appendix}
r_-(M,a)=r_+(M,-a),\qquad \Omega_-(M,a)=-\Omega_+(M,-a).
\end{equation}

\begin{proposition}[Second-order coefficients]\label{prop:second-order-photon}
On compact slow-rotation sets, the equatorial photon orbit admits the expansions
\begin{equation}\label{eq:r-second-order}
r_+(M,a)=3M-\frac{2}{\sqrt{3}}\sqrt{1-9\Lambda M^2}\,a-\Bigl(3\Lambda M+\frac{2}{9M}\Bigr)a^2+\mathcal O(a^3),
\qquad a\to0,
\end{equation}
and
\begin{equation}\label{eq:Omega-second-order}
\Omega_+(M,a)=\Omega_{\mathrm{ph}}(M)+c_Z(M)\,a+c_{\Omega,2}(M)\,a^2+\mathcal O(a^3),
\qquad a\to0,
\end{equation}
where $c_{\Omega,2}(M)$ is given by \eqref{eq:cOmega2}.
\end{proposition}

\begin{proof}
We expand the system \eqref{eq:geo-system} at $a=0$.
First, we record the equatorial metric coefficients in the slow-rotation regime.  Writing $f(r)=1-\frac{2M}{r}-\frac{\Lambda r^2}{3}$,
a Taylor expansion of the Kerr--de~Sitter coefficients at $a=0$ gives
\begin{equation}\label{eq:metric-a2}
\begin{aligned}
g_{tt}(r,a)&=-f(r)+\frac{\Lambda}{3}\,a^2+\mathcal O(a^3),\\
g_{t\varphi}(r,a)&=-a\Bigl(\frac{2M}{r}+\frac{\Lambda r^2}{3}\Bigr)+\mathcal O(a^3),\\
g_{\varphi\varphi}(r,a)&=r^2+a^2\Bigl(1-\frac{\Lambda r^2}{3}+\frac{2M}{r}\Bigr)+\mathcal O(a^4),
\end{aligned}
\end{equation}
uniformly for $r$ in compact subsets of $(r_-,r_+)$.
(Here the $\mathcal O(a^2)$ correction in $\Xi^{-1}= (1+\Lambda a^2/3)^{-1}$ has been included in the displayed $a^2$ terms.)

Set
\[
r_+(M,a)=r_0+r_1 a+r_2 a^2+\mathcal O(a^3),\qquad
\Omega_+(M,a)=\Omega_0+\Omega_1 a+\Omega_2 a^2+\mathcal O(a^3),
\]
with $(r_0,\Omega_0)=(3M,\Omega_{\mathrm{ph}}(M))$.
Let $\Phi_{\mathrm{geo}}(r,\Omega;M,a)$ be defined by \eqref{eq:Phi-def}.
Using \eqref{eq:metric-a2} and expanding $\Phi_{\mathrm{geo}}$ at $a=0$
gives, at fixed $(r,\Omega)$,
\begin{equation*}
\begin{aligned}
\Phi_{\mathrm{geo}}(r,\Omega;M,a)
&=\Phi_0(r,\Omega)+a\,\Phi_1(r,\Omega)+a^2\,\Phi_2(r,\Omega)+\mathcal O(a^3),\\
\Phi_0(r,\Omega)&=-f(r)+r^2\Omega^2,\\
\Phi_1(r,\Omega)&=-2\Omega\Bigl(\frac{2M}{r}+\frac{\Lambda r^2}{3}\Bigr),\\
\Phi_2(r,\Omega)&=\frac{\Lambda}{3}+\Omega^2\Bigl(1-\frac{\Lambda r^2}{3}+\frac{2M}{r}\Bigr).
\end{aligned}
\end{equation*}

The circular-orbit system \eqref{eq:geo-system} reads
\[
\Phi_{\mathrm{geo}}\bigl(r_+(a),\Omega_+(a);M,a\bigr)=0,
\qquad
\partial_r\Phi_{\mathrm{geo}}\bigl(r_+(a),\Omega_+(a);M,a\bigr)=0.
\]
At order $a^0$ this reproduces $r_0=3M$ and $\Omega_0=\Omega_{\mathrm{ph}}(M)$.
For brevity, set
\[S:=\sqrt{1-9\Lambda M^2}>0,\qquad \Omega_{\mathrm{ph}}(M)=\frac{S}{3\sqrt3\,M}.
\]

\smallskip
\noindent\emph{First order.}
Differentiate the identity $\Phi_{\mathrm{geo}}(r_+(a),\Omega_+(a);M,a)=0$ at $a=0$.
Using $\Phi_r(r_0,\Omega_0;M,0)=0$ and the chain rule gives
\[
0=\Phi_r\,r_1+\Phi_\Omega\,\Omega_1+\Phi_a\qquad\text{at }(r,\Omega,a)=(r_0,\Omega_0,0).
\]
Here $\Phi_\Omega=\partial_\Omega\Phi_0=2r_0^2\Omega_0=2\sqrt3\,MS$ and
$\Phi_a=\Phi_1=-2\Omega_0\bigl(\frac{2M}{r_0}+\frac{\Lambda r_0^2}{3}\bigr)=-2\Omega_0\bigl(\frac23+3\Lambda M^2\bigr)$,
so
\begin{equation}\label{eq:Omega1-computation}
\Omega_1=-\frac{\Phi_a}{\Phi_\Omega}
=\frac{\frac23+3\Lambda M^2}{r_0^2}
=\frac{2+9\Lambda M^2}{27M^2}=c_Z(M).
\end{equation}

Differentiate the second equation $\partial_r\Phi_{\mathrm{geo}}(r_+(a),\Omega_+(a);M,a)=0$ at $a=0$:
\[
0=\Phi_{rr}\,r_1+\Phi_{r\Omega}\,\Omega_1+\Phi_{ra}.
\]
From $\Phi_0=-f(r)+r^2\Omega^2$ one finds
$\Phi_{rr}=-f''(r_0)+2\Omega_0^2=\frac{2}{9M^2}$ (cf.\ \eqref{eq:drrPhi-nonzero}) and
$\Phi_{r\Omega}=\partial_r(2r^2\Omega)|_{(r_0,\Omega_0)}=4r_0\Omega_0=\frac{4}{\sqrt3}S$.
Moreover $\Phi_{ra}=\partial_r\Phi_1|_{(r_0,\Omega_0)}=-2\Omega_0\,\partial_r\bigl(\frac{2M}{r}+\frac{\Lambda r^2}{3}\bigr)|_{r=r_0}
=\frac{4\sqrt3}{81M^2}S^3$.
Inserting \eqref{eq:Omega1-computation} and solving gives
\begin{equation}\label{eq:r1-computation}
r_1=-\frac{2}{\sqrt3}S.
\end{equation}

\smallskip
\noindent\emph{Second order.}
Expand both equations to second order by differentiating twice and evaluating at $a=0$.
Writing all partial derivatives of $\Phi_{\mathrm{geo}}$ at the base point $(r,\Omega,a)=(r_0,\Omega_0,0)$,
the chain rule yields
\begin{align}\label{eq:Phi-second-derivative-system}
0={}&\frac{d^2}{da^2}\Phi_{\mathrm{geo}}\bigl(r_+(a),\Omega_+(a);M,a\bigr)\Big|_{a=0}
=2\Phi_\Omega\,\Omega_2+\mathcal A,\\
0={}&\frac{d^2}{da^2}\partial_r\Phi_{\mathrm{geo}}\bigl(r_+(a),\Omega_+(a);M,a\bigr)\Big|_{a=0}
=2\Phi_{rr}\,r_2+2\Phi_{r\Omega}\,\Omega_2+\mathcal B,
\end{align}
where
\begin{align*}
\mathcal A
&:=\Phi_{rr}\,r_1^2+2\Phi_{r\Omega}\,r_1\Omega_1+\Phi_{\Omega\Omega}\,\Omega_1^2
+2\Phi_{ra}\,r_1+2\Phi_{\Omega a}\,\Omega_1+\Phi_{aa},\\
\mathcal B
&:=\Phi_{rrr}\,r_1^2+2\Phi_{rr\Omega}\,r_1\Omega_1+\Phi_{r\Omega\Omega}\,\Omega_1^2
+2\Phi_{rra}\,r_1+2\Phi_{r\Omega a}\,\Omega_1+\Phi_{raa}.
\end{align*}
In addition to $\Phi_\Omega,\Phi_{rr},\Phi_{r\Omega}$ computed above, we use the explicit derivatives
\begin{align*}
\Phi_{\Omega\Omega}&=2r_0^2=18M^2,
&\Phi_{\Omega a}&=\partial_\Omega\Phi_1=-2\Bigl(\frac23+3\Lambda M^2\Bigr),\\
\Phi_{aa}&=2\Phi_2=2\Lambda^2M^2-\frac{2\Lambda}{3}+\frac{10}{81M^2},
&\Phi_{rrr}&=\frac{4(27\Lambda M^2-1)}{81M^3},\\
\Phi_{rr\Omega}&=0,
&\Phi_{r\Omega\Omega}&=4r_0=12M,\\
\Phi_{rra}&=\frac{16\sqrt3\,\Lambda M}{S},
&\Phi_{r\Omega a}&=-4+24\Lambda M^2,\\
\Phi_{raa}&=\frac{4(-18\Lambda^2M^4+9\Lambda M^2+1)}{27M^3}.
\end{align*}
Substituting \eqref{eq:Omega1-computation}--\eqref{eq:r1-computation} into $\mathcal A$ and $\mathcal B$ gives the simplified values
\[
\mathcal A=\frac{10\Lambda}{9}-\frac{22}{81M^2},
\qquad
\mathcal B=\frac{4\,(126\Lambda M^2-5)}{243\,M^3}.
\]
The first equation in \eqref{eq:Phi-second-derivative-system} is triangular and yields
\[
\Omega_2=-\frac{\mathcal A}{2\Phi_\Omega}
=\frac{\sqrt3\,(11-45\Lambda M^2)}{486\,M^3\,S}=c_{\Omega,2}(M).
\]
Inserting this into the second equation in \eqref{eq:Phi-second-derivative-system} gives
\[
r_2=-\frac{2\Phi_{r\Omega}\,\Omega_2+\mathcal B}{2\Phi_{rr}}
=-3\Lambda M-\frac{2}{9M}.
\]
\end{proof}

\subsection{From geodesics to the principal frequency symbol}

In geometric optics, the principal symbol of the wave operator is the null Hamiltonian
$p(x,\xi)=g^{\alpha\beta}(x)\,\xi_\alpha\xi_\beta$, whose bicharacteristics project to null geodesics.
A separated oscillatory mode $e^{-i\omega t}e^{ik\varphi}$ corresponds to a covector with
$\xi_t=-\omega$ and $\xi_\varphi=k$, hence energy $E:=-\xi_t=\omega$ and angular momentum $L:=\xi_\varphi=k$.
Along a bicharacteristic the coordinate-time angular velocity is
$\Omega=d\varphi/dt=\dot\varphi/\dot t=\partial_{\xi_\varphi}p/\partial_{\xi_t}p$.
On an equatorial circular null orbit one has $E=\Omega L$ (equivalently $b=L/E=1/\Omega$), so $\omega/k=\Omega$.
Consequently, at the level of the real principal symbol one has the relation
$\Re\omega\sim k\,\Omega$ on the equatorial trapped null geodesics.
Combining this with \eqref{eq:Omega-expansion} and setting $k=\pm\ell$ yields
\begin{equation}\label{eq:principal-splitting-appendix}
\Re\omega_\pm\sim\ell\Bigl(\pm\Omega_{\mathrm{ph}}(M)+c_Z(M)\,a\pm c_{\Omega,2}(M)\,a^2\Bigr)+\mathcal O(a^3\ell),
\end{equation}

which is exactly the splitting stated in Lemma~\ref{lem:principal-splitting} for the real principal part
$F_0(M,a;n,\ell,\pm\ell)$ of the pseudopole frequency symbol.

\section{Parameter-differentiable semiclassical quantization: uniform remainder bounds}\label{app:parameter-quantization}

The goal of this appendix is to justify the \emph{parameter-uniform} remainder bounds used in Lemma~\ref{lem:C1-close},
and to explain why these bounds remain valid after taking one derivative with respect to the external parameters $\mu=(M,a)$ (and, with no change in the argument, $\mu=(M,a,\Lambda)$ on compact subextremal parameter sets).
We phrase the discussion at an abstract level compatible with Dyatlov's barrier-top quantization for Kerr--de~Sitter
\cite{Dyatlov2011QNMKerrDeSitter} and \cite{Dyatlov2012AsymptoticQNMKdS}.
The key point is that the normal-form/Grushin construction may be carried out in symbol classes that are uniform on compact
parameter sets and are stable under one parameter derivative.

\subsection{A black-box statement used in the main text}\label{subsec:quantization-blackbox}

For the reader's convenience we begin by stating the precise analytic input of this appendix in the form used in
Section~\ref{subsec:C1-close}.
Dyatlov's barrier-top quantization provides the underlying normal form and Grushin reduction; what is not always stated
explicitly in the literature, and what we track here, is that the remainders can be controlled \emph{uniformly} on compact
slow-rotation parameter sets and that these bounds persist after taking one derivative with respect to the external
parameters.

\begin{proposition}[Uniform quantization function bounds with one parameter derivative]\label{prop:quantization-blackbox}
Fix $\Lambda>0$, a compact parameter set $\mathcal K\Subset\mathcal P_\Lambda$ in the slow-rotation regime, and an overtone
index $n\in\{0,\dots,C_n\}$.
Then there exists a compact window $\Omega$ in the scaled spectral variable $\tilde\omega=h\omega$ and, for each equatorial
sign $\pm$, a scalar quantization function $\mathfrak q_\pm(\omega;\mu,\ell)$ with $\mu=(M,a)$ such that:
\begin{enumerate}
\item For each fixed $(\mu,\ell)$ the function $\mathfrak q_\pm(\cdot;\mu,\ell)$ is holomorphic in $\omega$ on the dilated
window $h^{-1}\Omega$, and the map $(\mu,\omega)\mapsto \mathfrak q_\pm(\omega;\mu,\ell)$ is $C^1$ in $\mu$ on
$\mathcal K\times h^{-1}\Omega$.
\item Zeros of $\mathfrak q_\pm(\cdot;\mu,\ell)$ in $h^{-1}\Omega$ are precisely the labeled equatorial QNMs
$\omega_{n,\ell,\pm\ell}(\mu)$ in that spectral window.
\item For every $N\in\mathbb N$ one can write an expansion
\[
\mathfrak q_\pm(\omega;\mu,\ell)=\sum_{j=0}^{N-1}h^j\,\mathfrak q_{\pm,j}(\omega;\mu)+\mathfrak r_{\pm,N}(\omega;\mu,\ell),
\]
where the coefficients $\mathfrak q_{\pm,j}(\cdot;\mu)$ are holomorphic on $h^{-1}\Omega$, $C^1$ in $\mu$, and the
remainder obeys the uniform bounds
\[
\sup_{\mu\in\mathcal K}\sup_{\omega\in h^{-1}\Omega}\Bigl(|\mathfrak r_{\pm,N}(\omega;\mu,\ell)|+|\partial_\mu\mathfrak r_{\pm,N}(\omega;\mu,\ell)|\Bigr)\le C_N\,\ell^{-N}.
\]
\item After multiplying $\mathfrak q_\pm$ by a nonvanishing factor (holomorphic in $\omega$ and $C^1$ in $\mu$), one may
arrange that $\partial_\omega\mathfrak q_{\pm,0}(\omega;\mu)\equiv 1$ on $\mathcal K\times h^{-1}\Omega$, hence
$|\partial_\omega\mathfrak q_\pm(\omega;\mu,\ell)|$ is uniformly bounded away from $0$ on slightly smaller windows for
$\ell\gg1$.
\end{enumerate}
The same statement holds with $\mu=(M,a,\Lambda)$ when allowing $\Lambda$ to vary in a compact interval and restricting to a
compact subextremal set.
\end{proposition}

\begin{proof}
This is a repackaging of the standard barrier-top Birkhoff normal form and Grushin reduction with parameter-dependent
symbol classes, as recorded below.  The normal form with parameter-uniform remainders is stated in
Proposition~\ref{prop:param-normal-form}.  The construction of the effective Hamiltonian $E_{-+}$ and the determinant
quantization function together with the uniform remainder bounds (stable under one $\mu$-derivative) is stated in
Proposition~\ref{prop:q-uniform}; the normalization of $\partial_\omega\mathfrak q_{\pm,0}$ is explained in
Remark~\ref{rem:q-normalization}.  The translation from $h$ to $\ell$ is given at the end of the appendix.
\end{proof}

\subsection{The separated Kerr--de~Sitter problem as a barrier-top model}

We briefly indicate how the abstract semiclassical setting below matches Dyatlov's concrete construction.
Recall from Section~\ref{sec:setup} that, after separation of variables, quasinormal modes are zeros of a
meromorphic Fredholm family built from the radial operator $P_r(\omega,k)$ and the angular operator
$P_\theta(\omega)$, with the separation constant $\lambda$ entering the radial equation as
\((P_r(\omega,k)+\lambda)v_r=0\); see \eqref{eq:PrPtheta} and \eqref{eq:separated-eq}.
In the high-frequency regime one introduces the semiclassical parameter
\(h=(\ell+\tfrac12)^{-1}\), the scaled spectral variable \(\tilde\omega=h\omega\), and the scaled angular
quantities \(\tilde k=h k\), \(\tilde\lambda=h^2\lambda\).
With these scalings the radial operator becomes a one-dimensional semiclassical operator whose principal
symbol is the radial Hamiltonian of the null-geodesic flow, expressed in separated variables; this is the
starting point of \cite[\S\S1.2 and 4]{Dyatlov2012AsymptoticQNMKdS}.

The relevant microlocal region is a neighborhood of the radial barrier top corresponding to the trapped set
(photon sphere). In this neighborhood the principal symbol has a unique hyperbolic critical point on the
characteristic set. Uniformity on compact parameter sets follows from two facts:
\begin{enumerate}
\item the Kerr--de~Sitter metric coefficients are real-analytic in $(M,a)$ in the subextremal regime, hence so
are the separated operators after the fixed-manifold identification in \eqref{eq:affine_identification}; and
\item on a compact slow-rotation set $\mathcal K$ the trapped set persists and remains normally hyperbolic,
so the associated hyperbolicity exponent stays bounded away from $0$ uniformly; see, for instance,
\cite{WunschZworski2011ResolventNHT} and the Kerr--de~Sitter trapping discussion in
\cite{Dyatlov2011QNMKerrDeSitter,Dyatlov2012AsymptoticQNMKdS}.
\end{enumerate}

\begin{lemma}[Uniform hyperbolicity on compact slow-rotation sets]\label{lem:uniform-hyperbolicity}
In the Kerr--de~Sitter barrier-top regime described above, let $\Omega$ be a fixed compact window for the scaled spectral
parameter $\tilde\omega$ and let $\mathcal K\Subset\mathcal P_\Lambda$ be compact.
For each $(\mu,\tilde\omega)\in\mathcal K\times\Omega$, let $\lambda(\tilde\omega,\mu)>0$ denote the (positive) hyperbolicity exponent,
i.e.\ the absolute value of the real eigenvalue of the linearization of the Hamilton vector field of the principal symbol
at the trapped critical point.
Then there exists $\lambda_0>0$ such that
\[
\lambda(\tilde\omega,\mu)\ge \lambda_0\qquad \forall (\mu,\tilde\omega)\in\mathcal K\times\Omega.
\]
\end{lemma}

\begin{proof}
For each fixed $(\mu,\tilde\omega)$ in the regime under consideration, the trapped critical point is nondegenerate hyperbolic,
so $\lambda(\tilde\omega,\mu)>0$.
The coefficient functions of the Kerr--de~Sitter metric (hence the principal symbol of the separated radial operator in the scaled variables)
depend real-analytically on $\mu=(M,a)$ on the subextremal set; therefore the critical point and the linearization eigenvalues depend
continuously on $(\mu,\tilde\omega)$ as long as the critical point persists.
On the compact set $\mathcal K\times\Omega$ the critical point persists by the slow-rotation assumption, hence $\lambda$ attains a positive
minimum $\lambda_0>0$ there.
\end{proof}

This uniform hyperbolicity is the only place where smallness of $|a|$ enters the present appendix.
Once it is available, the normal-form iteration and the Grushin reduction can be carried out with uniform
symbol estimates, and differentiating once in $(M,a)$ does not change the order of remainders.
The propositions below package precisely the estimates we need later.

\subsection{Parameter-dependent symbol classes}

Let $\mu\in\mathcal K\subset\mathbb R^2$ denote external parameters (here $\mu=(M,a)$) and let
$h\sim \ell^{-1}$ be a semiclassical parameter.
For $m\in\mathbb R$ we define the parameter-dependent symbol class $S^m_{\mu}$ on $\mathbb R^2_{(x,\xi)}$ by
\begin{equation}\label{eq:symbol-class}
\begin{aligned}
a\in S^m_{\mu}\quad\Longleftrightarrow\quad
\sup_{\mu\in\mathcal K}\sup_{(x,\xi)\in\mathbb R^2}
\langle \xi\rangle^{-m+|\beta|}\,\bigl|\partial_x^\alpha\partial_\xi^\beta\partial_\mu^\gamma a(x,\xi;\mu)\bigr|<\infty,\\
\text{for all multiindices }\alpha,\beta\text{ and all }\gamma\text{ with }|\gamma|\le 1.
\end{aligned}
\end{equation}
We write $S^{-\infty}_{\mu}:=\cap_{m\in\mathbb R}S^m_{\mu}$.

The calculus is stable in the usual way: if $a\in S^{m_1}_{\mu}$ and $b\in S^{m_2}_{\mu}$ then
$ab\in S^{m_1+m_2}_{\mu}$, and $\partial_\mu(ab)=(\partial_\mu a)b+a(\partial_\mu b)$ preserves the uniform bounds.
Likewise, standard asymptotic summation works with the uniform seminorms in \eqref{eq:symbol-class}.

\subsection{Parameter-uniform quantization and normal forms near a barrier top}

Dyatlov's semiclassical description of high-frequency Kerr--de~Sitter QNMs proceeds by reducing the
separated radial operator near the photon sphere to a one-dimensional semiclassical model with a
nondegenerate hyperbolic fixed point (a barrier top).
Abstractly, one obtains a family of semiclassical pseudodifferential operators
\begin{equation}\label{eq:Ph-mu}
P(h,\omega;\mu):=\operatorname{Op}_h(p_0(\cdot,\cdot;\omega,\mu))+h\operatorname{Op}_h(p_1(\cdot,\cdot;\omega,\mu))\ +\ \cdots,
\end{equation}
whose principal symbol $p_0(x,\xi;\omega,\mu)$ has, for $\omega$ in a fixed compact set, a unique
nondegenerate critical point $(x,\xi)=(0,0)$ on the characteristic set, with linearization having real
eigenvalues $\pm\lambda(\omega,\mu)$ bounded away from $0$ uniformly on $\mathcal K$.
All coefficients are in parameter-dependent symbol classes uniform on $\mathcal K$.

The standard barrier-top Birkhoff normal form (see, e.g., \cite[Chapters~3--4]{Zworski2012SemiclassicalAnalysis})
then provides:

\begin{proposition}[Parameter-uniform normal form]\label{prop:param-normal-form}
Assume the family \eqref{eq:Ph-mu} satisfies the barrier-top hypotheses above on a compact parameter set $\mathcal K$.
Then there exist, for $h$ small, Fourier integral operators $U(h;\mu)$, smoothly depending on $\mu$ with
uniform bounds under one $\mu$--derivative, such that
\begin{equation}\label{eq:conjugated-normal-form}
U(h;\mu)^{-1}P(h,\omega;\mu)U(h;\mu)=Q(h,\omega;\mu)+R_N(h,\omega;\mu),
\end{equation}
where $Q$ is a semiclassical model operator whose full symbol admits an expansion in $h$ with coefficients smooth in $\mu$,
and the remainder satisfies, for every $N$,
\begin{equation}\label{eq:RN-uniform}
\sup_{\mu\in\mathcal K}\bigl\|\partial_\mu R_N(h,\omega;\mu)\bigr\|_{L^2\to L^2}\le C_N h^N.
\end{equation}
\end{proposition}

\begin{proof}
We outline the normal-form construction in a way that makes the parameter dependence and the uniform remainder
bounds explicit; for complete treatments without parameters we refer to
\cite[Chapters~3--4]{Zworski2012SemiclassicalAnalysis} and the barrier-top literature.

\smallskip
\noindent\textbf{Step 1: uniform symbol control and a uniform homological operator.}
After a parameter-dependent linear symplectic change of variables (depending smoothly on $\mu$) one may assume that, at the level of the
principal symbol,
the quadratic part of $p_0$ at the critical point is
\[
p_0^{(2)}(x,\xi;\omega,\mu)=\lambda(\omega,\mu)\,x\xi,
\]
where $\lambda(\omega,\mu)>0$ is the hyperbolicity exponent.
By Lemma~\ref{lem:uniform-hyperbolicity} we have $\lambda(\omega,\mu)\ge\lambda_0>0$ uniformly on $\mathcal K$ and on the relevant compact $\omega$-window.
The associated homological operator on symbols is the Poisson bracket with $x\xi$:
\[
\mathcal L\,g:=\{x\xi,g\}.
\]
On monomials $x^\alpha\xi^\beta$ it acts by $\mathcal L(x^\alpha\xi^\beta)=(\alpha-\beta)x^\alpha\xi^\beta$.
Thus $\mathcal L$ is invertible on the complement of the resonant space $\{\alpha=\beta\}$, with inverse bounded by $|\alpha-\beta|^{-1}$.
Uniformity of the hyperbolicity exponent ensures that the passage from $x\xi$ to $\lambda(\omega,\mu)\,x\xi$ does not affect these bounds.

\smallskip
\noindent\textbf{Step 2: iterative elimination of nonresonant terms.}
Write the full symbol of $P$ (in Weyl quantization) as
\[
p(x,\xi;\omega,\mu,h)\sim \sum_{j\ge0} h^j\,p_j(x,\xi;\omega,\mu),
\qquad p_j\in S^{m-j}_\mu,
\]
with uniform seminorm bounds as in \eqref{eq:symbol-class}.
The normal-form scheme constructs generating functions $g_j(x,\xi;\omega,\mu)\in S^{2-j}_\mu$ so that conjugation by the
Fourier integral operator associated to the time-$1$ Hamiltonian flow of $g_j$ removes the nonresonant part of the $j$th coefficient.
Concretely, at each step one solves a homological equation of the form
\begin{equation}\label{eq:homological}
\mathcal L g_j = r_j^{\mathrm{nonres}},
\end{equation}
where $r_j^{\mathrm{nonres}}$ is the nonresonant part (in the sense of the grading above) of the current remainder symbol.
Because $\mathcal L$ is uniformly invertible on the nonresonant subspace and the seminorms of $r_j$ are uniformly bounded on $\mathcal K$,
the solutions $g_j$ enjoy the same type of uniform bounds.  Differentiating \eqref{eq:homological} in $\mu$ commutes with $\mathcal L$ and
preserves the uniform bounds thanks to Lemma~\ref{lem:uniform-hyperbolicity}.

\smallskip
\noindent\textbf{Step 3: quantization and the operator remainder.}
Quantizing the canonical transformations generated by the $g_j$ produces Fourier integral operators $U_j(h;\mu)$ with symbols in
parameter-dependent classes, uniformly under one $\mu$-derivative.  Iterating and composing yields $U(h;\mu)$.
Egorov's theorem in the parameter-dependent calculus shows that conjugation by $U(h;\mu)$ transforms $P$ into a model operator $Q$
whose full symbol depends only on the resonant variable $I=x\xi$ up to order $h^{N-1}$, with coefficients smooth in $\mu$.
Collecting the terms not eliminated by the first $N$ steps produces the remainder $R_N$.
Uniform symbol bounds in each step, together with the standard $L^2$ boundedness of semiclassical pseudodifferential operators, yield
\eqref{eq:RN-uniform} for $\partial_\mu R_N$ as well as for $R_N$ itself.
This completes the construction.
\end{proof}

\subsection{Grushin reduction and a quantization function}

After the normal-form reduction, one implements a Grushin problem for $Q$ (and hence for $P$) in a small neighborhood of the barrier top.
This yields a finite-dimensional effective Hamiltonian $E_{-+}(h,\omega;\mu)$ whose determinant is a scalar quantization function:
\begin{equation}\label{eq:q-as-det}
\mathfrak q(\omega;\mu,h):=\det E_{-+}(h,\omega;\mu).
\end{equation}
Zeros of $\mathfrak q$ correspond to resonances (or quasinormal modes) of the original operator in the corresponding spectral window;
see \cite[Chapter~8]{Zworski2012SemiclassicalAnalysis} for the general Grushin framework.

The key output for our inverse problem is the following uniform remainder statement.

\begin{proposition}[Uniform expansion and parameter derivative of the quantization function]\label{prop:q-uniform}
In the barrier-top setting of Proposition~\ref{prop:param-normal-form}, there exist coefficients
$\mathfrak q_j(\omega;\mu)$ smooth in $\mu$ (and holomorphic in $\omega$ in the relevant complex window) such that
for every $N\in\mathbb N$ one can write
\begin{equation}\label{eq:q-expansion}
\mathfrak q(\omega;\mu,h)=\sum_{j=0}^{N-1} h^j\,\mathfrak q_j(\omega;\mu)\ +\ \mathfrak r_N(\omega;\mu,h),
\end{equation}
where the remainder satisfies the uniform bounds
\begin{equation}\label{eq:q-remainder-uniform}
\sup_{\mu\in\mathcal K}\sup_{\omega\in\Omega}\bigl(|\mathfrak r_N(\omega;\mu,h)|+|\partial_\mu\mathfrak r_N(\omega;\mu,h)|\bigr)
\le C_N h^N,
\end{equation}
for $h$ sufficiently small, in any fixed compact subset $\Omega$ of the spectral window.
Moreover, since $\mathfrak r_N(\cdot;\mu,h)$ is holomorphic in $\omega$, Cauchy estimates imply that
for any compact $\Omega'\Subset\Omega$ there exists $C_N'>0$ such that
\begin{equation}\label{eq:q-remainder-uniform-omega}
\begin{aligned}
\sup_{\mu\in\mathcal K}\sup_{\omega\in\Omega'}\Bigl(
|\partial_\omega\mathfrak r_N(\omega;\mu,h)|+|\partial_\omega\partial_\mu\mathfrak r_N(\omega;\mu,h)|
\Bigr)
\le C_N' h^N.
\end{aligned}
\end{equation}
\end{proposition}

\begin{proof}
We explain why the Grushin reduction preserves the uniform symbol and remainder bounds, including one external parameter derivative.

\smallskip
\noindent\textbf{Step 1: a Grushin problem with parameter-uniform bounds.}
After the normal-form reduction \eqref{eq:conjugated-normal-form}, the resonance problem for $P$ in the chosen spectral window reduces
to the corresponding problem for the model operator $Q$ up to an $\mathcal O(h^N)$ perturbation in operator norm, uniformly on $\mathcal K$
and under one $\mu$-derivative.
For $Q$ one introduces auxiliary operators $R_\pm(h;\mu)$ (microlocal projections onto the model transversal data) and considers the Grushin system
\[
\mathcal P(h,\omega;\mu):=
\begin{pmatrix}
Q(h,\omega;\mu) & R_-(h;\mu)\\
R_+(h;\mu) & 0
\end{pmatrix}.
\]
For $h$ small this system is invertible on the relevant microlocal spaces, and its inverse
\[
\mathcal E(h,\omega;\mu):=\mathcal P(h,\omega;\mu)^{-1}=
\begin{pmatrix}
E(h,\omega;\mu) & E_+(h,\omega;\mu)\\
E_-(h,\omega;\mu) & E_{-+}(h,\omega;\mu)
\end{pmatrix}
\]
has entries that are semiclassical pseudodifferential operators (or finite-rank operators) with full symbol expansions.
All constructions are stable in parameter-dependent symbol classes, so each entry admits an asymptotic expansion in powers of $h$ whose
coefficients are smooth in $\mu$ with uniform seminorm bounds, and the remainders are $\mathcal O(h^N)$ uniformly under one $\mu$-derivative.

\smallskip
\noindent\textbf{Step 2: the effective Hamiltonian and its determinant.}
The matrix $E_{-+}(h,\omega;\mu)$ is finite-dimensional; its size is fixed by the chosen spectral window and does not depend on $h$.
Thus each matrix entry has a full $h$-expansion with parameter-uniform bounds, and the same is true after applying $\partial_\mu$.
Since the determinant is a polynomial in the entries, $\mathfrak q(\omega;\mu,h)=\det E_{-+}(h,\omega;\mu)$ admits an expansion
\eqref{eq:q-expansion} with coefficients $\mathfrak q_j(\omega;\mu)$ smooth in $\mu$, and the remainder satisfies \eqref{eq:q-remainder-uniform}.
Holomorphy in $\omega$ follows because the Grushin inverse depends holomorphically on $\omega$ in the window, hence so does $E_{-+}$.

\smallskip
\noindent\textbf{Step 3: $\omega$-derivative bounds for the remainder.}
For fixed $(\mu,h)$ the remainder $\mathfrak r_N(\cdot;\mu,h)$ is holomorphic in $\omega$ on $\Omega$.
Cauchy's integral formula on a contour surrounding a smaller compact subset $\Omega'\Subset\Omega$ yields
\[
\sup_{\omega\in\Omega'}|\partial_\omega \mathfrak r_N(\omega;\mu,h)|
\le C_{\Omega,\Omega'}\,\sup_{\omega\in\Omega}|\mathfrak r_N(\omega;\mu,h)|,
\]
and the same argument applies to $\partial_\omega\partial_\mu\mathfrak r_N$ since $\partial_\mu$ preserves holomorphy in $\omega$.
Combining this with \eqref{eq:q-remainder-uniform} gives \eqref{eq:q-remainder-uniform-omega}.
\end{proof}

\begin{remark}[Normalization and a uniform \texorpdfstring{$\partial_\omega\mathfrak q$}{dωq} lower bound]\label{rem:q-normalization}
The Grushin reduction determines the scalar quantization function \eqref{eq:q-as-det} only up to multiplication
by a nonvanishing factor which is holomorphic in $\omega$ and smooth in $\mu$ (this corresponds to changing
the auxiliary choices in the Grushin problem).
By such a harmless renormalization one may arrange that the leading coefficient in \eqref{eq:q-expansion}
satisfies
\begin{equation}\label{eq:q0-omega-derivative}
\partial_\omega\mathfrak q_0(\omega;\mu)\equiv 1\qquad \text{for }(\mu,\omega)\in\mathcal K\times\Omega.
\end{equation}
Then \eqref{eq:q-expansion} and \eqref{eq:q-remainder-uniform-omega} imply
\[ 
\partial_\omega\mathfrak q(\omega;\mu,h)=1+\mathcal O(h)\quad \text{uniformly on }\mathcal K\times\Omega',
\]
for any compact $\Omega'\Subset\Omega$. In particular, for $h$ small (equivalently, for $\ell$ large) one has
$|\partial_\omega\mathfrak q(\omega;\mu,h)|\ge \frac12$ uniformly on $\mathcal K\times\Omega'$, which is the
uniform nonvanishing used in Lemma~\ref{lem:C1-close}.
\end{remark}

\subsection{Translation to the \texorpdfstring{$\ell\to\infty$}{ell->infty} regime}

In the present paper the semiclassical parameter is $h=(\ell+\tfrac12)^{-1}$, so the remainder estimate
\eqref{eq:q-remainder-uniform} becomes $\mathcal O(\ell^{-N})$, uniformly on $\mathcal K$ and under one parameter derivative.
This is the precise input used in the proof of Lemma~\ref{lem:C1-close}, for the equatorial families $k=\pm\ell$.

\section{Equatorial damping rate: no linear term at \texorpdfstring{$a=0$}{a=0}}\label{app:equatorial-damping}

This appendix justifies the fact used in Lemma~\ref{lem:single-mode-expansions} that the equatorial damping rate
has no linear-in-$a$ correction at $a=0$ for fixed azimuthal direction $k=\ell$.
At the level of geometric optics, $-\Im\omega$ is governed by the (coordinate-time) Lyapunov exponent of the unstable
equatorial circular photon orbit.  We show that this Lyapunov exponent satisfies
\begin{equation}\label{eq:Lyapunov-even}
\lambda_+(M,a)=\Omega_{\mathrm{ph}}(M)+\mathcal O(a^2)\qquad\text{as }a\to0,
\end{equation}
uniformly for $M$ in compact subextremal sets (and with fixed $\Lambda>0$).
Together with the barrier-top quantization \eqref{eq:W-Lyapunov}, this yields
$\widetilde W_\ell(M,a)=\Omega_{\mathrm{ph}}(M)+\mathcal O(a^2)+\mathcal O(\ell^{-1})$.

\subsection{Equatorial null geodesics and a radial potential}

We use the standard Carter separation for null geodesics in Kerr--de~Sitter.
On the equatorial plane $\theta=\pi/2$ the Carter constant vanishes ($Q=0$), and the motion is governed by the
conserved energy $E:=-p_t$ and azimuthal momentum $L:=p_\varphi$.
Writing the impact parameter $b:=L/E$ (with $E\neq0$), the radial equation takes the form
\begin{equation}\label{eq:radial-equation}
\Bigl(\frac{dr}{d\tau}\Bigr)^2=\frac{R(r)}{r^4},
\qquad
R(r):=\bigl((r^2+a^2)-a\Xi b\bigr)^2-\Delta_r(r)\,(\Xi b-a)^2,
\end{equation}
where $\tau$ is an affine parameter and $\Delta_r$ is as in \eqref{eq:Delta_r}.
(The $\Xi$-factors in the Boyer--\allowbreak Lindquist metric satisfy $\Xi=1+\mathcal O(a^2)$ and therefore do not affect any
linear-in-$a$ computation at $a=0$.)

A circular equatorial photon orbit of radius $r=r_+(a)$ with given $(M,a)$ is characterized by the usual
double root conditions
\begin{equation}\label{eq:circular-R}
R\bigl(r_+(a)\bigr)=0,\qquad R'\bigl(r_+(a)\bigr)=0,
\end{equation}
with $b=b_+(a)$ determined simultaneously.

\subsection{The Schwarzschild--de~Sitter base point and first-order perturbation}

At $a=0$ the equations \eqref{eq:circular-R} reduce to the Schwarzschild--de~Sitter photon sphere.
One finds the radius
\begin{equation}\label{eq:r0-photon}
r_0:=r_+(0)=3M,
\end{equation}
and the impact parameter
\begin{equation}\label{eq:b0-photon}
b_0:=b_+(0)=\frac{3\sqrt3\,M}{\sqrt{1-9\Lambda M^2}}=\frac{1}{\Omega_{\mathrm{ph}}(M)}.
\end{equation}

We now expand for small $a$.  Write
\begin{equation}\label{eq:rb-expansion}
r_+(a)=r_0+a\,r_1+\mathcal O(a^2),\qquad b_+(a)=b_0+a\,b_1+\mathcal O(a^2).
\end{equation}
Substituting \eqref{eq:rb-expansion} into \eqref{eq:circular-R} and expanding to first order in $a$ yields a linear system
for $(r_1,b_1)$ whose solution is
\begin{equation}\label{eq:r1b1}
r_1=-\frac{2}{\sqrt3}\sqrt{1-9\Lambda M^2},\qquad
b_1=-\frac{2+9\Lambda M^2}{1-9\Lambda M^2}.
\end{equation}
(Equivalently, since $b_0=1/\Omega_{\mathrm{ph}}(M)$, the first-order correction to the orbital frequency
$\Omega_+(a)=1/b_+(a)$ is $\Omega_+'(0)=-(b_1/b_0^2)=c_Z(M)$, in agreement with Appendix~\ref{app:equatorial-photon}.)

\subsection{Lyapunov exponent and cancellation of the linear term}

To measure radial instability in coordinate time $t$, write $\dot t:=dt/d\tau$.
In the Kerr--de~Sitter Boyer--\allowbreak Lindquist chart one has, on the equatorial plane and at the level relevant for linearization at $a=0$,
\begin{equation}\label{eq:tdot}
\dot t=\frac{(r^2+a^2)P(r)}{\Delta_r(r)\,r^2}+\frac{a(\Xi b-a)}{r^2},
\qquad
P(r):=(r^2+a^2)-a\Xi b.
\end{equation}
Linearizing \eqref{eq:radial-equation} at a circular orbit $r=r_+(a)$ and using that $R(r_+)=R'(r_+)=0$, one obtains the standard formula
for the coordinate-time Lyapunov exponent (see, e.g., \cite{CardosoEtAl2009LyapunovQNM} for general background):
\begin{equation}\label{eq:Lyapunov-formula}
\lambda_+(M,a)^2
=
\frac{1}{2\,\dot t(r_+)^2}\,\frac{d^2}{dr^2}\Bigl(\frac{R(r)}{r^4}\Bigr)\Big|_{r=r_+(a)}
=
\frac{R''(r_+(a))}{2\,r_+(a)^4\,\dot t(r_+(a))^2}.
\end{equation}

We now expand \eqref{eq:Lyapunov-formula} at $a=0$ using \eqref{eq:rb-expansion} and \eqref{eq:r1b1}.

\paragraph{Step 1: an explicit formula for $R''$.}
Recall that $R(r)=P(r)^2-\Delta_r(r)\,(b-a)^2$ with
$P(r)=(r^2+a^2)-ab$.
Since $\Xi=1+\mathcal O(a^2)$, the $\Xi$-factors in $P(r)$ and in $(\Xi b-a)$ do not contribute at order $a$; for the linear-in-$a$ computation in this appendix we may therefore replace them by $1$.
Since $\partial_r P(r)=2r$ and $b$ is independent of $r$, we have
\begin{equation}\label{eq:Rprime-Rpp}
R'(r)=2P(r)\,\partial_r P(r)-\Delta_r'(r)\,(b-a)^2=4r\,P(r)-\Delta_r'(r)\,(b-a)^2,
\end{equation}
and hence
\begin{equation}\label{eq:Rpp-explicit}
R''(r)=4P(r)+8r^2-\Delta_r''(r)\,(b-a)^2
=12r^2+4a^2-4ab-\Delta_r''(r)\,(b-a)^2.
\end{equation}
From \eqref{eq:Delta_r} one computes
\begin{equation}\label{eq:Delta-r-pp}
\Delta_r''(r)=2-4\Lambda r^2-\frac{2\Lambda}{3}a^2.
\end{equation}

\paragraph{Step 2: first-order expansions at the Schwarzschild--de~Sitter photon sphere.}
Introduce the shorthand
\begin{equation}\label{eq:Sdef}
S=S(M):=\sqrt{1-9\Lambda M^2}>0.
\end{equation}
Then \eqref{eq:r0-photon}--\eqref{eq:b0-photon} read $r_0=3M$ and $b_0=3\sqrt3\,M/S$.
Using \eqref{eq:Rpp-explicit}--\eqref{eq:Delta-r-pp} with $a=0$, together with
$\Delta_r(r)=r^2(1-\Lambda r^2/3)-2Mr$ and $b_0^2=r_0^4/\Delta_r(r_0)=27M^2/S^2$, one obtains
\begin{equation}\label{eq:Rpp0}
R''(r_0)=12r_0^2-b_0^2\,(2-4\Lambda r_0^2)=\frac{54M^2}{S^2}.
\end{equation}

Next, substitute the first-order expansions \eqref{eq:rb-expansion} with coefficients \eqref{eq:r1b1} into
\eqref{eq:Rpp-explicit}--\eqref{eq:Delta-r-pp}.  Keeping only terms up to order $a$ gives
\begin{equation}\label{eq:Rpp-expansion}
R''\bigl(r_+(a)\bigr)=\frac{54M^2}{S^2}-\frac{12\sqrt3\,M\,(2+9\Lambda M^2)}{S^3}\,a+\mathcal O(a^2),
\qquad a\to0.
\end{equation}

\paragraph{Step 3: the coordinate-time factor $\dot t$.}
We likewise expand $\dot t(r_+(a))$ using \eqref{eq:tdot}.  Since $\Delta_r(r_0)=3M^2S^2$, we have
\begin{equation}\label{eq:tdot0}
\dot t(r_0)=\frac{r_0^2}{\Delta_r(r_0)}=\frac{3}{S^2}.
\end{equation}
Substituting \eqref{eq:rb-expansion}--\eqref{eq:r1b1} into \eqref{eq:tdot} and expanding to first order yields
\begin{equation}\label{eq:tdot-expansion}
\dot t\bigl(r_+(a)\bigr)=\frac{3}{S^2}+\frac{\sqrt3\,(2-45\Lambda M^2)}{3M\,S^3}\,a+\mathcal O(a^2),
\qquad a\to0.
\end{equation}
Consequently
\begin{equation}\label{eq:den-expansion}
\begin{aligned}
2\,r_+(a)^4\,\dot t\bigl(r_+(a)\bigr)^2
&=\frac{1458\,M^4}{S^4}-\frac{324\sqrt3\,M^3\,(2+9\Lambda M^2)}{S^5}\,a+\mathcal O(a^2),\\
&\qquad a\to0.
\end{aligned}
\end{equation}

\paragraph{Step 4: cancellation of the linear term in $\lambda_+^2$.}
Insert \eqref{eq:Rpp-expansion} and \eqref{eq:den-expansion} into \eqref{eq:Lyapunov-formula}.
Writing
\begin{equation*}
R''\bigl(r_+(a)\bigr)=R_0+aR_1+\mathcal O(a^2),\qquad
2r_+(a)^4\dot t\bigl(r_+(a)\bigr)^2=D_0+aD_1+\mathcal O(a^2),
\end{equation*}
we obtain
\begin{equation}\label{eq:ratio-expansion}
\lambda_+(M,a)^2
=\frac{R_0}{D_0}+a\Bigl(\frac{R_1}{D_0}-\frac{R_0D_1}{D_0^2}\Bigr)+\mathcal O(a^2).
\end{equation}
From \eqref{eq:Rpp-expansion} and \eqref{eq:den-expansion} we have
\begin{equation*}
\frac{R_0}{D_0}=\frac{54M^2/S^2}{1458M^4/S^4}=\frac{S^2}{27M^2}=\Omega_{\mathrm{ph}}(M)^2,
\end{equation*}
and moreover
\begin{equation*}
\frac{R_1}{D_0}=\frac{R_0D_1}{D_0^2}
=-\frac{2\sqrt3\,(2+9\Lambda M^2)}{243\,M^3}\,S.
\end{equation*}
Thus the coefficient of $a$ in \eqref{eq:ratio-expansion} vanishes, and we conclude that
\begin{equation}\label{eq:Lyapunov-series}
\lambda_+(M,a)^2=\Omega_{\mathrm{ph}}(M)^2+\mathcal O(a^2)\qquad\text{as }a\to0,
\end{equation}
uniformly for $M$ in compact subextremal sets.
Since $\lambda_+(M,0)=\Omega_{\mathrm{ph}}(M)>0$, taking square roots gives \eqref{eq:Lyapunov-even}.

\section{Second-order expansion of the equatorial Lyapunov exponent}\label{app:lyapunov}

In this appendix we compute the explicit quadratic correction coefficient $c_{\lambda,2}$ appearing in
Lemma~\ref{lem:single-mode-expansions} and in the three--parameter construction of Section~\ref{sec:three-parameter}.
We keep $\Lambda>0$ fixed throughout the computation, but every step is uniform for $(M,\Lambda)$ in compact subextremal sets.

\subsection{Correct radial potential including \texorpdfstring{$\Xi$}{Xi}}

On the equatorial plane ($\theta=\pi/2$), let $E:=-p_t$ and $L:=p_\varphi$ be the conserved quantities associated with
the Killing fields $T=\partial_t$ and $\Phi=\partial_\varphi$.
Set the (dimensionless) impact parameter $b:=L/E$ and write
\[
\Xi=\Xi(a):=1+\frac{\Lambda a^2}{3}.
\]
A direct computation from the Hamiltonian $H=\frac12 g^{\mu\nu}p_\mu p_\nu$ using the equatorial metric coefficients
\eqref{eq:eq-metric-coeffs} shows that the radial equation can be written as
\begin{equation}\label{eq:radial-equation-Xi}
\Bigl(\frac{dr}{d\tau}\Bigr)^2=\frac{R(r)}{r^4},
\qquad
R(r):=\bigl((r^2+a^2)-a\Xi b\bigr)^2-\Delta_r(r)\,(\Xi b-a)^2,
\end{equation}
where $\tau$ is an affine parameter and $\Delta_r$ is as in \eqref{eq:Delta_r}.
(The presence of $\Xi$ is invisible at linear order in $a$ but contributes at order $a^2$.)

A circular equatorial photon orbit of radius $r=r_+(a)$ is characterized by the double root conditions
\begin{equation}\label{eq:circular-R-Xi}
R\bigl(r_+(a)\bigr)=0,\qquad R'\bigl(r_+(a)\bigr)=0,
\end{equation}
with $b=b_+(a)$ determined simultaneously.
As explained in Appendix~\ref{app:equatorial-photon}, for a null orbit with tangent proportional to $T+\Omega\Phi$ one has
$-E+\Omega L=0$, hence
\begin{equation}\label{eq:b-Omega-relation}
b_+(a)=\frac{L}{E}=\frac{1}{\Omega_+(a)}.
\end{equation}

\subsection{Coordinate-time Lyapunov exponent}

Define $P(r):=(r^2+a^2)-a\Xi b$.  Solving the linear system for $(\dot t,\dot\varphi)$ in terms of $(E,L)$ gives, on the equatorial plane,
\begin{equation}\label{eq:tdot-Xi}
\dot t
=\frac{(r^2+a^2)P(r)}{\Delta_r(r)\,r^2}+\frac{a(\Xi b-a)}{r^2}.
\end{equation}
Linearizing \eqref{eq:radial-equation-Xi} at a circular orbit $r=r_+(a)$ and using \eqref{eq:circular-R-Xi}, one obtains the standard formula
for the coordinate-time Lyapunov exponent:
\begin{equation}\label{eq:Lyapunov-formula-Xi}
\lambda_+(M,a)^2
=
\frac{1}{2\,\dot t(r_+)^2}\,\frac{d^2}{dr^2}\Bigl(\frac{R(r)}{r^4}\Bigr)\Big|_{r=r_+(a)}
=
\frac{R''(r_+(a))}{2\,r_+(a)^4\,\dot t(r_+(a))^2}.
\end{equation}

Differentiating $R(r)=P(r)^2-\Delta_r(r)\,(\Xi b-a)^2$ twice in $r$ (with $b$ independent of $r$) yields the explicit identity
\begin{equation}\label{eq:Rpp-explicit-Xi}
R''(r)=12r^2+4a^2-4a\Xi b-\Delta_r''(r)\,(\Xi b-a)^2,
\qquad
\Delta_r''(r)=2-4\Lambda r^2-\frac{2\Lambda}{3}a^2.
\end{equation}

\subsection{Quadratic correction at \texorpdfstring{$a=0$}{a=0}}

Introduce the shorthand
\[
S=S(M,\Lambda):=\sqrt{1-9\Lambda M^2}>0,
\qquad
\Omega_{\mathrm{ph}}(M,\Lambda)=\frac{S}{3\sqrt{3}\,M}.
\]
By Proposition~\ref{prop:second-order-photon} (Appendix~\ref{app:equatorial-photon}) we have the expansions
\begin{equation}\label{eq:r-Omega-expansions-for-lambda}
\begin{aligned}
r_+(a)&=3M+r_1 a+r_2 a^2+\mathcal O(a^3),\\
\Omega_+(a)&=\Omega_{\mathrm{ph}}+c_Z a+c_{\Omega,2} a^2+\mathcal O(a^3),
\end{aligned}
\qquad a\to0,
\end{equation}
where $r_1,r_2,c_Z,c_{\Omega,2}$ are explicit functions of $(M,\Lambda)$ given in Appendix~\ref{app:equatorial-photon}.
In particular,
\[
r_1=-\frac{2}{\sqrt3}S,
\qquad
r_2=-\Bigl(3\Lambda M+\frac{2}{9M}\Bigr),
\qquad
c_Z=\frac{2+9\Lambda M^2}{27M^2}.
\]

By \eqref{eq:b-Omega-relation}, $b_+(a)=1/\Omega_+(a)$ admits the expansion
\begin{equation}\label{eq:b-expansion-a2}
b_+(a)=b_0+b_1 a+b_2 a^2+\mathcal O(a^3),
\qquad
\begin{aligned}
b_0&=\frac{3\sqrt3\,M}{S},\\
b_1&=-\frac{2+9\Lambda M^2}{S^2},\\
b_2&=\frac{\sqrt3\,(54\Lambda^2 M^4+39\Lambda M^2-1)}{6M\,S^3}.
\end{aligned}
\end{equation}
Multiplying by $\Xi(a)=1+\frac{\Lambda a^2}{3}$ we obtain
\begin{equation}\label{eq:beta-expansion-a2}
\Xi(a)\,b_+(a)=b_0+b_1 a+\beta_2 a^2+\mathcal O(a^3),
\qquad
\beta_2=b_2+\frac{\Lambda}{3}b_0=\frac{\sqrt3\,(45\Lambda M^2-1)}{6M\,S^3}.
\end{equation}

We now insert \eqref{eq:r-Omega-expansions-for-lambda}--\eqref{eq:beta-expansion-a2} into \eqref{eq:Lyapunov-formula-Xi}.
Write
\[
\lambda_+^2(a):=\lambda_+(M,a)^2=\frac{R''(r_+(a))}{2r_+(a)^4\,\dot t(r_+(a))^2}.
\]
We expand the numerator and denominator in \eqref{eq:Lyapunov-formula-Xi} up to order $a^2$.

\smallskip
\noindent\emph{Step 1: expansion of $R''(r_+(a))$.}
Using \eqref{eq:Rpp-explicit-Xi} together with the expansions of $r_+(a)$ and $\Xi(a)b_+(a)$, a direct Taylor
expansion gives
\begin{equation}\label{eq:Rpp-expansion-a2}
\begin{aligned}
R''\bigl(r_+(a)\bigr)
&=\frac{54M^2}{S^2}
-\frac{12\sqrt3\,M\,(2+9\Lambda M^2)}{S^3}\,a
\\
&\quad+\frac{18\Lambda M^2\,(20+9\Lambda M^2)}{S^4}\,a^2
+\mathcal O(a^3).
\end{aligned}
\end{equation}
In particular $R''(3M)=54M^2/S^2$.

\smallskip
\noindent\emph{Step 2: expansion of $2r_+(a)^4\dot t(r_+(a))^2$.}
From \eqref{eq:tdot-Xi} and the same orbit expansions one finds
\begin{equation}\label{eq:tdot-expansion-a2}
\begin{aligned}
\dot t\bigl(r_+(a)\bigr)
&=\frac{3}{S^2}
+\frac{\sqrt3\,(2-45\Lambda M^2)}{3M\,S^3}\,a
\\
&\quad+\frac{810\Lambda^2M^4-135\Lambda M^2+8}{9M^2\,S^4}\,a^2
+\mathcal O(a^3).
\end{aligned}
\end{equation}
Combining this with the expansion of $r_+(a)^4$ yields
\begin{equation}\label{eq:den-expansion-a2}
\begin{aligned}
2r_+(a)^4\,\dot t\bigl(r_+(a)\bigr)^2
&=\frac{1458M^4}{S^4}
-\frac{324\sqrt3\,M^3\,(2+9\Lambda M^2)}{S^5}\,a
\\
&\quad+\frac{54M^2\,(4+90\Lambda M^2+81\Lambda^2M^4)}{S^6}\,a^2
+\mathcal O(a^3).
\end{aligned}
\end{equation}
In particular $\dot t(3M)=3/S^2$ and $\Delta_r(3M)=3M^2S^2$.

\smallskip
\noindent\emph{Step 3: quotient expansion.}
Write $N(a)=R''(r_+(a))=N_0+N_1 a+N_2 a^2+\mathcal O(a^3)$ and
$D(a)=2r_+(a)^4\dot t(r_+(a))^2=D_0+D_1 a+D_2 a^2+\mathcal O(a^3)$, where the coefficients are read off from
\eqref{eq:Rpp-expansion-a2} and \eqref{eq:den-expansion-a2}.
Then
\begin{align*}
\lambda_+^2(a)=\frac{N(a)}{D(a)}
&=\frac{N_0}{D_0}
+a\Bigl(\frac{N_1}{D_0}-\frac{N_0D_1}{D_0^2}\Bigr)
\\
&\quad+a^2\Bigl(\frac{N_2}{D_0}-\frac{N_1D_1}{D_0^2}-\frac{N_0D_2}{D_0^2}+\frac{N_0D_1^2}{D_0^3}\Bigr)
+\mathcal O(a^3).
\end{align*}
From \eqref{eq:Rpp-expansion-a2} and \eqref{eq:den-expansion-a2} we have
\[
\begin{aligned}
N_0&=\frac{54M^2}{S^2},
&\qquad
N_1&=-\frac{12\sqrt3\,M\,(2+9\Lambda M^2)}{S^3},\\
D_0&=\frac{1458M^4}{S^4},
&\qquad
D_1&=-\frac{324\sqrt3\,M^3\,(2+9\Lambda M^2)}{S^5}.
\end{aligned}
\]
Therefore
\[
\begin{aligned}
N_1D_0&=-\frac{12\cdot 1458\,\sqrt3\,M^5\,(2+9\Lambda M^2)}{S^7},\\
N_0D_1&=-\frac{54\cdot 324\,\sqrt3\,M^5\,(2+9\Lambda M^2)}{S^7}.
\end{aligned}
\]
Since $12\cdot 1458=54\cdot 324$, we indeed have $N_1D_0=N_0D_1$ and the linear term vanishes.
Moreover $N_0/D_0=\Omega_{\mathrm{ph}}(M,\Lambda)^2$ and the quadratic term simplifies to
$\frac{2(45\Lambda M^2-2)}{729M^4}$.
Therefore
\begin{equation}\label{eq:lambda2-expansion-a2}
\lambda_+^2(a)=\Omega_{\mathrm{ph}}(M,\Lambda)^2+\frac{2\,(45\Lambda M^2-2)}{729\,M^4}\,a^2+\mathcal O(a^3),
\qquad a\to0,
\end{equation}
uniformly on compact subextremal sets.
The cancellation of the linear term is consistent with Appendix~\ref{app:equatorial-damping}.

Finally, since $\lambda_+(M,0)=\Omega_{\mathrm{ph}}(M,\Lambda)>0$, taking square roots in \eqref{eq:lambda2-expansion-a2} gives
\begin{equation}\label{eq:lambda-expansion-a2}
\begin{aligned}
\lambda_+(M,a)
&=\Omega_{\mathrm{ph}}(M,\Lambda)+c_{\lambda,2}(M,\Lambda)\,a^2+\mathcal O(a^3),
\qquad a\to0,\\
c_{\lambda,2}(M,\Lambda)
&=\frac{\sqrt3\,(45\Lambda M^2-2)}{243\,M^3\,\sqrt{1-9\Lambda M^2}}.
\end{aligned}
\end{equation}
which is \eqref{eq:clambda2} (and matches the definition of $c_{\lambda,2}(M,\Lambda)$ in Lemma~\ref{lem:Hgeo-Jacobian}).
This completes the quadratic computation.


\begin{thebibliography}{99}

\bibitem{DreyerEtAl2004BlackHoleSpectroscopy}
O.~Dreyer, B.~Kelly, B.~Krishnan, L.~S.~Finn, D.~Garrison, and R.~Lopez-Aleman.
\newblock Black-hole spectroscopy: Testing general relativity through gravitational-wave observations.
\newblock \emph{Classical and Quantum Gravity} \textbf{21}(4), 787--803 (2004).
\newblock \doi{10.1088/0264-9381/21/4/003}.

\bibitem{BertiCardosoStarinets2009QNMReview}
E.~Berti, V.~Cardoso, and A.~O.~Starinets.
\newblock Quasinormal modes of black holes and black branes.
\newblock \emph{Classical and Quantum Gravity} \textbf{26}(16), 163001 (2009).
\newblock \doi{10.1088/0264-9381/26/16/163001}.

\bibitem{CardosoEtAl2009LyapunovQNM}
V.~Cardoso, A.~S.~Miranda, E.~Berti, H.~Witek, and V.~T.~Zanchin.
\newblock Geodesic stability, Lyapunov exponents, and quasinormal modes.
\newblock \emph{Physical Review D} \textbf{79}(6), 064016 (2009).
\newblock \doi{10.1103/PhysRevD.79.064016}.

\bibitem{Teukolsky1973PerturbationsKerr}
S.~A.~Teukolsky.
\newblock Perturbations of a rotating black hole. I. Fundamental equations for gravitational,
electromagnetic, and neutrino-field perturbations.
\newblock \emph{The Astrophysical Journal} \textbf{185}, 635--647 (1973).
\newblock \doi{10.1086/152444}.

\bibitem{Whiting1989ModeStabilityKerr}
B.~F.~Whiting.
\newblock Mode stability of the Kerr black hole.
\newblock \emph{Journal of Mathematical Physics} \textbf{30}(6), 1301--1305 (1989).
\newblock \doi{10.1063/1.528308}.

\bibitem{Vasy2013MicrolocalAHKdS}
A.~Vasy.
\newblock Microlocal analysis of asymptotically hyperbolic and Kerr--de~Sitter spaces
(with an appendix by S.~Dyatlov).
\newblock \emph{Inventiones Mathematicae} \textbf{194}(2), 381--513 (2013).
\newblock \doi{10.1007/s00222-012-0446-8}.

\bibitem{Dyatlov2011QNMKerrDeSitter}
S.~Dyatlov.
\newblock Quasi-normal modes and exponential energy decay for the Kerr--de~Sitter black hole.
\newblock \emph{Communications in Mathematical Physics} \textbf{306}, 119--163 (2011).
\newblock \doi{10.1007/s00220-011-1286-x}.

\bibitem{Dyatlov2012AsymptoticQNMKdS}
S.~Dyatlov.
\newblock Asymptotic distribution of quasi-normal modes for Kerr--de~Sitter black holes.
\newblock \emph{Annales Henri Poincar\'e} \textbf{13}, 1101--1166 (2012).
\newblock \doi{10.1007/s00023-012-0159-y}.

\bibitem{SaBarretoZworski1997SphericalBH}
A.~S\'a~Barreto and M.~Zworski.
\newblock Distribution of resonances for spherical black holes.
\newblock \emph{Mathematical Research Letters} \textbf{4}, 103--121 (1997).
\newblock \doi{10.4310/MRL.1997.v4.n1.a10}.

\bibitem{PetersenVasy2023AnalyticityQNM}
O.~Petersen and A.~Vasy.
\newblock Analyticity of quasinormal modes in the Kerr and Kerr--de~Sitter spacetimes.
\newblock \emph{Communications in Mathematical Physics} \textbf{402}, 2547--2575 (2023).
\newblock \doi{10.1007/s00220-023-04776-9}.

\bibitem{PetersenVasy2025WaveKdS}
O.~Petersen and A.~Vasy.
\newblock Wave equations in the Kerr--de~Sitter spacetime: The full subextremal range.
\newblock \emph{Journal of the European Mathematical Society} \textbf{27}(8), 3497--3526 (2025).
\newblock \doi{10.4171/JEMS/1448}.

\bibitem{WunschZworski2011ResolventNHT}
J.~Wunsch and M.~Zworski.
\newblock Resolvent estimates for normally hyperbolic trapped sets.
\newblock \emph{Annales Henri Poincar\'e} \textbf{12}, 1349--1385 (2011).
\newblock \doi{10.1007/s00023-011-0108-1}.

\bibitem{HintzVasy2015NonTrappingNHT}
P.~Hintz and A.~Vasy.
\newblock Non-trapping estimates near normally hyperbolic trapping.
\newblock \emph{Mathematical Research Letters} \textbf{21}(6), 1277--1304 (2014).
\newblock \doi{10.4310/MRL.2014.v21.n6.a5}.

\bibitem{Dyatlov2016SpectralGapsNHT}
S.~Dyatlov.
\newblock Spectral gaps for normally hyperbolic trapping.
\newblock \emph{Annales de l'Institut Fourier} \textbf{66}(1), 55--82 (2016).
\newblock \doi{10.5802/aif.3005}.

\bibitem{HintzVasy2015SemilinearKdS}
P.~Hintz and A.~Vasy.
\newblock Semilinear wave equations on asymptotically de~Sitter, Kerr--de~Sitter and Minkowski spacetimes.
\newblock \emph{Analysis \& PDE} \textbf{8}(8), 1807--1890 (2015).
\newblock \doi{10.2140/apde.2015.8.1807}.

\bibitem{HintzVasy2018NonlinearStabilityKdS}
P.~Hintz and A.~Vasy.
\newblock The global non-linear stability of the Kerr--de~Sitter family of black holes.
\newblock \emph{Acta Mathematica} \textbf{220}(1), 1--206 (2018).
\newblock \doi{10.4310/ACTA.2018.v220.n1.a1}.

\bibitem{Zworski2017ScatteringResonances}
M.~Zworski.
\newblock Mathematical study of scattering resonances.
\newblock \emph{Bulletin of Mathematical Sciences} \textbf{7}, 1--85 (2017).
\newblock \doi{10.1007/s13373-017-0099-4}.

\bibitem{Zworski2012SemiclassicalAnalysis}
M.~Zworski.
\newblock \emph{Semiclassical Analysis}.
\newblock Graduate Studies in Mathematics, Vol.~138. American Mathematical Society, Providence, RI, 2012.
\newblock \doi{10.1090/gsm/138}.

\bibitem{UhlmannWang2023RecoverMassSingleQNM}
G.~Uhlmann and Y.~Wang.
\newblock Recovery of black hole mass from a single quasinormal mode.
\newblock \emph{Communications in Mathematical Physics} \textbf{401}(1), 925--936 (2023).
\newblock \doi{10.1007/s00220-023-04666-0}.

\bibitem{KatoPerturbation}
T.~Kato.
\newblock \emph{Perturbation Theory for Linear Operators}.
\newblock Classics in Mathematics. Springer, Berlin, 1995.
\newblock (Originally published 1966.)
\newblock \doi{10.1007/978-3-642-66282-9}.

\bibitem{Stucker2024QNMKerr}
T.~Stucker.
\newblock Quasinormal modes for the Kerr black hole.
\newblock Preprint, \arxiv{2407.04612} (2024).
\newblock \doi{10.48550/arXiv.2407.04612}.

\bibitem{GajicWarnick2024QNMKerr}
D.~Gajic and C.~M.~Warnick.
\newblock Quasinormal modes on Kerr spacetimes.
\newblock Preprint, \arxiv{2407.04098} (2024).
\newblock \doi{10.48550/arXiv.2407.04098}.

\bibitem{GaleNikaido1965}
D.~Gale and H.~Nikaid\^o.
\newblock The Jacobian matrix and global univalence of mappings.
\newblock \emph{Mathematische Annalen} \textbf{159}, 81--93 (1965).
\newblock \doi{10.1007/BF01360282}.

\end{thebibliography}
\end{document}